\documentclass[letterpaper,twocolumn,10pt]{article}
\usepackage{usenix-2020-09}

\usepackage{tikz}
\usepackage{amsmath}

\usepackage{filecontents}
\usepackage{graphicx} 
\usepackage{amsfonts}
\usepackage{array}
\usepackage{graphicx}
\usepackage{amsmath}
\usepackage{comment}
\usepackage{algorithm}
\usepackage{algorithmic}
\usepackage{cancel}
\usepackage{booktabs}
\usepackage{enumerate}
\usepackage{amsthm}
 \usepackage{amssymb}
\usepackage{makecell}
\usepackage[caption=false,font=normalsize,labelfont=sf,textfont=sf]{subfig}
\usepackage[export]{adjustbox}
\usepackage[T1]{fontenc}

\newtheorem{theorem}{Theorem}
\usepackage{multirow} 
\usepackage[utf8]{inputenc} 
\usepackage[T1]{fontenc}    
\usepackage{hyperref}       
\usepackage{url}            
\usepackage{booktabs}       
\usepackage{amsfonts}       
\usepackage{nicefrac}       
\usepackage{microtype}      
\usepackage{xcolor}         
\usepackage{enumitem}
\usepackage{enumerate}

\usepackage[capitalize]{cleveref}

\usepackage[utf8]{inputenc} 
\usepackage[T1]{fontenc}    
\usepackage{hyperref}       
\usepackage{url}            
\usepackage{booktabs}       
\usepackage{amsfonts}       
\usepackage{nicefrac}       
\usepackage{microtype}      
\usepackage{xcolor}         

\usepackage{caption}
\usepackage{subcaption}

\usepackage{wrapfig}
\usepackage{bm}

\newcommand{\name}{\texttt{AGNNCert}}
\newcommand{\nameE}{\texttt{AGNNCert-E}}
\newcommand{\nameN}{\texttt{AGNNCert-N}}

\begin{document}

\date{}


\title{\Large \bf 
{\name}: Defending Graph Neural Networks against Arbitrary Perturbations with Deterministic Certification}

\author{
{\rm Jiate Li$^{1}$, Binghui Wang$^{1,*}$}\\
$^1$Illinois Institute of Technology, $^*$Corresponding Author
}

\maketitle


\begin{abstract}

Graph neural networks (GNNs)
achieve the state-of-the-art on graph-relevant tasks such as node and graph classification. However, recent works show GNNs are vulnerable to adversarial perturbations include the perturbation on edges, nodes, and node features, the three components that form  a graph. 
Empirical defenses against such attacks are soon broken by adaptive ones. While certified defenses offer robustness guarantees, they face several limitations: 1) almost all restrict the adversary’s capability to only one type
of perturbation, which is impractical; 2) all are designed for a particular GNN task, which limits their applicability; and 3) the robustness guarantees of all methods except one are not 100\% accurate. 

We address all these limitations by developing {\name}, the first certified defense for GNNs against arbitrary (edge, node, and node feature) perturbations with deterministic robustness
guarantees, and applicable to the two most common node and graph classification tasks.    
{\name} also encompass existing certified defenses as special cases. 
Extensive evaluations on multiple {benchmark} node/graph classification datasets {and two real-world graph datasets}, and multiple GNNs validate the effectiveness of {\name} to provably defend against arbitrary perturbations. {\name} also shows its superiority over the state-of-the-art certified defenses against the individual edge perturbation and node perturbation.\footnote{Source code is at 
\url{https://github.com/JetRichardLee/AGNNCert}.}

\end{abstract}

\section{Introduction}
\label{sec:intro}

Graph is a natural representation for many real-world data, such as social networks, biological networks, and financial networks. In recent years, there has been a great surge of research interest on graph neural networks (GNNs)~\cite{scarselli2008graph,kipf2017semi,hamilton2017inductive,velivckovic2018graph,xu2019powerful} for representation learning on graphs, in which each node recursively aggregates representations of its neighbors to update its  representation. The learnt representations 
can be used for various graph-relevant tasks, e.g., node classification~\cite{kipf2017semi,xu2019powerful} and  graph classification~\cite{hamilton2017inductive,gilmer2017neural}.
For instance, in node classification, GNNs learn a node classifier to predict the label for each node, and learn a graph classifier to predict the label for an entire graph in graph classification.  
GNNs have achieved  outstanding performance on these tasks  for various computer security applications, including 
fraud detection (e.g., detecting fake accounts/users and fake news in social networks \cite{wang2017gang,wang2017sybilscar,gao2018sybilfuse,wang2018sybilblind,xu2022evidence}, fake reviewers and reviews in recommendation systems \cite{wang2019graph,dou2020enhancing}, fraud transactions in e-commerce systems \cite{zhang2022efraudcom}, and credit card fraud and money laundering in finance systems \cite{weber2018scalable,cheng2020graph}), 
intrusion detection \cite{zhong2024survey}, and software vulnerability detection \cite{zhou2019devign,cheng2021deepwukong,zhang2023detecting,chu2024graph}.

In GNNs, a graph is often represented as three components: nodes, their features, and edges that connect the nodes.
However, various works~\cite{zugner2018adversarial,dai2018adversarial,wu2019adversarial,wang2019attacking,xu2019topology,sun2020adversarial,zhang2021backdoor,wan2021adversarial,zhang2022projective,ma2020towards,mu2021a,wang2022bandits,wang2023turning,chenunderstanding,ju2023let,wang2024efficient} have shown that GNNs are vulnerable to \emph{test-time} adversarial attacks, where an adversary can successfully perform the attack by perturbing any  individual component or their combinations in the graph.   
Specifically, given a node/graph classifier and a graph, an attacker could inject a few nodes~\cite{sun2020adversarial,ju2023let}, slightly modify the edges~\cite{zugner2018adversarial,xu2019topology,wang2019attacking} on the graph\footnote{Edge features are typically incorporated in the edge matrix, whose perturbation can be viewed as a special case of edge perturbation.}, and/or perturb features of certain nodes~\cite{zugner2018adversarial} such that the classifier makes wrong predictions for the target node/graph.
Taking GNN based fake user detection in social networks (e.g., Twitter) as an example. In this context, nodes represent users, edges denote following-follower relationships, and node features capture user profile information. A strategic attacker (i.e., fake users) can manipulate their profiles, modify their connections with other users, and create new fake accounts and connections to evade detection \cite{wang2017gang}.

To mitigate the attacks, two lines of defenses have been proposed. 
Empirical defenses~\cite{wu2019adversarial,xu2019topology,zhu2019robust,entezari2020all,tao2021adversarial,zhao2021expressive} are developed with heuristic strategies, but were later broken by adaptive/stronger attacks~\cite{mujkanovic2022defenses}.
Certified defenses~\cite{bojchevski2020efficient,jin2020certified,jia2020certified,wang2021certified,xia2024gnncert,lai2023nodeawarebismoothingcertifiedrobustness} address the issue by offering robustness guarantees against the worst-case attack scenario. For instance, Bi-RS \cite{lai2023nodeawarebismoothingcertifiedrobustness} achieves the state-of-the-art certified defense performance against the node injection attack, while GNNCert \cite{xia2024gnncert} achieves the state-of-the-art against the edge or/and node feature perturbation attack. 
However, all existing certified defenses face several fundamental limitations shown below  (See Table~\ref{tbl:sum_CS} a comprehensive summary). 

\begin{enumerate}[leftmargin=*]
\vspace{-2mm}
\item They all restrict the adversary's capability to only 
\emph{one} type of perturbation, except \cite{xia2024gnncert} to edge and node feature perturbation. 
In practice, however, an attacker could  simultaneously manipulate nodes, node features, as well as edges to perform the best-possible attack. 
\vspace{-2mm}
\item They are designed for a particular GNN task, e.g., node classification or graph classification. This naturally limits the applicability of these defenses. 

\vspace{-2mm}
\item Their robustness guarantees are probabilistic (i.e., not 100\%), with the exception of \cite{xia2024gnncert}. This  implies the guarantees could be inaccurate with a certain probability. 
\vspace{-2mm}
\end{enumerate}

\noindent {\bf Our work:} We develop a voting-based defense, called {\name}, to address all the above limitations. 
{\name} is the  \emph{first certified defense} for GNNs on the two most common \emph{node and graph classification tasks} against \emph{arbitrary perturbations} with  \emph{deterministic} robustness guarantees.
Here, an {arbitrary perturbation} is the perturbation that can arbitrarily manipulate the 
nodes (i.e., inject new nodes and delete existing nodes), edges (i.e., inject new edges and delete existing edges), and node features on a graph.  
More specifically, {\name} can provably predict the same label for a test node/graph with arbitrary perturbation whose perturbation size (i.e., 
the total number of manipulated nodes, nodes with feature perturbations, and edges) is bounded by a threshold, which we call the \emph{certified perturbation size}.

Generally, given a graph and a GNN node/graph classifier,  our voting-based defense includes three 
steps:  
\begin{itemize}[leftmargin=*]
\vspace{-2mm}
\item  {\bf Step I: Divide a graph into multiple subgraphs.} 
We use a hash function \cite{xia2024gnncert} 
to deterministically divide the given graph into multiple subgraphs. 

\vspace{-2mm}
\item {\bf Step II: Build a voting node/graph classifier on the subgraphs.}
We use the node/graph classifier to predict the label of subgraphs, where each prediction is treated as a vote. We then count the votes for each label, and build a voting 
classifier 
that returns the label with the most votes. 

\vspace{-2mm}
\item {\bf Step III: Derive the deterministic robustness guarantee.} 
We derive the certified perturbation size for the voting node/graph 
classifier against arbitrary perturbations on the given graph with deterministic certification.  
\vspace{-2mm}
\end{itemize}

Under this setup, we first derive the sufficient condition for certified robustness against arbitrary  perturbations on GNNs---the number of different predictions on subgraphs generated from the given graph and from the arbitrarily perturbed graph should be bounded (Theorem~\ref{thm:suffcond}).  
We then propose two graph division strategies, one is \emph{edge-centric} and the other is \emph{node-centric}, to obtain the upper bounded altered predictions.

\noindent {\bf Edge-centric graph division:} This strategy is inspired by \cite{xia2024gnncert}, in which we use a hash function to 
map \emph{edges} from the given graph into multiple subgraphs \emph{such that the edges  are disjoint in any two subgraphs} ({\bf Step I}). 
With it, we show that manipulating any edge in the given graph (via edge injection or deletion) only perturbs one subgraph and hence at most one subgraph prediction is altered (Theorem~\ref{thm:edgeperturb}). 
Further, by leveraging the underlying message-passing mechanism in GNNs and with careful analysis, we prove the generated subgraphs can also bound the different subgraph predictions caused by the node manipulation (Theorem~\ref{thm:nodeperturb}) and node feature manipulation (Theorem~\ref{thm:nodefeaperturb}). 
Together, these theorems ensure the number of subgraph  predictions be altered for any node/graph 
after arbitrary perturbation is bounded (Theorem \ref{thm:edgebased}). 
Further, based on the voting 
classifier in {\bf Step II} and  Theorem~\ref{thm:suffcond}, we derive in Theorem~\ref{thm:certifyedgebased} the certified perturbation size 
({\bf Step III}).  

\noindent {\bf Node-centric graph division:} The theoretical result under edge-centric
graph division reveals the robustness guarantee  
is largely dominated by the number of edges induced by the manipulated nodes and node features, which could be ineffective in practice. For instance, injected nodes could produce many edges by linking with many nodes in the graph to exceed the certified perturbation size. 
To mitigate the issue, we propose a \emph{node-centric} graph division method. 
Our key idea is that if we can ensure all edges of a manipulated node is in a same subgraph, this subgraph is the only one being affected under every node or node feature manipulation. 
However, naive solutions are ineffective. For instance, we can map nodes into different subgraphs such that they are \emph{non-overlapped}, but it fails for node classification, as every node only appears once in all subgraphs and all target nodes for classification only receive one vote, yielding vacuous robustness. 

To address it, we innovatively treat every undirected edge as two directed edges and map each node into a subgraph index only using its outgoing edges ({\bf Step I}). In doing so, all subgraphs are directed and only contain outgoing edges of the nodes with the corresponding index.
By leveraging these \emph{directed subgraphs} and the message-passing mechanism in GNNs, we can derive the same bounded number of altered subgraph predictions against edge manipulation (Theorem~\ref{thm:edgeperturb2}) as in edge-centric graph division. Moreover, this bound 
against arbitrary node or node feature manipulation is the  number of injected/deleted  nodes (Theorem~\ref{thm:nodeperturb2}) or number of nodes whose features can be arbitrarily perturbed (Theorem~\ref{thm:nodefeaperturb2}).  
\emph{This implies the bound is robust to the manipulated node that links with many even infinite number of edges.} 
Combining them, we derive the total bounded number of altered subgraph predictions against arbitrary perturbation in Theorem \ref{thm:nodebased}, 
and the certified perturbation size in Theorem~\ref{thm:certifynodebased} 
({\bf Step III}) with the built voting classifier on the directed subgraphs ({\bf Step II}).

\noindent {\bf Evaluation:} We extensively evaluate {\name} on multiple graph datasets and multiple node and graph classifiers against arbitrary perturbations. 
We use the certified node/graph  accuracy at perturbation size $m$ 
as the evaluation metric, which means the fraction of test nodes/graphs that are provably classified as the true label against arbitrary perturbations whose perturbation size is $m$. 
Our results show that: 1) Under edge-centric graph division, {\name} can obtain about 70\% (or 60\%) certified node (or graph) accuracy when 
the perturbation size is 200 (or 10), i.e., 200 (or 10) edges induced by the edge manipulation, injected/deleted edges associated with the node manipulation, and edges associated with node feature manipulation are arbitrarily perturbed; 2) Under node-centric graph division, {\name} can obtain similar 
certified node (or graph) accuracy when the total number of 200 (or 10) edges and nodes induced by edge, node, and node feature manipulations are arbitrarily perturbed, where the manipulated nodes allow to have infinite number of edges. 

As {\name} can also defend against fewer manipulations, 
we further compare it with the state-of-the-art certified defenses of GNNs for node classification against node injection attack~\cite{lai2023nodeawarebismoothingcertifiedrobustness}, and for graph classification against node feature or/and edge manipulation~\cite{xia2024gnncert}. 
Our results show {\name} significantly outperforms \cite{lai2023nodeawarebismoothingcertifiedrobustness}  under node-centric graph division, and outperforms \cite{xia2024gnncert} 
under both graph division methods. 

We also evaluate {\name} on two real-world graph datasets (Amazon co-purchasing dataset \cite{clusterGCN} with 2M nodes and 51M edges and Big-Vul code vulnerability dataset \cite{big_vul} with 10,900 vulnerable C++ codes) to demonstrate its scalability and practicability. Our results show {\name} obtains promising robustness guarantees with an acceptable computational overhead over the undefended GNNs.

\vspace{+0.05in}
\noindent {\bf Contributions:} Our contributions are summarized below: 
\begin{itemize}[leftmargin=*]

\vspace{-2mm}
\item We develop the first certified defense to robustify GNNs for node and graph classification against arbitrary perturbations on individual graphs. 

\vspace{-2mm}
\item We propose two strategies to realize our defense that leverages the 
unique message-passing mechanism in GNNs.  

\vspace{-2mm}
\item Our robustness guarantee  
is accurate with probability 1. 

\vspace{-2mm}
\item Our defense treat existing certified defenses as special cases, as well as significantly outperforming them. 
\vspace{-4mm}
\end{itemize}

\begin{table*}[!t]
    \vspace{-2mm}
    \addtolength{\tabcolsep}{-2pt}
    \footnotesize
    \centering 
    \begin{tabular}{c|c|c|c|c|c|c|c|c|c|c|c|c|c}
     \toprule
      {\bf GNN Task} &\multicolumn{6}{c|}{\bf Node Classification}&\multicolumn{6}{c|}{\bf Graph Classification}&\multirow{2}{*}{\bf Certification Type}\\
         \cline{1-13} 
        {\bf Attack Type} & $\mathcal{E}_{\pm}$ & {\bf X}! & \multicolumn{1}{c|}{$\mathcal{V}_{\pm}$} & $\mathcal{E}_{\pm} \& {\bf X}!$ & $\mathcal{V}_{\pm} \& {\bf X}!$ & Arbitrary & $\mathcal{E}_{\pm}$ & {\bf X}! & \multicolumn{1}{c|}{$\mathcal{V}_{\pm}$} & $\mathcal{E}_{\pm} \& {\bf X}!$ & $\mathcal{V}_{\pm} \& {\bf X}!$ & Arbitrary & {}\\
         \cline{1-14} 
         {\bf RS} \cite{wang2021certified} &\checkmark &\checkmark & $\times$ &$\times$ & $\times$ &$\times$ & \checkmark &\checkmark &$\times$&$\times$&$\times$ & $\times$ &\multirow{3}{*}{\bf Probabilistic} \\
         \cline{1-13} 
         {\bf Sparsity-Aware RS} \cite{bojchevski2020efficient} &\checkmark &\checkmark & \checkmark& $\bigcirc$ &$\times$& $\times$& \checkmark &\checkmark   & \checkmark&$\times$&$\times$&$\times$&{}\\
         \cline{1-13} 
         {\bf Node-Aware Bi-RS}\cite{lai2023nodeawarebismoothingcertifiedrobustness}&$\times$ &$\times$ & \checkmark & $\times$&$\times$&$\times$&$\times$ &$\times$  & $\times$ & $\times$&$\times$&$\times$&{}\\
         \cline{1-14} 
         {\bf GNNCert} \cite{xia2024gnncert}&$\bigcirc$ &$\bigcirc$ & $\times$&$\bigcirc$ &$\times$ &$\times$&\checkmark&\checkmark & $\times$  & \checkmark &$\times$&$\times$&\multirow{3}{*}{\bf Deterministic}\\
         \cline{1-13} 
         {\bf {\nameE}} (Ours)&\checkmark &\checkmark& \checkmark&\checkmark &\checkmark &\checkmark&\checkmark &\checkmark&\checkmark&\checkmark&\checkmark & \checkmark \\
         \cline{1-13} 
         {\bf {\nameN}} (Ours)&\checkmark &\checkmark & \checkmark&\checkmark &\checkmark &\checkmark&\checkmark &\checkmark&\checkmark&\checkmark&\checkmark  & \checkmark \\
       \bottomrule
    \end{tabular}
    \caption{Summarizing the existing certified defenses of GNN against adversarial perturbations and their capability against different types of manipulations. $\mathcal{E}_{\pm}$, $\mathcal{V}_{\pm}$, and ${\bf X}!$ represent the edge manipulation (injection/deletion), node manipulation (injection/deletion), and node feature perturbation, respectively. $\checkmark$ means the defense is able to defend the respective attack, $\bigcirc$ means the defense could be adapted to defend the attack, and $\times$ means not able to.}
    \label{tbl:sum_CS}
    \vspace{-4mm}
\end{table*}

\section{Background and Problem Definition}
\label{sec:background}

 \subsection{\bf Graph Neural Network (GNN)} 

Let a graph be $G=\{\mathcal{V},\mathcal{E},{\bf X}\}$, which consists of the nodes $\mathcal{V}$, node features  ${\bf X}$, and edges $\mathcal{E}$. We denote $u\in \mathcal{V}$ as a node, $e=(u,v) \in \mathcal{E}$ as an edge, and ${\bf X}_u$ as node $u$'s feature. 

GNNs learn representations for graph data 
by following the \textit{message passing} strategy with two operations, i.e.,  the {aggregate} operation \texttt{Agg} and {combine} operation \texttt{Comb}.
Specifically, \texttt{Agg} iteratively aggregates the representations of all neighbors of a node, while \texttt{Comb} updates the node’s representation by combining it with the aggregated neighbors’ representations. The two operations are formally defined below:

{
\vspace{-4mm}
\small
\begin{align}
\label{aggregate}
\bm{l}_v^{(k)} =  \texttt{Agg} \big(\big\{ \bm{h}_u^{(k-1)}: u \in \mathcal{N}(v) \big\} \big), \, \bm{h}_v^{(k)} = \texttt{Comb}\big(\bm{h}_v^{(k-1)},  \bm{l}_v^{(k)} \big),
\end{align}
}%
where $\bm{h}_v^{(k)}$ denotes node $v$'s representation in the $k$-th layer and $\bm{h}_v^{(0)}={\bf X}_v$. $\mathcal{N}(v)$ denotes the neighbors of $v$. 

Different GNNs use different aggregate and combine operations. For example, in Graph Convolutional Network (GCN)~\cite{kipf2017semi}, the two operations are integrated as follows:

{
\vspace{-4mm}
\small
 \begin{align}
 \label{aggregate_gcn}
 \bm{h}_v^{(k)} & = \texttt{ReLU}\big( \bm{W}^{(k)} \cdot \texttt{Mean}  \big \{  \bm{h}_{u}^{(k-1)}: u \in \mathcal{N}(v) \bigcup \bm{h}_v^{(k-1)} \big \} \big), 
 \end{align}
 }%
where the element-wise mean pooling function \texttt{Mean} acts as the aggregate operation and $\texttt{ReLU}$ the combine operation. $\Theta = \{\bm{W}^{(1)}, \cdots, \bm{W}^{(K)} \} $ are all the learned parameters.

A node $v$’s final representation $\bm{h}_v^{(K)}$ captures structural information within $v$’s $K$-hop neighbors, which are used for many tasks. In this paper we focus on the two classic classification tasks on graphs: node classification and graph classification. We  denote $f$ as the GNN node or graph classifier and $\mathcal{Y}$ as the set of all labels.  

\vspace{+0.05in}
 \noindent {\bf Node classification:} $f$ takes a graph $G$ as input and predicts each node $v \in G$ a label $y_v \in \mathcal{Y}$ based on $v$'s learnt representation $\bm{h}_v^{(K)}$. That is, $y_v = f(G)_v = \texttt{softmax}(\bm{h}_v^{(K)})$. 

\vspace{+0.05in}
 \noindent {\bf Graph classification:} $f$ takes a graph $G$ as input and predicts a label $y_G \in \mathcal{Y}$ for the whole graph $G$ by using all nodes' representations $\{\bm{h}_v^{(K)}\}_{v\in G}$. For instance, when averaging all nodes' final representations, we have $y_G = f(G) = \texttt{softmax}(\texttt{Avg}(\{\bm{h}_v^{(K)}\}_{v \in G}))$.

\subsection{Adversarial Attacks on GNNs}

In adversarial attacks against GNNs, an attacker can manipulate a graph $G=\{\mathcal{V},\mathcal{E},{\bf X}\}$ into a perturbed one $G' = \{\mathcal{V}',\mathcal{E}',{\bf X}'\}$, where $\mathcal{V}'$, $\mathcal{E}'$, ${\bf X}'$ are the perturbed version of $\mathcal{V}$, $\mathcal{E}$, and ${\bf X}$, respectively. 

\vspace{+0.05in}
\noindent {\bf Edge manipulation:} The attacker can 1) \emph{inject new edges} $\mathcal{E}_+$, and 2) \emph{delete existing edges}, denoted as $\mathcal{E}_{-}\subset \mathcal{E}$ from $G$.

\vspace{+0.05in}
\noindent {\bf Node manipulation:}
The attacker perturbs $G$ by (1) \emph{injecting new nodes} $\mathcal{V}_+$,  whose node feature denoted as ${\bf X}'_{\mathcal{V}_+}$ can be arbitrary, together with the arbitrarily injected new edges $\mathcal{E}_{\mathcal{V}_+} \subseteq \{(u,v) \notin \mathcal{E}, \forall u \in \mathcal{V}_+ \vee v \in \mathcal{V}_+ \}$ induced by 
$\mathcal{V}_+$; and (2) \emph{deleting existing nodes} $\mathcal{V}_- \subset \mathcal{V}$. When $\mathcal{V}_-$ are deleted, their features  ${\bf X}_{\mathcal{V}_-} \subseteq {\bf X}$ and all connected edges  $\mathcal{E}_{\mathcal{V}_-} = \{(u,v) \in \mathcal{E}, \forall u \in \mathcal{V}_- \vee v \in \mathcal{V}_- \}$ are also removed.

\vspace{+0.05in}
\noindent {\bf Node feature manipulation:} 
The attacker arbitrarily manipulates features ${\bf X}_{\mathcal{V}_r}$ of a set of representative nodes $\mathcal{V}_{r}$ to be ${\bf X}'_{\mathcal{V}_r}$. 
We also denote the edges connected with nodes $\mathcal{V}_{r}$  
as $\mathcal{E}_{\mathcal{V}_r}=\{(u,v) \in \mathcal{E}: \forall u \in\mathcal{V}_{r} \vee v \in\mathcal{V}_{r} \}$.

\vspace{+0.05in}
\noindent {\bf Arbitrary manipulation:} The attacker can manipulate the graph $G$ with an 
arbitrary combined perturbations on 
edges, nodes, and node features. 

\emph{For description simplicity, we will use $\{\mathcal{E}_+, \mathcal{E}_-\}$ to indicate the edge manipulation with arbitrary injected edges $\mathcal{E}_+$  and deleted edges $\mathcal{E}_-$ on $G$. Similarly, we will use $\{\mathcal{V}_+, \mathcal{E}_{\mathcal{V}_+}, {\bf X}'_{\mathcal{V}_+}, \mathcal{V}_-, \mathcal{E}_{\mathcal{V}_-}\}$ to indicate the node manipulation, and $\{\mathcal{V}_r, \mathcal{E}_{\mathcal{V}_r},{\bf X}'_{\mathcal{V}_r}\}$ the node feature manipulation. Any combination of the manipulations is inherently well-defined.}

\subsection{Voting based Certified Defense}
\label{sec:GNNCert}

Voting-based GNNCert~\cite{xia2024gnncert}
has achieved state-of-the-art certified defense performance against 
node feature and edge manipulation. Here we review~\cite{xia2024gnncert} since  
our defense is also based on voting.   
GNNCert is only applicable for graph classification and consists of three steps.  

\vspace{+0.05in}
\noindent {\bf Step I:  divide a graph into multiple subgraphs.} Given a graph $G=\{\mathcal{V},\mathcal{E},{\bf X}\}$, and a graph classifier $f$.  GNNCert uses a hash function $h$ (e.g., MD5) to generate the subgraphs for $G$. A hash function takes a bit string as input and outputs an integer (e.g., within a range $[0,2^{128}-1]$). It uses the string of edge or node index as the input to the hash function. 
For instance, for a node $u$,  its string is denoted as $\texttt{str}(u)$, while for an edge $e=(u,v)$, its string is $\texttt{str}(u)+\texttt{str}(v)$, where ``+" means string concatenation, and $\texttt{str}$ turns the node index into a string and adds “0” prefix to align it into a fixed length.

To defend against edge manipulation, it uses 
$h$ to map each edge into a subgraph index.
Assuming $T_e$ subgraphs 
in total, the subgraph index $i_e$ of every edge $e=(u,v)$ is defined as\footnote{In the undirected graph, we put the node with a smaller index (say  $u$) first and let 
$h[\mathrm{str}(v) + \mathrm{str}(u)]=h[\mathrm{str}(u) + \mathrm{str}(v)]$.} 

{
\vspace{-4mm}
\begin{align}
\label{eqn:edgehash}
i_e = h[\mathrm{str}(u) + \mathrm{str}(v)] \, \, \texttt{mod} \, \, T_e+1, 
\end{align}
}%
where $\texttt{mod}$ is the module function. Denoting $\mathcal{E}^i$ as the set of edges whose subgraph index is $i$, i.e., $\mathcal{E}^i = \{\forall e \in \mathcal{E}: i_e= i \}$,  $T_e$ subgraphs for $G$ can be built as $\mathcal{G}^e_T = \{ {G}_i = (\mathcal{V},\mathcal{E}^i,{\bf X}): i=1,2,\cdots, T_e\}$, where edges in different subgraphs are disjoint, i.e., $\mathcal{E}^i \cap \mathcal{E}^j =  \emptyset, \forall i,j \in \{1, \cdots, T_e\}, i \neq j$. 

To defend against node feature manipulation, it 
 uses $h$ to map each node into a subgraph index. 
Assuming $T_n$ subgraphs in total, the subgraph index $i_u$ of every node $u$ is

{
\vspace{-4mm}
\begin{align}
\label{eqn:nodehash}
i_u = h[\mathrm{str}(u)] \, \, \texttt{mod} \, \, T_n+1, 
\end{align} 
}%
It then uses ${\bf X}^i$ to denote the features of nodes whose subgraph index is $i$. 
Then, $T_n$ subgraphs  can be built as: $\mathcal{G}^n_T = \{ {G}_i = (\mathcal{V},\mathcal{E},{\bf X}^i): i=1,2,\cdots, T_n\}$,

To defend against both manipulations, it then constructs a total of $T = T_e \cdot T_n$ subgraphs $\mathcal{G}_T = \{G_t = (\mathcal{V},\mathcal{E}^i,{\bf X}^j), t = 1, \cdots, T_e \cdot T_n, i=\lceil t/T_n\rceil, j=t-(i-1)\cdot T_n\}$.

\noindent {\bf Step II: build a voting graph classifier on all subgraphs.} 
GNNCert applies the graph classifier $f$  to make predictions on all $T$ subgraphs, and count the vote $c_y$ for every class $y \in \mathcal{Y}$. 

{
\vspace{-4mm}
\begin{align}
c_{y_G} = \sum\nolimits_{i=1}^{T}\mathbb{I}(f(G_{i})=y_G), \forall y_G \in \mathcal{Y} \label{eqn:vote_GC} 
\end{align}
}%

It then defines a \emph{voting graph classifier} $\bar{f}$ as returning the class with the most vote:

{
\vspace{-4mm}
\begin{align}
\bar{f}(G) = \underset{y_G \in \mathcal{Y}}{\arg\max} 
\,  c_{y_G} \label{eqn:vc_GC} 
\vspace{-2mm}
\end{align}
}%

\noindent {\bf Step III: derive the deterministic robustness guarantee for the voting graph classifier.}
GNNCert guarantees that $T_n$ (or $T_e$) subgraphs are corrupted when an attacker injects or deletes an \emph{arbitrary} edge (or \emph{arbitrarily} perturb the features of a node). Then, GNNCert shows the voting classifier $\bar{f}$ tolerates up to $\lfloor M^f/T_e \rfloor$ perturbed edges {\bf \emph{OR}} $\lfloor M^f/T_n \rfloor$ of nodes with adversarially perturbed features, where $M^f \in [0,T_n\cdot T_e/2]$ is a constant depending on the number of votes of $f$'s output.

\vspace{+0.05in}
\noindent {\bf Limitations of GNNCert:} 1) It only derives the robustness guarantee against edge manipulation \emph{OR} node feature manipulation. 
{Under a very special case when $T_n = T_e = T$, we can derive its robustness against both edge \emph{AND} node feature manipulation, where the certified perturbation size is $\lfloor M^f/T \rfloor$.} However, its performance is worse than ours (See Figure~\ref{fig:ours-vs-gnncert-edgenode}). 
2) It is only applicable for graph classification. 3) It cannot defend against the well-known node injection attack.

\subsection{Problem Statement}
\label{sec:problem}

{\bf Threat model:} Given a GNN node/graph classifier $f$ and a graph $G$, the adversary can \emph{arbitrarily} manipulate a number of the edges, nodes, and node features in $G$ such that $f$ misclassifies target graphs in graph classification or target nodes in node classification. 
{For instance, when a social network platform deploys a GNN detector to detect fake users (the adversary) \cite{wang2017gang,xu2022evidence}, the fake users is motivated to evade them \cite{wang2019attacking}: they can modify their profiles,
their connections with some users, and create new fake accounts and connections to bypass detection.}
Since we focus on certified defenses, we consider the strongest attack where the adversary has white-box access to $G$ and $f$, i.e., it knows all the edges, nodes, and node features in $G$, and all the model parameters about $f$.

\vspace{+0.05in}
\noindent {\bf Defense goal:}
We aim to build a certifiably robust GNN that: 
\begin{itemize}[leftmargin=*]
\vspace{-2mm}
\item  has a deterministic robustness guarantee; 
\vspace{-2mm}
\item is suitable for both node and graph classification tasks; 
\vspace{-2mm}
\item provably predicts the same label against the arbitrary perturbation when the \emph{perturbation size}, i.e., the total number of manipulated nodes, nodes with feature perturbation, and edges, is bounded by a threshold, which we call the \emph{certified perturbation size}. 
\vspace{-2mm}
\end{itemize}

Our ultimate goal is to obtain the largest-possible certified perturbation size that satisfies all the above conditions.

\section{Our Voting-based Defense: {\name}}
\label{sec:defense}

In this section we introduce our voting-based certified defense {\name} for GNNs against arbitrary perturbations. We first give an overview of {\name} in Section~\ref{Sec:overview}, which consists of three critical steps, e.g., the first step is to divide a graph into multiple subgraphs with disjoint edges. We then 
design two distinct graph division strategies (one is edge-centric in Section~\ref{Sec:edgebased} inspired by \cite{xia2024gnncert} and the other is node-centric in Section~\ref{Sec:nodebased} by further enhancing the robustness guarantee).  
Within each strategy, we derive our deterministic certified robustness results, which can treat existing defenses as special cases. 
Figure~\ref{fig:overview} briefly illustrates our {\name}.

\begin{figure*}[t]
    \centering
    \captionsetup[subfloat]{labelsep=none, format=plain, labelformat=empty}

    \includegraphics[width=0.98\linewidth]{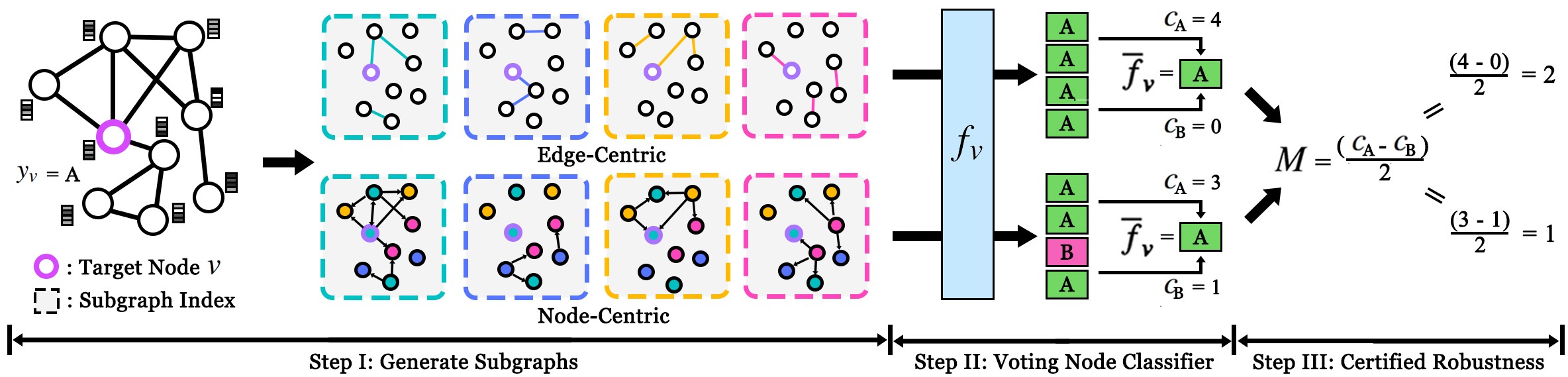}
    \vspace{-2mm}
    \caption{Overview of our {\name} (use node classification for illustration), which consists of three steps. Assume we are given an input graph $G$, a GNN node classifier $f$, and a target node $v$ with label $y_v$ for classification. {\bf Step I:}  it divides $G$ into a set of (e.g., 4) subgraphs via the proposed \emph{Edge-Centric Graph Division} (Section~\ref{Sec:edgebased}) or \emph{Node-Centric Graph Division} (Section~\ref{Sec:nodebased}) strategy. 
    {\bf Step II:} it builds a voting node classifier $\bar{f}$ based on all the subgraphs. Specifically, the target node's predicted class  by $f$ on all subgraphs are treated as votes, and $\bar{f}$ returns the class with the most vote as the final prediction. {\bf Step III:} it derives the certified perturbation size $M$ for $\bar{f}$  against arbitrary perturbations with a deterministic (100\%) guarantee. 
    }
    \label{fig:overview}
    \vspace{-4mm}
\end{figure*}

\vspace{-2mm}
\subsection{Overview}
\label{Sec:overview}

Given a graph $G=\{\mathcal{V},\mathcal{E},{\bf X}\}$, 
a GNN node/edge classifier $f$, the set of classes $\mathcal{Y}$, and a target node $v\in\mathcal{V}$ if the task is node classification. 
At a high level, our defense framework is similar to \cite{xia2024gnncert} that consists of three steps below:  

\vspace{+0.05in}
\noindent {\bf Step I:  divide the graph into multiple subgraphs.} We divide $G$ into 
a set of $T$ subgraphs $\mathcal{G}_T=\{G_{1},G_{2},\dots,G_{T}\}$ via a hash function and ensure  
edges in different subgraphs are \emph{disjoint}.

\vspace{+0.05in}
\noindent {\bf Step II: build a voting-based node/graph classifier:} We apply the GNN classifier $f$  to make predictions on all the $T$ subgraphs, and count the vote $c_y$ for every class $y$ in $\mathcal{Y}$. 

{
\vspace{-4mm}
\begin{align}
& \textbf{Node classifier: } c_{y_v} = \sum\nolimits_{i=1}^{T}\mathbb{I}(f(G_{i})_v=y_v), \forall y_v \in \mathcal{Y} \label{eqn:vote_NC} \\
& \textbf{Graph clasifier: } c_{y_G} = \sum\nolimits_{i=1}^{T}\mathbb{I}(f(G_{i})=y_G), \forall y_G \in \mathcal{Y} \label{eqn:vote_GC} 
\end{align}
}%

We then define our \emph{voting node/graph classifier} $\bar{f}$ as returning the class with the most vote:

{
\vspace{-4mm}
\begin{align}
& \textbf{Voting node classifier: } \bar{f}(G)_v = \underset{y_v \in \mathcal{Y}}{\arg\max} c_{y_v} \label{eqn:vc_NC} \\
& \textbf{Voting graph classifier: } \bar{f}(G) = \underset{y_G \in \mathcal{Y}}{\arg\max} c_{y_G} \label{eqn:vc_GC} 
\end{align}
}%

\vspace{-2mm}
\noindent {\bf Step III: derive the deterministic robustness guarantee.} 
We denote $y_a$ and $y_b$  as the class with the most vote $c_{y_a}$ and 
the second-most vote $c_{y_b}$, respectively.   
We pick the class with a smaller index if ties exist.  
Denote $G'$ as the perturbed graph of $G$ under arbitrary perturbation, and  $\mathcal{G}_T'=\{G'_{1},G'_{2},\dots,G'_{T}\}$ be the set of $T$ subgraphs generated for $G'$ under the same graph division strategy.  
Then we have the below condition for certified robustness against arbitrary attacks on GNNs.

\begin{theorem}[Sufficient Condition for Certified Robustness]
\label{thm:suffcond}
\vspace{-1mm}
Let $y_a, y_b, c_{y_a}, c_{y_b}$ be defined above in node classification or graph classification, and let $M = {\lfloor c_{y_a}-c_{y_b}-\mathbb{I}(y_{a}>y_{b})\rfloor} / {2}$. The voting classifier $\bar{f}$ guarantees the same prediction on both $G'$ and $G$ for the target node $v$ in node classification or the target graph $G$ in graph classification, if the number of subgraphs' predictions on $\{G_i\}$'s and $\{G'_i\}$' that are different under the arbitrary perturbation is bounded by $M$. I.e., 

{
\vspace{-4mm}
\small
\begin{align}
    & \forall G': \sum\nolimits_{i=1}^{T}\mathbb{I}(f(G_{i})_v\neq f(G'_{i})_v) \leq M \implies \bar{f}(G)_v = \bar{f}(G')_v \label{eqn:suff-NC} \\ 
    & \forall G': \sum\nolimits_{i=1}^{T}\mathbb{I}(f(G_{i}) \neq f(G'_{i})) \leq M  \implies \bar{f}(G) = \bar{f}(G')
    \label{eqn:suff-GC}
\end{align}
}
\end{theorem}

\noindent \emph{Proof.} See Appendix \ref{app:suffcond}.

\vspace{+0.05in}
The above theorem motivates us to design the graph division method such that: 1) the number of different subgraph predictions on $\mathcal{G}_T$ and $\mathcal{G}'_T$ can be upper bounded (and the smaller the better).  
2) the difference between the most vote $c_{y_a}$ and second-most vote  $c_{y_b}$ is as large as possible, in order to ensure larger certified perturbation size.  

Next, we introduce our two graph division methods.  
Figure~\ref{fig:subgraphs} visualizes the divided subgraphs of the two methods without and with the adversarial manipulation.

\begin{figure*}[t]
    \centering
    \captionsetup[subfloat]{labelsep=none, format=plain, labelformat=empty}

    \subfloat[{(a) Edge-Centric Graph Division against edge injection and node injection attacks}]{
     \includegraphics[width=0.9\linewidth]{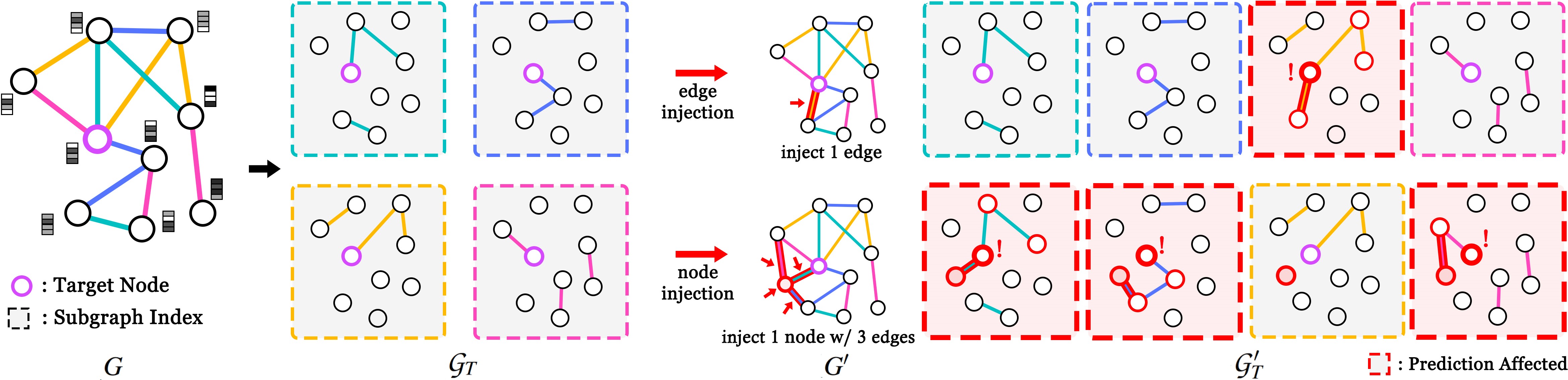}}
   
    \hspace{+20mm}
    \subfloat[{(b) Node-Centric Graph Division against edge injection and node injection attacks}]{
    \includegraphics[width=0.9\linewidth]{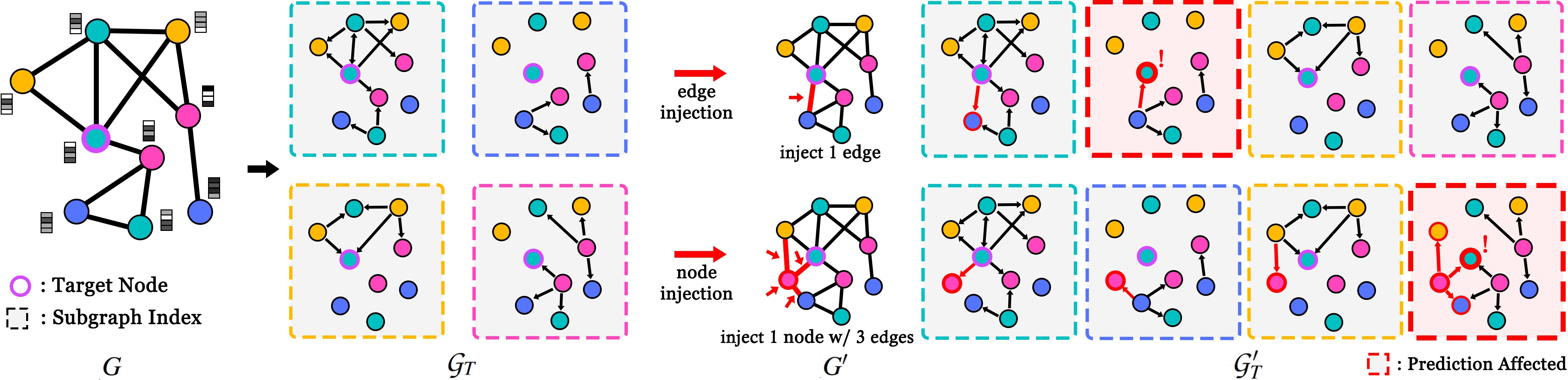}}
    
    \vspace{-2mm}
    \caption{Illustration of our edge-centric and node-centric graph division strategies for node classification. 
    We use edge injection and node injection attacks to show the bounded number of altered predictions on the generated subgraphs after the attack. {\bf To summarize:} 1 injected edge  affects at most 1 subgraph prediction in both graph division strategies. In contrast, 1 injected node with, e.g., $3$ injected edges can affect (at most) 3 subgraph predictions with edge-centric graph division, but at most 1 subgraph prediction with node-centric graph division. Figures~\ref{fig:subgraphs_NC_more}-\ref{fig:subgraphs_GC} 
    in Appendix 
    also show other attacks and on graph classification.    
    }
    \label{fig:subgraphs}
    \vspace{-4mm}
\end{figure*}

\label{Sec:Subgraph}

\subsection{Edge-Centric Graph Division}
\label{Sec:edgebased}

Our first graph division method is edge-centric inspired by~\cite{xia2024gnncert}. 
The idea is 
to divide \emph{edges} in a graph into different subgraphs, such that each edge is deterministically mapped into \emph{only one subgraph}. 
With this strategy, we can bound the number of altered predictions on these subgraphs before and after the arbitrary perturbation (Theorem~\ref{thm:edgebased}), which facilitates deriving the certified perturbation size (Theorem~\ref{thm:certifyedgebased}).
Next, we show our edge-centric graph division method in detail.

\noindent {\bf Generating edge-centric subgraphs:} 
We follow \cite{xia2024gnncert} to 
use the hash function to map edges as shown in Equation \ref{eqn:edgehash}. 
We build $T$ subgraphs for $G$ as $\mathcal{G}_T = \{ {G}_i = (\mathcal{V},{\bf X}, \mathcal{E}^i): i=1,2,\cdots, T\}$, where 
$\mathcal{E}^i \cap \mathcal{E}^j =  \emptyset, \forall i,j \in \{1, \cdots, T\}, i \neq j$.

Recall that \cite{xia2024gnncert} maps both edges and node features to generate two sets of subgraphs to defend against node feature and edge manipulations. Instead, our method only needs to map edges into a set of subgraphs, which is not only efficient, but also obtains much defense performance. 

\noindent {\bf Bounding the number of different subgraph predictions:} 
For a perturbed graph $G'$, we use the same graph division strategy to generate a set of $T$ subgraphs $\mathcal{G}'_T = \{G'_1, G'_2, \cdots, G'_T\}$. Then, we can upper bound the number of different subgraph predictions on $\mathcal{G}_T$ and $\mathcal{G}'_T$ against any individual perturbation. 

\begin{theorem}[]
\label{thm:edgeperturb}
Assume a graph $G$ is under the edge manipulation $\{\mathcal{E}_+,\mathcal{E}_-\}$, 
then at most $|\mathcal{E}_+| + |\mathcal{E}_-|$ subgraphs generated by our edge-centric graph division have different predictions between $\mathcal{G}'_T$ and $\mathcal{G}_T$. 
\end{theorem}

\begin{proof}
Edges in all subgraphs of $\mathcal{G}_T$ are disjoint. 
Hence, when any edge in $G$ is deleted or added by an adversary, only one subgraph from $\mathcal{G}_T$ is affected. Further, when any $|\mathcal{E}_+| + |\mathcal{E}_-|$ edges in $G$ are perturbed, there are at most $|\mathcal{E}_+| + |\mathcal{E}_-|$ subgraphs between $\mathcal{G}_T$ and $\mathcal{G}'_T$ are different.
By applying the node/graph classifier on $\mathcal{G}_T$ and $\mathcal{G}'_T$, there are at most $|\mathcal{E}_+| + |\mathcal{E}_-|$ predictions that are different between them.  
\end{proof}

 Unlike edge manipulation, both node and node feature manipulations involve all components (i.e., edges, nodes, and node features) in the graph. At first glance, it seems hard to bound the alter subgraph predictions in this case. After careful analysis, we observe 
 the underlying  message-passing mechanism in GNNs (Section~\ref{sec:background}) still facilitates us to obtain the upper bound shown below. 

\begin{theorem}[]
\label{thm:nodeperturb} 
Assume a graph $G$ is under the node manipulation  
$\{\mathcal{V}_+, \mathcal{E}_{\mathcal{V}_+},{\bf X}'_{\mathcal{V}_+},\mathcal{V}_-, \mathcal{E}_{\mathcal{V}_-}\}$,
then at most $|\mathcal{E}_{\mathcal{V}_+}| + | \mathcal{E}_{\mathcal{V}_-}| $ subgraphs generated by our edge-centric graph division have different predictions between $\mathcal{G}'_T$ and $\mathcal{G}_T$. 
\end{theorem}

\begin{theorem}[]
\label{thm:nodefeaperturb} 
Assume a graph $G$ is under the node feature manipulation 
$\{\mathcal{V}_r, \mathcal{E}_{\mathcal{V}_r},{\bf X}'_{\mathcal{V}_r}\}$, 
then at most $|\mathcal{E}_{\mathcal{V}_r}|$ subgraphs generated by our edge-centric graph division have different predictions between $\mathcal{G}'_T$ and $\mathcal{G}_T$. 
\end{theorem}

\begin{proof}
Our proof for the above two theorems is based on the key observation that manipulations on isolated nodes have no influence on other nodes' representations in GNNs. 
Take node injection for instance and the proof for other cases are similar.  
Note that all subgraphs after node injection will contain the newly injected nodes, but they still do not have overlapped edges between each other via the hash mapping. Hence, the edges $E_{\mathcal{V}_+}$ induced by the injected nodes $\mathcal{V}_+$ exist in at most $|E_{\mathcal{V}_+}|$ subgraphs. In other word, the injected nodes $\mathcal{V}_+$ in at least $T-|E_{+}|$ subgraphs have no  edges and are isolated. 

Due to the message passing mechanism in GNNs, every node only uses its neighboring nodes' representations to update its own representation. Hence, 
the isolated injected nodes, whatever their features ${\bf X}'_{\mathcal{V}_+}$ are, would have no influence on other nodes' representations, implying at least $T-|E_{+}|$ subgraphs' predictions maintain the same. 
\vspace{-2mm}
\end{proof}

With above theorems, we can bound the total number of different subgraph predictions with \emph{arbitrary perturbation}.

\begin{theorem}[Bounded Number of Edge-Centric Subgraphs with Altered Predictions under Arbitrary Perturbation]
\label{thm:edgebased} 
Given any GNN node/graph classifier $f$, a graph $G$,  
and $T$ edge-centric subgraphs $\mathcal{G}_T$ for $G$. 
A perturbed graph $G'$ of 
$G$ is 
with arbitrary edge manipulation $\{\mathcal{E}_+,\mathcal{E}_-\}$, node manipulation  
$\{\mathcal{V}_+, \mathcal{E}_{\mathcal{V}_+}, \mathcal{V}_-, \mathcal{E}_{\mathcal{V}_-}\}$, and node feature manipulation 
$\{{\bf X}_{\mathcal{V}_r}, \mathcal{V}_r, \mathcal{E}_{\mathcal{V}_r}\}$. 
Then at most $m=|\mathcal{E}_+| + |\mathcal{E}_-| + |\mathcal{E}_{\mathcal{V}_+}| + |\mathcal{E}_{\mathcal{V}_-}| + |\mathcal{E}_{\mathcal{V}_r}|$ 
predictions are different by the node/graph classifier $f$ on the subgraphs $\mathcal{G}'_T$ generated for the perturbed graph $G'$ and on $\mathcal{G}_T$. 
In other words, $\sum_{i=1}^{T}\mathbb{I}(f(G_{i})_v\neq f(G'_{i})_v) \leq m$ for any target node $v \in G$ in node classification or $\sum_{i=1}^{T}\mathbb{I}(f(G_{i})\neq f(G'_{i})) \leq m$ in graph classification. 
\end{theorem}

\noindent {\bf Deriving the robustness guarantee against arbitrary perturbation:} 
Based on Theorem~\ref{thm:suffcond} and Theorem~\ref{thm:edgebased}, we can derive the certified perturbation size as the maximal perturbation such that Equation~\ref{eqn:suff-NC} or Equation~\ref{eqn:suff-GC} is satisfied. Formally,

\begin{theorem}[Certified Robustness Guarantee with Edge-Centric Subgraphs against Arbitrary Perturbation]
\vspace{-2mm}
\label{thm:certifyedgebased} 
Let $f, y_a, y_b, c_{y_a}, c_{y_b}$ be defined above for edge-centric subgraphs, and   
$m$ be the perturbation size induced by an arbitrary perturbed graph $G'$ on $G$. 
The voting classifier $\bar{f}$ guarantees the same prediction on both $G'$ and $G$ for the target node $v$ in node classification (i.e., $\bar{f}(G')_v = \bar{f}(G)_v$) or target graph $G$ in graph classification (i.e., $\bar{f}(G') = \bar{f}(G)$), when $m$ satisfies
\begin{align}
\label{eqn:cpz_edge}
m \leq M = {\lfloor c_{y_a}-c_{y_b}-\mathbb{I}(y_{a}>y_{b})\rfloor} / {2}.
\end{align}
In other words, the maximum certified perturbation size is  $M$.
\end{theorem}
\noindent \emph{Remark:} We have the following remarks from our theoretical result in Theorem \ref{thm:certifyedgebased}. 

\begin{itemize}[leftmargin=*]
\vspace{-2mm}
\item {No (adaptive/unknown) attack can break {\name} if its perturbation budget is within the derived bound $M$, regardless of the attack knowledge of ANNCert.}
\vspace{-2mm}
\item It can be applied for any GNN node/graph classifier. 
\vspace{-2mm}

\item The  guarantee is true with a probability 100\%. 
\vspace{-2mm}
\item It treats existing robustness guarantees as special cases. 
\begin{itemize}[leftmargin=*]
    \item For edge manipulation $\{\mathcal{E}_+,\mathcal{E}_-\}$~\cite{bojchevski2020efficient,wang2021certified,xia2024gnncert}, the voting classifier $\bar{f}$ is certified robust if $
    |\mathcal{E}_+|+|\mathcal{E}_-|\leq M$.  
    \vspace{-1mm}
    
    \item For node manipulation $\{\mathcal{V}_+, \mathcal{E}_{\mathcal{V}_+}, 
    {\bf X}'_{\mathcal{V}_r}, \mathcal{V}_-, \mathcal{E}_{\mathcal{V}_-}\}$ \cite{lai2023nodeawarebismoothingcertifiedrobustness}, 
    $\bar{f}$ is certified robust if $|\mathcal{E}_{\mathcal{V}_+}|+|\mathcal{E}_{\mathcal{V}_-}|\leq M$.
    \vspace{-1mm}
    
    \item For node feature manipulation $\{ \mathcal{V}_r, \mathcal{E}_{\mathcal{V}_r}, {\bf X}_{\mathcal{V}_r}\}$~\cite{jin2020certified,xia2024gnncert}, 
$\bar{f}$ is certified robust if $|\mathcal{E}_{\mathcal{V}_r}|\leq M$.
    \vspace{-1mm}
    \item For both edge {and} node feature manipulation \cite{xia2024gnncert}, 
    $\bar{f}$ is certified robust if $|\mathcal{E}_+|+|\mathcal{E}_-| + |\mathcal{E}_{\mathcal{V}_r}|\leq M$. 
    
\end{itemize}

\end{itemize}

\subsection{Node-Centric Graph Division}
\label{Sec:nodebased}

We observe the robustness guarantee under edge-centric graph division is largely dominated by the  
edges (i.e., $\mathcal{E}_{\mathcal{V}_+}, \mathcal{E}_{\mathcal{V}_-}$) induced by the manipulated nodes $\mathcal{V}_+, \mathcal{V}_-$, and edges $\mathcal{E}_{\mathcal{V}_r}$ by the perturbed node features ${\bf X}'_{\mathcal{V}_r}$. 
This guarantee could be weak against node or node feature manipulation, 
as the number of edges (i.e., $|\mathcal{E}_{\mathcal{V}_+}|, |\mathcal{E}_{\mathcal{V}_-}|, |\mathcal{E}_{\mathcal{V}_r}|$) could be much larger, compared with the number of the nodes (i.e., $|{\mathcal{V}_+}|, |{\mathcal{V}_-}|, |{\mathcal{V}_r}|$). 
For instance, an injected node could link with many edges to a given graph in practice, and when the number exceeds $M$ in Equation~\ref{eqn:cpz_edge},  the certified robustness guarantee is ineffective. 

This flaw inspires us to generate subgraphs, where we expect at most one subgraph is affected 
under every node or node feature manipulation (this means all edges of a manipulated node should be in a same subgraph). 
We design a tailored {node-centric graph division} strategy to achieve our goal.

\noindent {\bf Naive solutions are ineffective:}  
A first solution is to map nodes into different subgraphs that are \emph{non-overlapped}, like mapping edges into subgraphs that are non-overlapped in edge-centric method.  
Though this method may work for graph classification, 
 it completely fails for node classification, as every node only appears once in all subgraphs and all target nodes can only receive one vote, yielding vacuous robustness. 

A second solution is to 
retain all nodes in every subgraph (say $G_i$), but keep only edges connected to nodes with the index $i$.  However, this idea still does not work, because some nodes not with index $i$ may still connect to nodes with index $i$, and manipulations on nodes with index $i$ would still influence representations of those nodes with a different index. 

\noindent {\bf Generating node-centric directed subgraphs:} 
We notice the failure of the second solution is because the message passing between two connected nodes $u$ and $v$ is bidirectional. 
If we decompose an undirected edge into two directed edges, and only use the outgoing edges of nodes, e.g., with index $i$, then the message 
is passed in one direction, i.e., from index $i$ nodes to their connected nodes. 
Hence, we propose dividing graphs into directed subgraphs.

We use a hash function $h$ to generate directed subgraphs for a given graph $G=(\mathcal{V},\mathcal{E},{\bf X})$. 
Our node-centric graph division strategy as follow: (1) we treat every undirected edge $e=(u,v) \in G$ as two directed edges for $u$\footnote{GNNs inherently handles directed graphs with directed message passing. Particularly, each node only uses its incoming neighbors' message for update.}: the outgoing edge $u \rightarrow v$ and incoming edge $v \rightarrow u$; (2) for every node $u$, we compute the subgraph index of its every outgoing edge $u \rightarrow v$: 
\begin{align}
\label{eqn:nodehash}
i_{u \rightarrow v} = h[\mathrm{str}(u)] \, \, \mathrm{mod} \, \, T+1. 
\end{align}
Note all outgoing edges of $u$ are mapped in the same subgraph.

We use $\vec{\mathcal{E}}_i$ to denote the set of directed edges whose subgraph index is $i$, i.e., $\vec{\mathcal{E}}_i = \{\forall u \rightarrow v \in {\mathcal{E}}: i_{u \rightarrow v}= i \}.$ 
Then, we can construct $T$ \emph{directed} subgraphs for $G$ as $\vec{\mathcal{G}}_T = \{ \vec{G}_i = (\mathcal{V},\vec{\mathcal{E}}_i,{\bf X}): i=1,2,\cdots, T\}$. 
{Here, we mention that we need to further postprocess the subgraphs for graph classification, in order to derive the robustness guarantee. 
Particularly, in each subgraph $\vec{G}_i$, we remove all other nodes whose subgraph index is not $i$. This is because although they have no influence on other nodes' representation, their information would still be passed to the global graph embedding aggregation. To make up the loss of connectivity between nodes and simulate the aggregation, we add an extra node with a zero feature, and add an outgoing edge from every node with index $i$ to it.}

\vspace{+0.05in}
\noindent {\bf Bounding the number of different subgraph predictions:} 
Similarly, for a perturbed graph ${G}'$, we use the same graph division strategy to generate a set of $T$ \emph{directed subgraphs} $\vec{\mathcal{G}}'_T = \{\vec{G}'_1, \vec{G}'_2, \cdots, \vec{G}'_T\}$. 
We first show the theoretical results that can upper bound the number of different subgraph predictions on $\vec{\mathcal{G}}_T$ and $\vec{\mathcal{G}}'_T$ against any individual perturbation.

\begin{theorem}[]
\label{thm:edgeperturb2}
Assume a graph $G$ is under the edge manipulation $\{\mathcal{E}_+,\mathcal{E}_-\}$, 
then at most $|\mathcal{E}_+| + |\mathcal{E}_-|$ subgraphs generated by our node-centric graph division have different predictions between $\vec{\mathcal{G}}'_T$ and $\vec{\mathcal{G}}_T$.  
\end{theorem}

\begin{proof}
We simply analyze when an arbitrary edge $(u, v)$ is deleted/added from $G$. It is obvious at most two subgraphs $G_{i_{u \rightarrow v} }$ and $G_{i_{v \rightarrow u} }$ are perturbed after perturbation, but via detailed analysis, at most one subgraph's prediction is affected. 

We consider the following two cases: i)  $i_{u \rightarrow v} = i_{v \rightarrow u}$. This means $u$ and $v$ are in the same subgraph, hence at most one subgraph's prediction is affected. ii) $i_{u \rightarrow v} \neq i_{v \rightarrow u}$. In subgraph $\vec{G}_{i_{u \rightarrow v}}$, $v$ only has incoming edges. Due to the message passing mechanism in GNNs, only the node $v$'s representation $\bm{h}_v^{(K)}$ is affected. Symmetrically in subgraph $\vec{G}_{i_{v \rightarrow u}}$,  only node $u$'s representation $\bm{h}_u^{(K)}$ is affected. Therefore, for node classification on a target node $w \in \mathcal{V}$, there exists at most one subgraph whose prediction is affected (when $w=u$ or $w=v$); for  graph classification, since $u$ (or $v$) is removed in subgraph $\vec{G}_{i_{v \rightarrow u}}$ (or $\vec{G}_{i_{u \rightarrow v}}$), no prediction is changed on the two subgraphs.

Generalizing the analysis to any $|\mathcal{E}_+| + |\mathcal{E}_-|$ edges in $G$ being perturbed, at most $|\mathcal{E}_+| + |\mathcal{E}_-|$ predictions are different between $\mathcal{G}_{T}$ and  $\mathcal{G}_{T}'$.  
\vspace{-2mm}
\end{proof}

\begin{theorem}[]
\label{thm:nodeperturb2} 
Assume a graph $G$ is under the node manipulation  
$\{\mathcal{V}_+, \mathcal{E}_{\mathcal{V}_+}, {\bf X}'_{\mathcal{V}_+}, \mathcal{V}_-, \mathcal{E}_{\mathcal{V}_-}\}$,
then at most $|\mathcal{V}_+| + |\mathcal{V}_-| $ subgraphs generated by our node-centric graph division have different predictions between $\vec{\mathcal{G}}'_T$ and $\vec{\mathcal{G}}_T$. 
\end{theorem}

\begin{theorem}[]
\label{thm:nodefeaperturb2} 
Assume a graph $G$ is under the node feature manipulation 
$\{\mathcal{V}_r, \mathcal{E}_{\mathcal{V}_r},{\bf X}'_{\mathcal{V}_r}\}$, 
then at most $|\mathcal{V}_r|$ subgraphs generated by our edge-centric graph division have different predictions between $\vec{\mathcal{G}}'_T$ and $\vec{\mathcal{G}}_T$. 
\end{theorem}

\begin{proof}
Our proof for the above two theorems is based on the key observation that: \emph{in a directed graph, manipulations on nodes with no outgoing edge have no influence on other nodes' representations in GNNs}. For any node  $u \in G$, only one subgraph $\vec{G}_{h[\mathrm{str}(u)] \, \, \texttt{mod} \, \, T+1}$ has outgoing edges.
Take node injection for instance and the proof for other cases are similar. Note that all subgraphs after node injection will contain newly injected nodes $V_{+}$, but they still do not have overlapped nodes with outgoing edges between each other via the hashing mapping. Hence, the injected nodes only have outgoing edges in at most $|V_{+}|$ subgraphs. 
Due to the directed message passing mechanism in GNNs, every node only uses its incoming neighboring nodes' representation to update its own representation. Hence, 
the injected nodes with no outgoing edges, whatever their features ${\bf X}'_{\mathcal{V}_+}$ are, would have no influence on other nodes' representations, implying at least $T-|V_{+}|$ subgraphs' predictions maintain the same.
\vspace{-2mm}
\end{proof}

\noindent \emph{Remark:} 
With edge manipulation, like Theorem~\ref{thm:edgeperturb}, Theorem~\ref{thm:edgeperturb2} has the same bounded number of altered subgraph predictions w.r.t. manipulated edges. 
Unlike Theorems~\ref{thm:nodeperturb} and~\ref{thm:nodefeaperturb}, Theorems~\ref{thm:nodeperturb2} and~\ref{thm:nodefeaperturb2} bound the number of altered subgraph predictions w.r.t. manipulated nodes. \emph{Importantly, we highlight these two bounds allow a manipulated node to link with many even infinite number of edges. Hence, these bounds are inherently robust against node inject attacks which often inject few nodes but with moderate number of edges, and node feature perturbations where the perturbed nodes have high degrees.}

With above  theorems, the total number of different subgraph predictions between  $\vec{\mathcal{G}}'_T$ and $\vec{\mathcal{G}}_T$ with \emph{arbitrary perturbation} can be straightforwardly bounded below.  

\begin{theorem}[Bounded Number of Node-Centric Subgraphs with Altered Predictions under Arbitrary Perturbation]
\label{thm:nodebased}
Let $f, v, G, \mathcal{E}_+, \mathcal{E}_-, {\mathcal{V}_+}, {\mathcal{V}_-}, \mathcal{V}_{r}$ be defined in Theorem~\ref{thm:edgebased}, and $\vec{\mathcal{G}}_T, \vec{\mathcal{G}}'_T$ contain directed subgraphs under the node-centric graph division.  
Then, at most $\bar{m} = |\mathcal{E}_+|+|\mathcal{E}_-| + |{\mathcal{V}_+}| + |{\mathcal{V}_-}| + |{\mathcal{V}_r}|$ predictions are different by the node/graph classifier $f$ on $\vec{\mathcal{G}}'_T$ and on $\vec{\mathcal{G}}_T$.   
In other words, $\sum_{i=1}^{T}\mathbb{I}(f(\vec{G}_{i})_v\neq f(\vec{G}'_{i})_v) \leq \bar{m}$ for any target node $v \in G$ in node classification or $\sum_{i=1}^{T}\mathbb{I}(f(\vec{G}_{i})\neq f(\vec{G}'_{i})) \leq \bar{m}$ in graph classification. 
\end{theorem}

\noindent {\bf Deriving the robustness guarantee against arbitrary perturbation:} Based on Theorem~\ref{thm:suffcond} and Theorem~\ref{thm:nodebased}, we can derive the certified perturbation size formally stated below

\begin{theorem}[Certified Robustness Guarantee with Node-Centric Subgraphs against Arbitrary Perturbation]
\label{thm:certifynodebased} 
Let $f, y_a, y_b, c_{y_a}, c_{y_b}$\footnote{Note that $c_{y_a}, c_{y_b}$ have different values with those in edge-centric graph division, as the generated node-centric subgraphs are different from edge-centric subgraphs. Here we use the same notation for description brevity.} be defined above for node-centric subgraphs, and   
$\bar{m}$ be the perturbation size induced by an arbitrary perturbed graph $G'$ on $G$. 
With a probability 100\%, the voting classifier $\bar{f}$ guarantees the same prediction on both $G'$ and $G$ for the target node $v$ in node classification (i.e., $\bar{f}(G')_v = \bar{f}(G)_v$) or the target graph $G$ in graph classification (i.e., $\bar{f}(G') = \bar{f}(G)$), 
if 
\begin{align}
\label{eqn:cpz_node}
\bar{m} \leq M = {\lfloor c_{y_a}-c_{y_b}-\mathbb{I}(y_{a}>y_{b})\rfloor} / {2}.
\end{align}
\end{theorem}

\noindent \emph{Remark:} 
Similarly,  our theoretical result can be applied for any GNN node/graph classifier, is true with probability 100\%, {and cannot be broken by any attack with perturbation budget $\leq M$.} 
Further, it can treat existing defenses as special cases. 

\begin{itemize}[leftmargin=*]
    \vspace{-1mm}
    \item For edge manipulation $\{\mathcal{E}_+,\mathcal{E}_-\}$~\cite{bojchevski2020efficient,wang2021certified,xia2024gnncert}, the voting classifier $\bar{f}$ is certified robust if $|\mathcal{E}_+|+|\mathcal{E}_-|\leq M$.
        
    \vspace{-1mm}
    \item For node manipulation $\{\mathcal{V}_+, \mathcal{E}_{\mathcal{V}_+}, {\bf X}'_{\mathcal{V}_+},\mathcal{V}_-, \mathcal{E}_{\mathcal{V}_-}\}$~\cite{lai2023nodeawarebismoothingcertifiedrobustness}, 
    $\bar{f}$ is certified robust if $|{\mathcal{V}_+}|+|{\mathcal{V}_-}|\leq M$.

    \vspace{-1mm}
    \item For node feature manipulation $\{\mathcal{V}_r, \mathcal{E}_{\mathcal{V}_r},{\bf X}'_{\mathcal{V}_r}\}$~\cite{jin2020certified,xia2024gnncert}, 
$\bar{f}$ is certified robust if $|{\mathcal{V}_r}|\leq M$. 

    \vspace{-1mm}
    \item For both edge {and} node feature manipulation \cite{xia2024gnncert}, 
    $\bar{f}$ is certified robust if $|\mathcal{E}_+|+|\mathcal{E}_-| + |{\mathcal{V}_r}|\leq M$. 

\end{itemize}

\section{Experiments}
\label{sec:eval}

\subsection{Experiment Settings}

\noindent {\bf Datasets:} We use four node classification datasets (Cora-ML~\cite{mccallum2000automating}, Citeseer~\cite{sen2008collective},  PubMed~\cite{sen2008collective}, Amazon-C~\cite{yang2021extract}) and four graph classification datasets (AIDS~\cite{riesen2008iam}, MUTAG~\cite{debnath1991structure}, PROTEINS~\cite{Borgwardt2005}, and DD~\cite{Dobson2003}) for evaluation. In each dataset, we take $30\%$ nodes (for node classification) or 50\% graphs (for graph classification) as the training set, $10\%$ and 20\% as the validation set, and the remaining nodes/graphs as the test set. 
Table~\ref{tab:datasets} shows the basic statistics of them. 
Our experiments are tested on a machine with NVIDIA RTX-4090 24G GPU, AMD EPYC 7352 CPU, and 60G RAM.

\vspace{+0.05in}
\noindent {\bf GNN classifiers and {\name} training:} 
We adopt the three well-known GNNs as the base node/graph classifiers: GCN~\cite{kipf2017semi}, GSAGE~\cite{hamilton2017inductive} and GAT~\cite{velivckovic2018graph}, and use their official source  code\footnote{https://github.com/tkipf/gcn; https://github.com/williamleif/GraphSAGE; https://github.com/PetarV-/GAT}.   
To enhance the robustness performance, existing certified defense~\cite{xia2024gnncert,lai2023nodeawarebismoothingcertifiedrobustness} augment the training set with generated subgraphs~\cite{xia2024gnncert} or noisy graphs~\cite{lai2023nodeawarebismoothingcertifiedrobustness} to train the GNN classifier.
Similarly, {\name} trains the GNN classifier using both the training nodes/graphs and their generated subgraphs, whose labels are same as the training nodes/graphs'. 
We denote the two versions of {\name} under edge-centric graph division and node-centric graph division as {\nameE} and {\nameN}, respectively. 
By default, we use GCN as the node/graph classifier in our experiments.

\vspace{+0.05in}
\noindent {\bf Evaluation metric:} 
Following existing works~\cite{xia2024gnncert,lai2023nodeawarebismoothingcertifiedrobustness,wang2021certified}, we use the certified node/graph  accuracy at perturbation size as the evaluation metric. 
For arbitrary perturbation, the perturbation size is 
the total number of manipulated nodes, edges, and nodes whose features can be arbitrarily perturbed. 
Given a perturbation size $m$ and test nodes/graphs,  
certified node/graph accuracy at $m$ is the fraction of test nodes/graphs that are accurately classified by the voting node/graph classifier and its certified perturbation size is no smaller than $m$. 
Note that the standard node/graph accuracy is under $m=0$.

\begin{figure*}[t]
\centering
\subfloat[Cora-ML]{\includegraphics[width=0.25\textwidth]{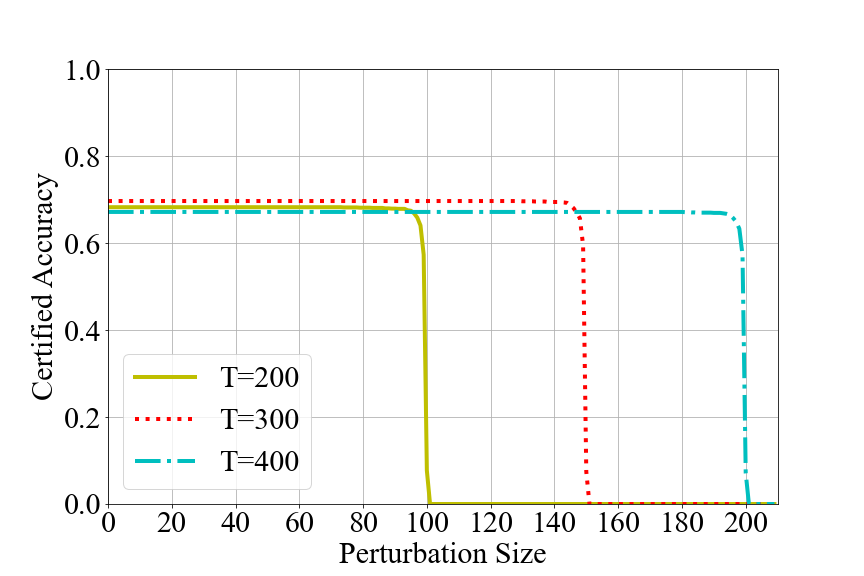}}\hfill
\subfloat[Citeseer]{\includegraphics[width=0.25\textwidth]{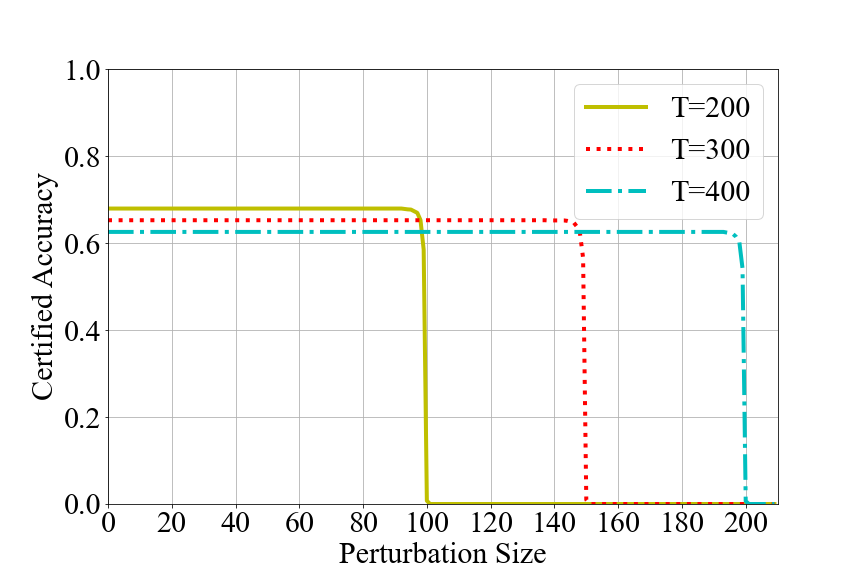}}\hfill
\subfloat[Pubmed]{\includegraphics[width=0.25\textwidth]{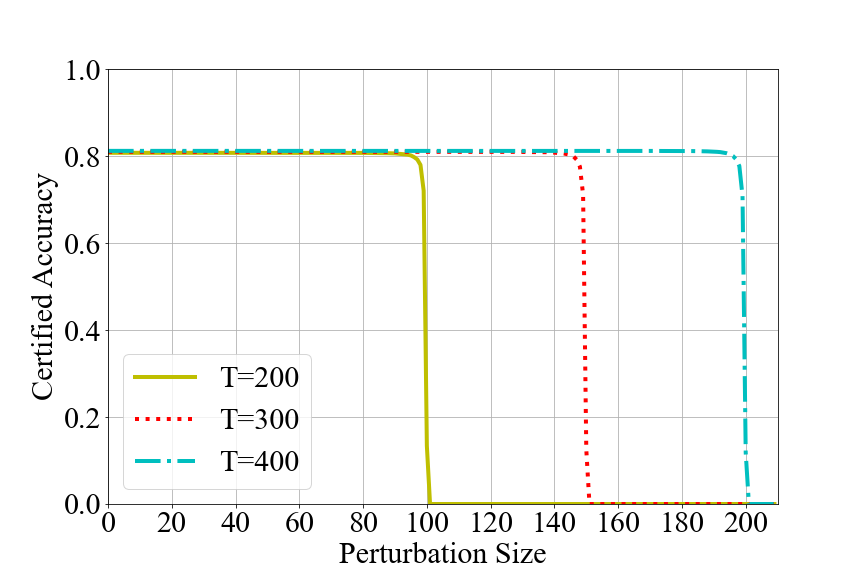}}\hfill
\subfloat[Amazon-C]{\includegraphics[width=0.25\textwidth]{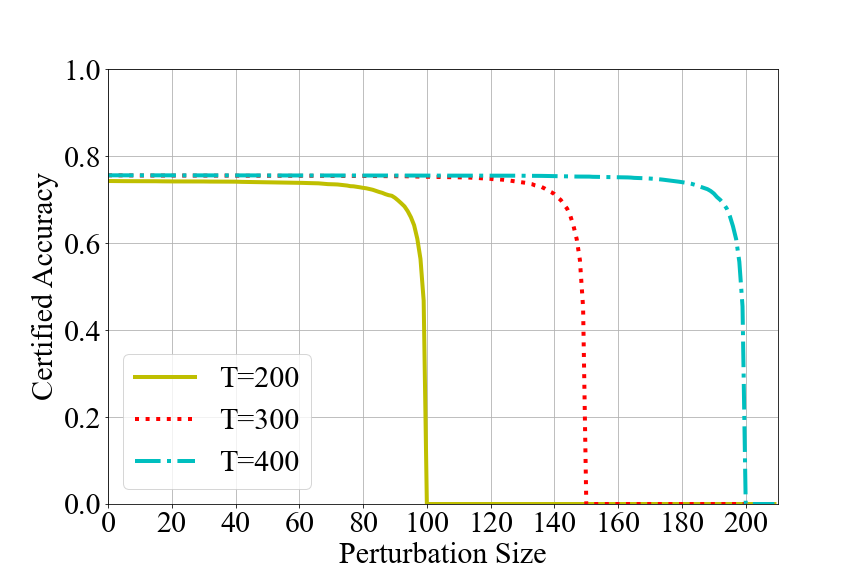}}\\
\caption{Certified node accuracy of our {\nameE} w.r.t. the number of subgraphs $T$. 
}
\label{fig:node-EC-T}
\vspace{-3mm}
\end{figure*}

\begin{figure*}[!t]
\centering
\subfloat[Cora-ML]{\includegraphics[width=0.25\textwidth]{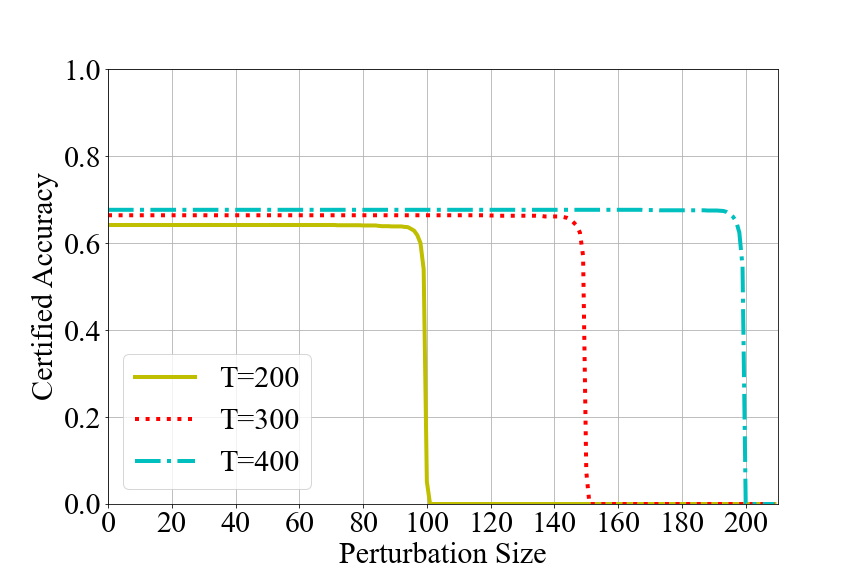}}\hfill
\subfloat[Citeseer]{\includegraphics[width=0.25\textwidth]{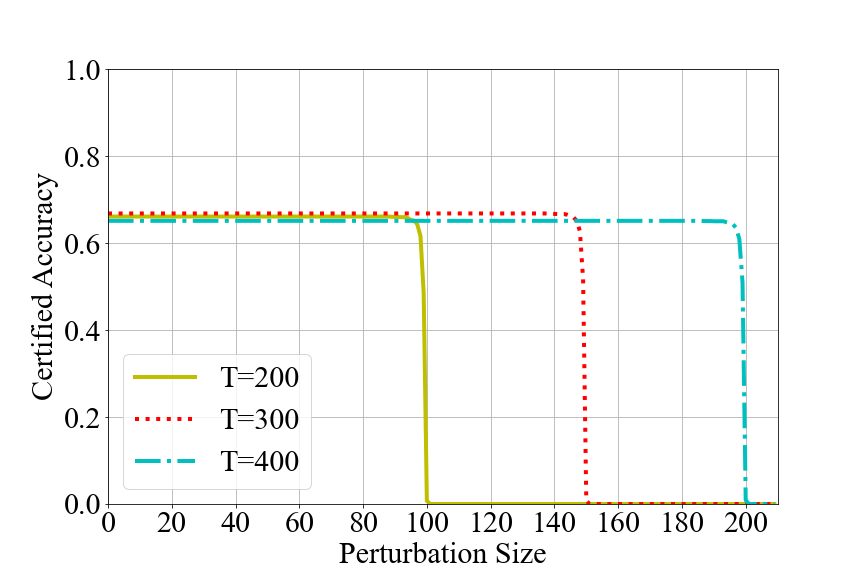}}\hfill
\subfloat[Pubmed]{\includegraphics[width=0.25\textwidth]{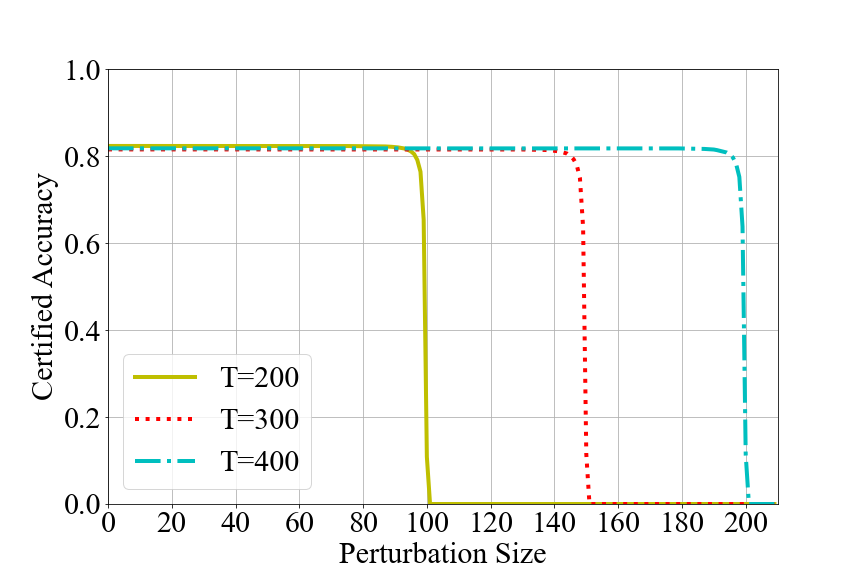}}\hfill
\subfloat[Amazon-C]{\includegraphics[width=0.25\textwidth]{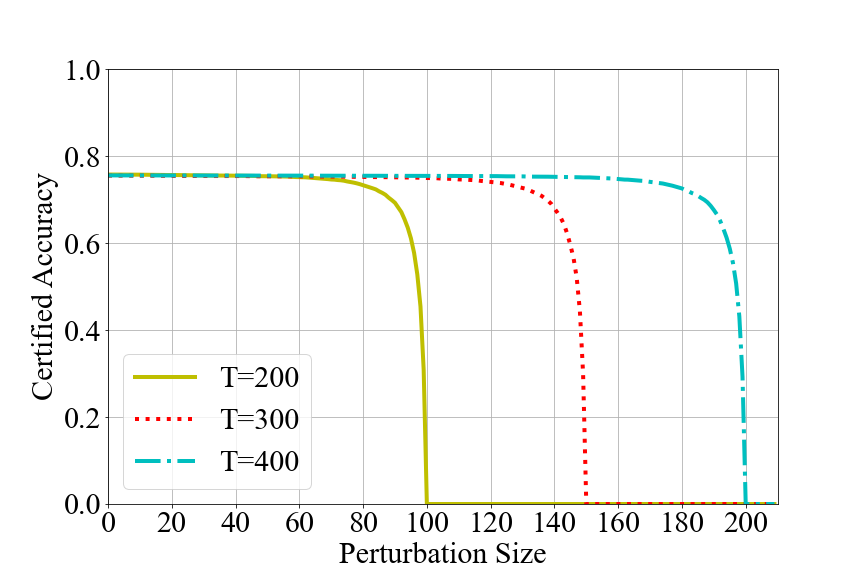}}\\
\caption{Certified node accuracy of our {\nameN} w.r.t. the number of subgraphs $T$.
}
\label{fig:node-NC-T}
\vspace{-3mm}
\end{figure*}

\begin{figure*}[!t]
\centering
\subfloat[AIDS]{\includegraphics[width=0.25\textwidth]{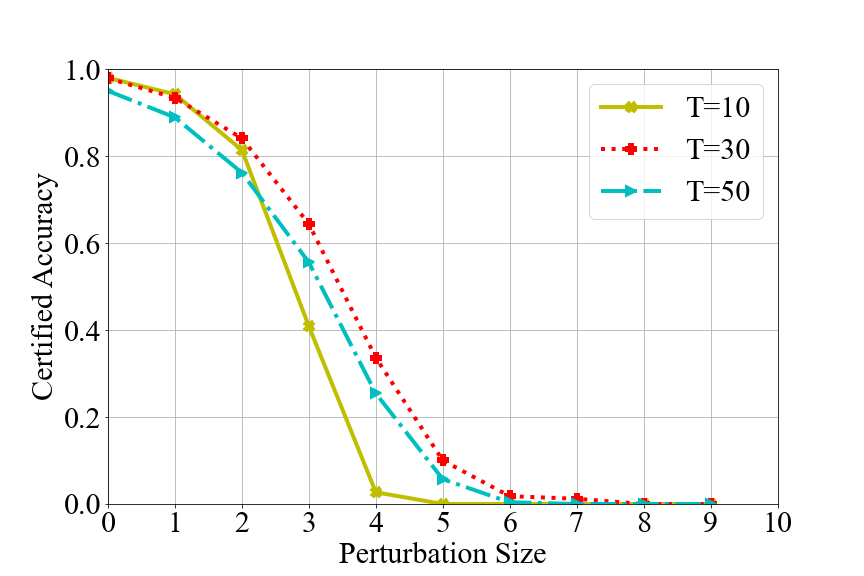}}\hfill
\subfloat[MUTAG]{\includegraphics[width=0.25\textwidth]{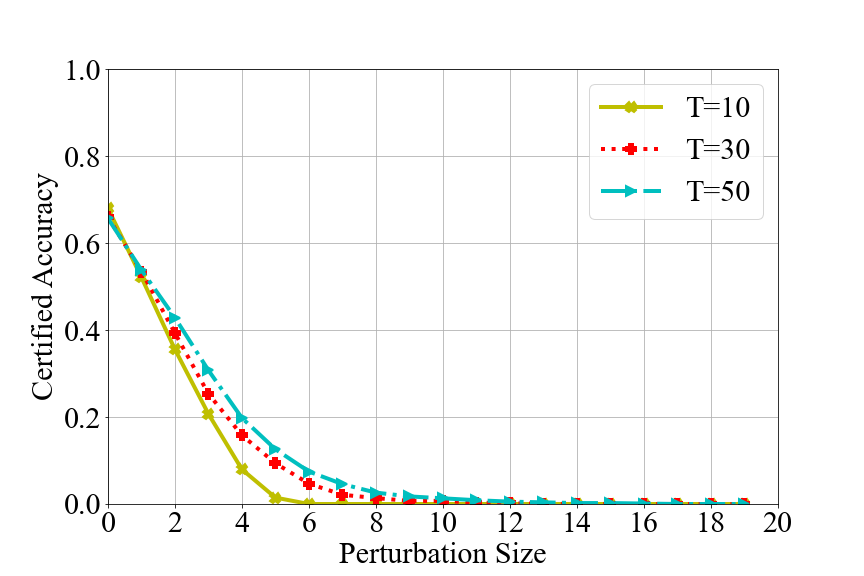}}\hfill
\subfloat[PROTEINS]{\includegraphics[width=0.25\textwidth]{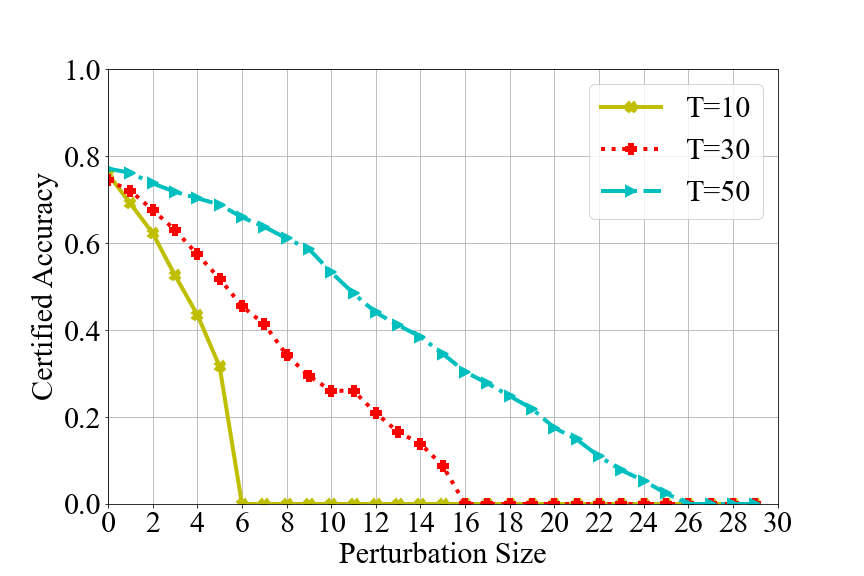}}\hfill
\subfloat[DD]{\includegraphics[width=0.25\textwidth]{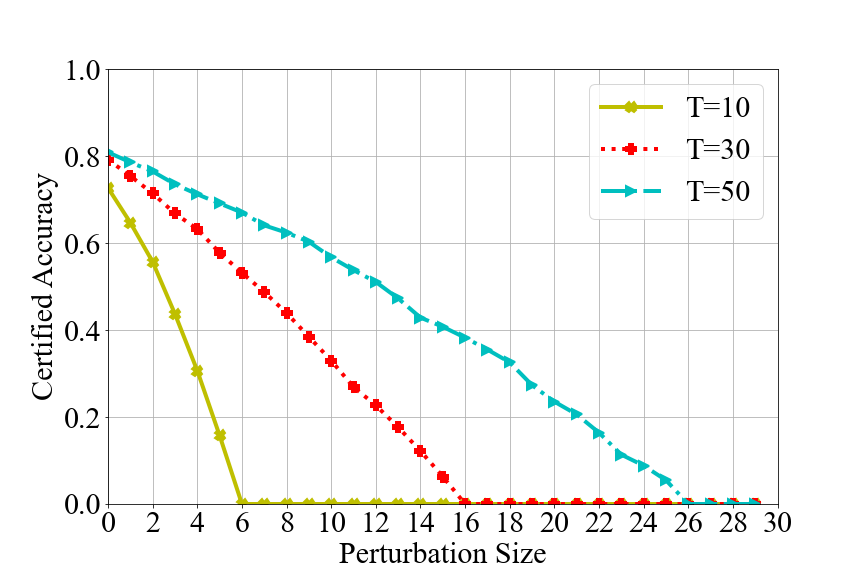}}\\
\caption{Certified graph accuracy of our {\nameE} w.r.t. the number of subgraphs $T$.
}
\label{fig:graph-EC-T}
\vspace{-3mm}
\end{figure*}

\begin{figure*}[!t]
\centering
\subfloat[AIDS]{\includegraphics[width=0.25\textwidth]{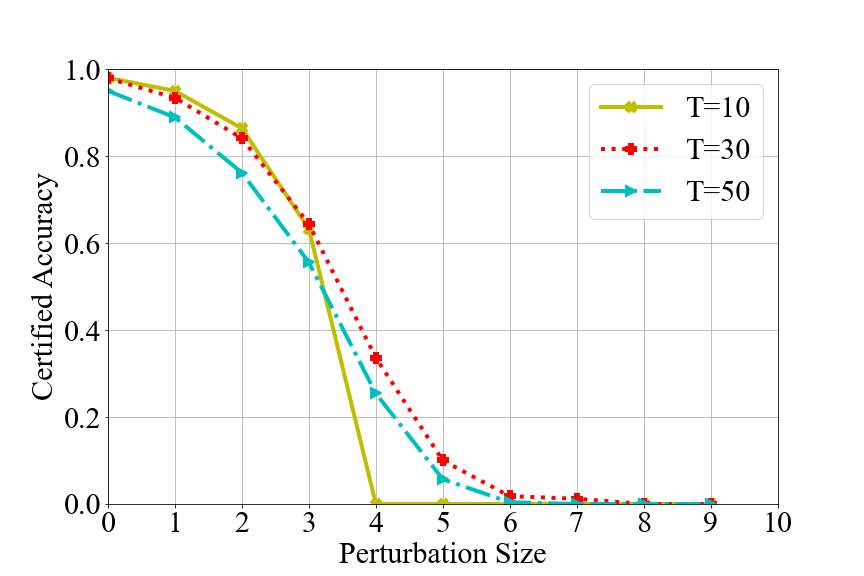}}\hfill
\subfloat[MUTAG]{\includegraphics[width=0.25\textwidth]{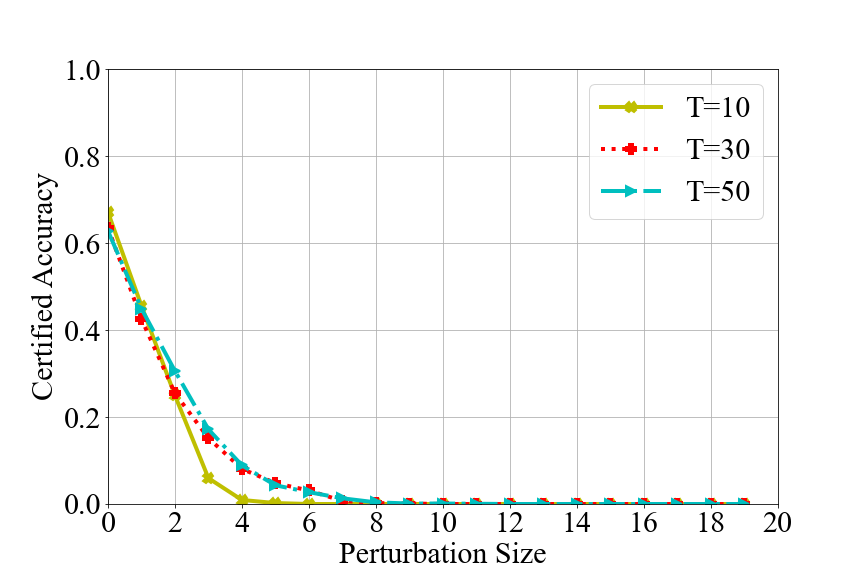}}\hfill
\subfloat[PROTEINS]{\includegraphics[width=0.25\textwidth]{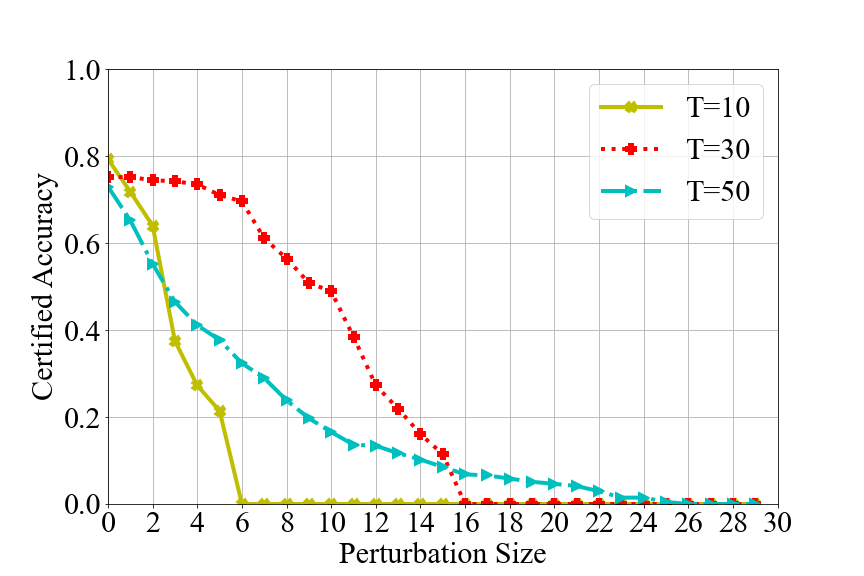}}\hfill
\subfloat[DD]{\includegraphics[width=0.25\textwidth]{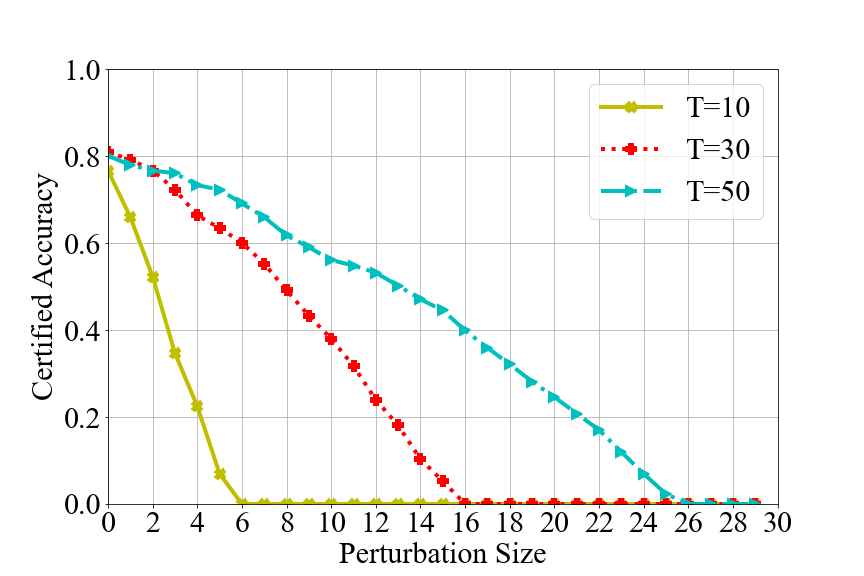}}\\
\caption{Certified graph accuracy of our {\nameN} w.r.t. the number of subgraphs $T$.}
\label{fig:graph-NC-T}
\vspace{-2mm}
\end{figure*}

\begin{table}[!t]
\vspace{+2mm}
    \footnotesize
    \centering 
    \addtolength{\tabcolsep}{-2pt}
    \renewcommand\arraystretch{0.9}
    \begin{tabular}{c|c|c|c|c}
     \toprule
          {\bf Node Classification}&{\bf Ave degree}&{$|\mathcal{V}|$}&$|\mathcal{E}|$&$|\mathcal{C}|$ \\
         \Xhline{0.9pt}
       Cora-ML&5.6&2, 995&8,416&7\\
         \cline{1-5} 
         Citeseer&2.8&3,327&4,732&6\\
         \cline{1-5} 
         Pubmed&4.5&19,717&44,338&3\\
         \cline{1-5} 
         Amazon-C&71.5&13,752&491,722&10\\
         \cline{1-5} 
        {Amazon2M}&50.5&2,449,029&61,859,140&47 \\
         \Xhline{1.2pt} 
       {\bf Graph Classification}&$|\mathcal{G}|$&$|\mathcal{V}|_{avg}$&$|\mathcal{E}|_{avg}$&$|\mathcal{C}|$ \\
         \Xhline{0.9pt}
         {AIDS}&2,000&15.7&16.2&2\\
         \cline{1-5} 
        MUTAG&4,337&30.3&30.8&2\\
         \cline{1-5} 
         PROTEINS&1,113&39.1&72.8&2\\
         \cline{1-5} 
        DD&1,178&284.3&715.7&2\\
         \cline{1-5} 
        {Big-Vul} &18,103 &35.5&117.3&2 \\
       \bottomrule
    \end{tabular}
    \caption{Datasets and their statistics.}
    \label{tab:datasets}
    \vspace{-2mm}
\end{table}

\vspace{+0.05in}
\noindent {\bf Compared baselines:} 
As  {\name} encompasses existing defenses as special cases, we can compare {\name} with them against less types of perturbation. 
Here, we choose the state-of-the-art Bi-RS~\cite{lai2023nodeawarebismoothingcertifiedrobustness} 
and GNNCert~\cite{xia2024gnncert} for comparison.

\begin{itemize}[leftmargin=*]

\vspace{-2mm}
\item {\bf Bi-RS:} It certifies 
GNN for \emph{node classification} against node inject attacks with a probabilistic guarantee. 
During training, 
Bi-RS augments the graph with 
$N_1$ noisy graphs from a smoothing distribution (defined in its Eqn.3) and trains the node classifier with both clean graphs and their noisy ones. 
During certification, Bi-RS utilizes Monte-Carlo sampling to compute the certified perturbation size. Given 
a graph and the trained node classifier, Bi-RS first generates $N_2$ noisy graphs for the given graph and then derives the robustness guarantee for each target node on the noisy graphs that is correct with a probability $1-\alpha$. 
Note that ensuring a smaller $\alpha$ needs more samples.
Bi-RS sets $N_2=50,000$ and $\alpha=0.01$. In our experiment, 
we also set $N_1=T$.

\vspace{-2mm}
\item {\bf GNNCert:} It is the state-of-the-art certified defense (with a deterministic guarantee) of GNN for \emph{graph classification} against edge manipulation, and both edge  and node feature manipulation (more details see Section~\ref{sec:GNNCert}). 
We denote the two variants as GNNCert-E and GNNCert-EN, respectively. 
During training,  GNNCert-E and  GNNCert-EN use the extra $T_e$ and $T_e \cdot T_n$ subgraphs for training the base graph classifier. During certification, GNNCert-E and  GNNCert-EN also use the same number of subgraphs.  
\emph{We highlight that, for edge manipulation, GNNCert-E has the same bound as our {\nameE} under edge-centric graph division. This is because the generated subgraphs of both defenses are exactly the same, and so does the voting graph classifier when using the same base GNN classifier.}  

\end{itemize}

\vspace{+0.05in}
\noindent {\bf Parameter setting:} {\name} has two hyperparameters: the hash function $h$ and the number of subgraphs $T$. By default, we use MD5 as the hash function and set $T=30,300$ respectively for node and graph classification, considering their different graph sizes. 
 We also study the impact of them.

\subsection{Experiment Results}

\subsubsection{{\name} against Arbitrary Perturbation}

\noindent {\bf Main results:} Figures~\ref{fig:node-EC-T}-\ref{fig:node-NC-T} show the certified node accuracy and Figures~\ref{fig:graph-EC-T}-\ref{fig:graph-NC-T} show the certified graph accuracy at perturbation size $m$ w.r.t. $T$ under the two graph division strategies, respectively.
We have the following observations. 

\begin{itemize}[leftmargin=*]
\vspace{-2mm}
\item Both {\nameE} and {\nameN} can tolerate the perturbation size up to 200 and 25, on the node classification and graph classification datasets, respectively. This means {\nameE} can defend against a total of 200 (25) arbitrary edges, while {\nameN}  against a total of 200 (25) arbitrary edges and nodes caused by the arbitrary perturbation,  on the node (graph) classification datasets, respectively. 
Note that node classification datasets have several orders of more nodes/edges than graph classification datasets, hence {\name} can tolerate more perturbations on them.

\vspace{-2mm}
\item $T$ acts as the robustness-accuracy tradeoff. That is, a larger (smaller) $T$ yields a higher (lower) certified perturbation size, but a smaller (higher) normal accuracy ($m=0$). 

\vspace{-2mm}
\item In {\nameN}, the guaranteed perturbed nodes can 
have an infinite number of edges. This thus implies  {\nameN} produces better   robustness than {\nameE} against the perturbed  edges by node/node feature manipulation. 
\end{itemize}

\vspace{+0.05in}
\noindent {\bf Impact of hash function:} 
Figure~\ref{fig:node-EC-hash}-Figure~\ref{fig:graph-NC-hash} in Appendix 
show the certified node/edge accuracy of {\nameE} and  {\nameN} with different hash functions. We observe that our certified accuracy and certified perturbation size are almost the same in all cases. This reveals  {\name} is insensitive to hash functions, and \cite{xia2024gnncert} draws a similar conclusion. 

\vspace{+0.05in}
\noindent {\bf Impact of base GNN classifiers:} Figures~\ref{fig:node-EC-T-GSAGE}-\ref{fig:graph-NC-T-GSAGE} and 
Figures~\ref{fig:node-EC-T-GAT}-\ref{fig:graph-NC-T-GAT} in Appendix show the certified accuracy at perturbation size using GSAGE and GAT as the base classifier, respectively. We have similar observations as those results with GCN. For instance, $T$ trade offs robustness and accuracy.

\begin{table}[!t]
\centering
\footnotesize
\addtolength{\tabcolsep}{-3.5pt}
\caption{Node/graph accuracy of normally trained GNN and of {\name} with GNN trained on the subgraphs.}
\begin{tabular}{|c|c|c|c|c|c|c|c|c|c|}
\hline
\multirow{2}{*}{\bf Dataset} & \multirow{2}{*}{\bf GCN} & \multicolumn{2}{c|}{\name} & \multirow{2}{*}{\bf GSAGE} & \multicolumn{2}{c|}{\name} & \multirow{2}{*}{\bf GAT} & \multicolumn{2}{c|}{\name} \\ \cline{3-4}\cline{6-7}\cline{9-10}
&&-E&-N&&-E&-N&&-E&-N
\\ \hline
{\bf Cora-ML} &0.73&0.70&0.68&0.67& 0.67 & 0.68 &0.74&  0.68&0.69 \\ \hline
{\bf Citeseer} &0.66&0.65&0.67&0.64& 0.63 & 0.64 &0.66& 0.65&0.66 
  \\ \hline
{\bf Pubmed} &0.86&0.81&0.82&0.84&  0.84&0.84&0.85& 0.84& 0.84 \\ \hline
{\bf Amazon-C}&0.81&0.76&0.76&0.80&0.77  &0.75 &0.78& 0.74& 0.74 \\ \hline \hline
{\bf AIDS}&0.99&0.98&0.96&0.97& 0.96 &0.97 &0.96&0.98& 0.98  \\ \hline
{\bf MUTAG}&0.71&0.66&0.65&0.70& 0.66 &0.67 &0.71& 0.67&0.66 \\ \hline
{\bf Proteins}&0.75&0.75&0.75&0.80& 0.79 &  0.77&0.82
&0.77 & 0.77\\ \hline
{\bf DD}&0.80 &0.79
&0.81&0.81& 0.80 & 0.81&0.81& 0.77& 0.80\\ \hline
\end{tabular}
\label{tbl:normalacc}
\end{table}

\vspace{+0.05in}
\noindent {\bf Impact of subgraphs on the certified accuracy:} 
We test the certified accuracy of (not) using subgraphs to train the GNN classifier.  
Figures~\ref{fig:node-EC-w-wo}-\ref{fig:graph-EC-w-wo} in Appendix show the comparison results under the default $T$ for node and graph classification.    
The results show training with subgraphs can enhance the certified robustness of {\name}, especially on large datasets. This is because training and certification both involve raw graphs and the subgraphs, making their distributions similar.

\vspace{+0.05in}
\noindent {\bf Impact of subgraphs on the normal accuracy:} 
We test the normal accuracy of (not) using subgraphs to train the GNN classifier.  
Table~\ref{tbl:normalacc} shows the comparison results of the test node/graph accuracy of the normally trained GNN without sbugraphs and  {\name} with GNN trained on the subgraphs. 
We observe that the accuracy of {\name} is 5\% smaller than that of normally trained GNN in almost all cases, and in some cases even larger. This implies the augmented subgraphs for training marginally affects the normal test accuracy.

\begin{figure*}[!t]
\centering
\subfloat[Cora-ML]{\includegraphics[width=0.25\textwidth]{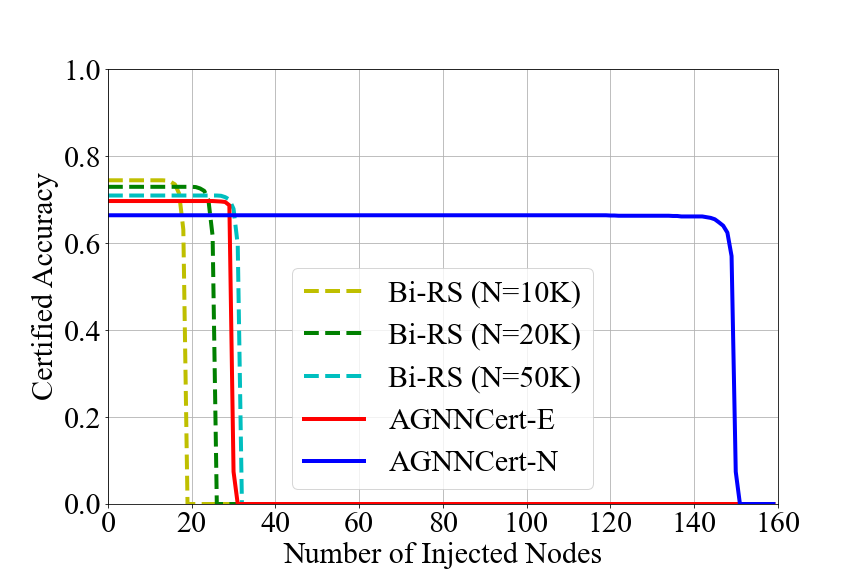}}\hfill
\subfloat[Citeseer]{\includegraphics[width=0.25\textwidth]{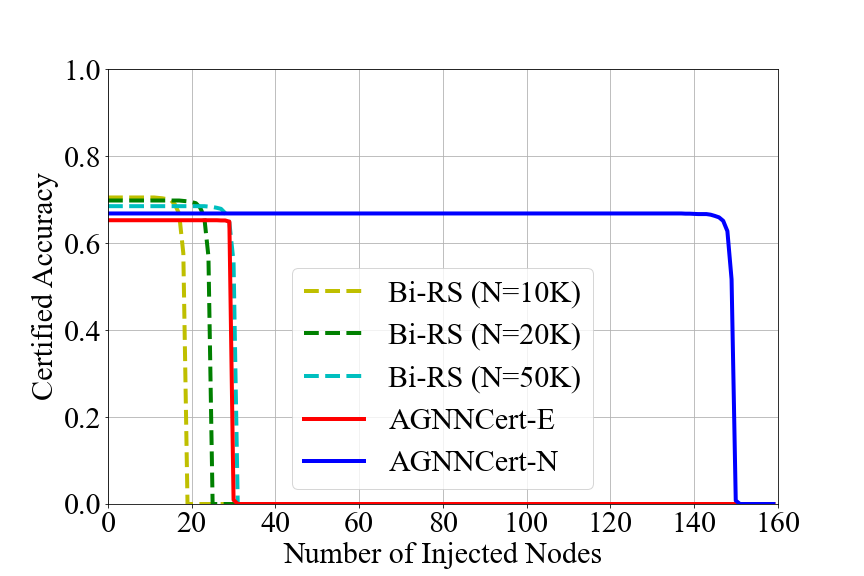}}\hfill
\subfloat[Pubmed]{\includegraphics[width=0.25\textwidth]{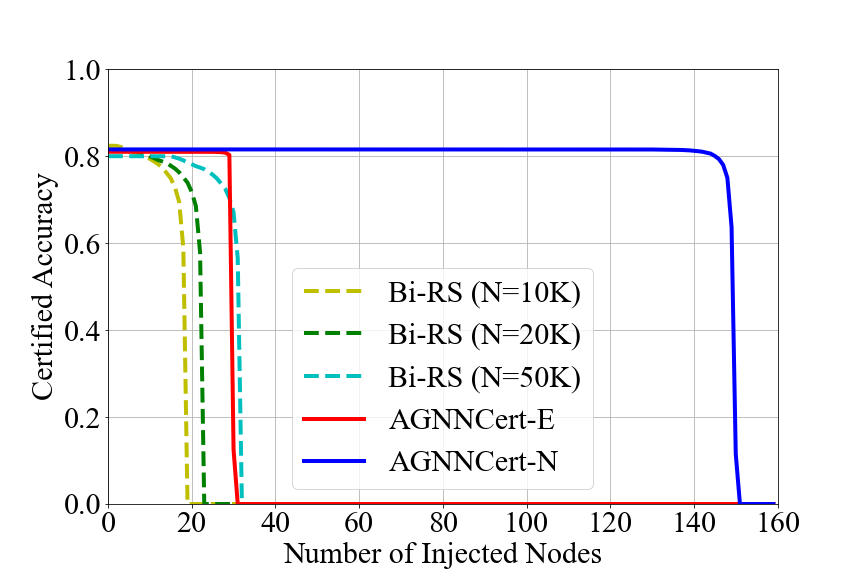}}\hfill
\subfloat[Amazon-C]{\includegraphics[width=0.25\textwidth]{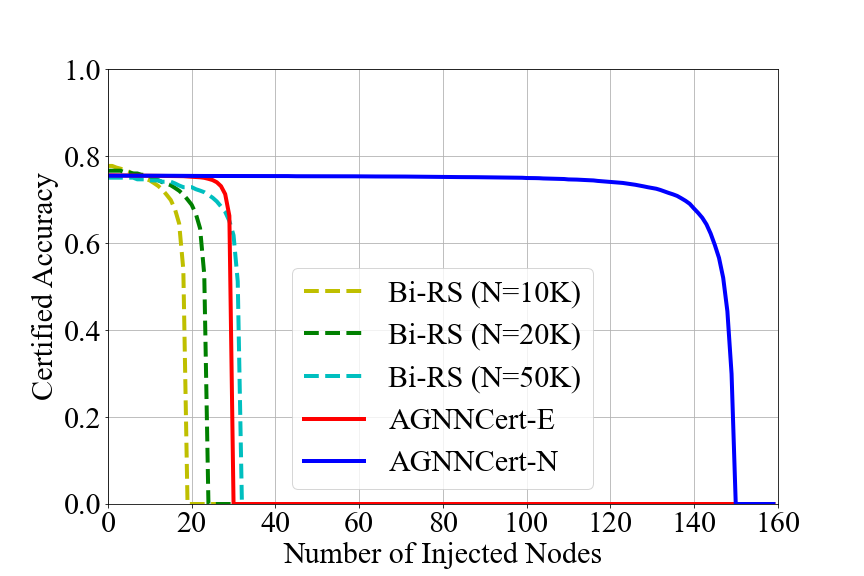}}\\
\caption{Certified node accuracy of {\name} and Bi-RS against node inject attacks. 
}
\label{fig:ours-E-vs-biRS}
\vspace{-4mm}
\end{figure*}

\begin{figure*}[!t]
\centering
\subfloat[AIDS]{\includegraphics[width=0.25\textwidth]{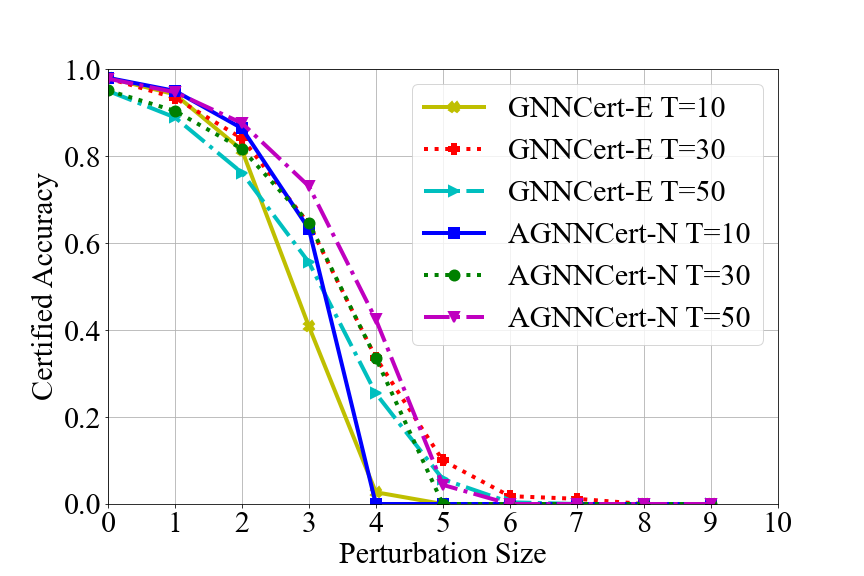}}\hfill
\subfloat[MUTAG]{\includegraphics[width=0.25\textwidth]{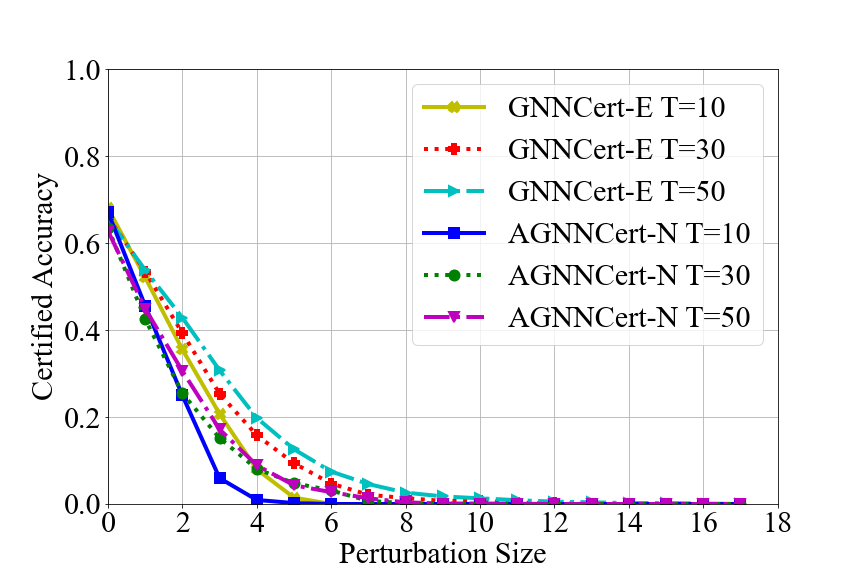}}\hfill
\subfloat[PROTEINS]{\includegraphics[width=0.25\textwidth]{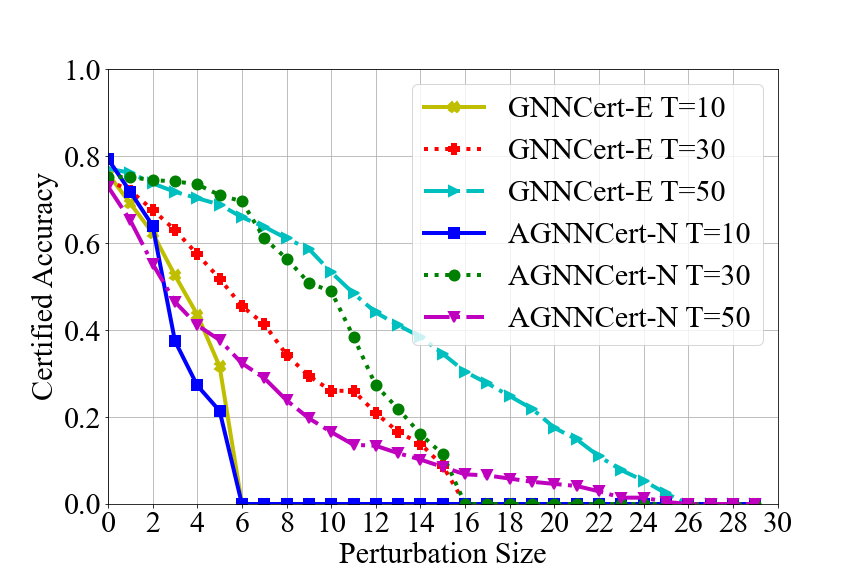}}\hfill
\subfloat[DD]{\includegraphics[width=0.25\textwidth]{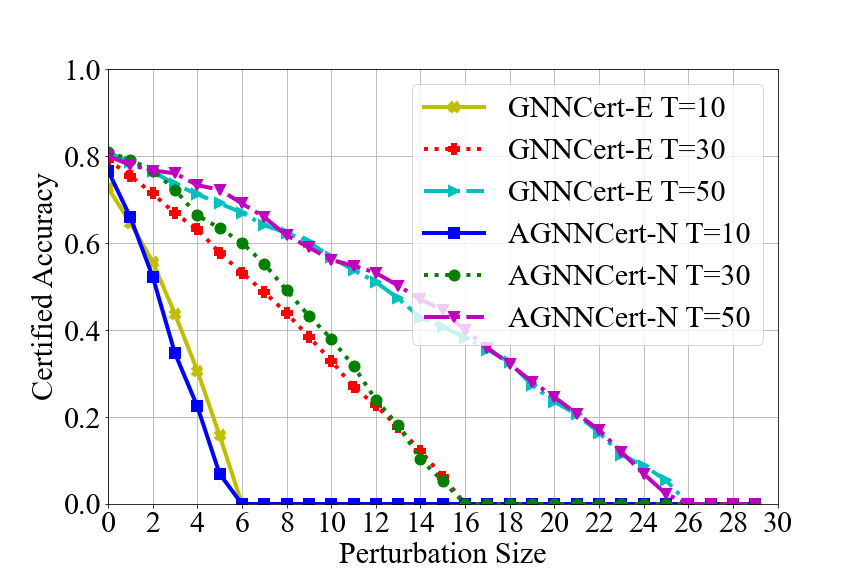}}\\
\caption{Certified graph accuracy of {\nameN} and GNNCert-E against edge manipulation.}
\label{fig:ours-vs-gnncert-edge}
\vspace{-4mm}
\end{figure*}

\begin{figure*}[!t]
\centering
\subfloat[AIDS]{\includegraphics[width=0.25\textwidth]{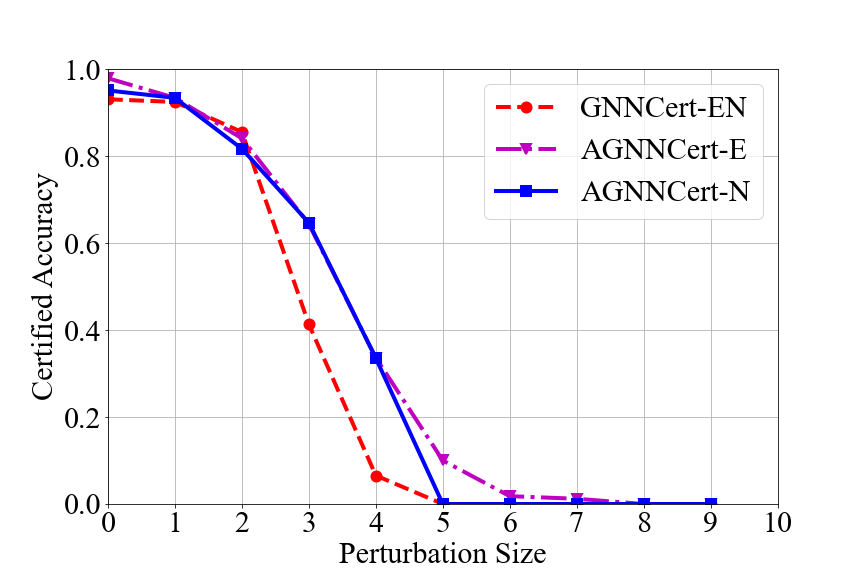}}\hfill
\subfloat[MUTAG]{\includegraphics[width=0.25\textwidth]{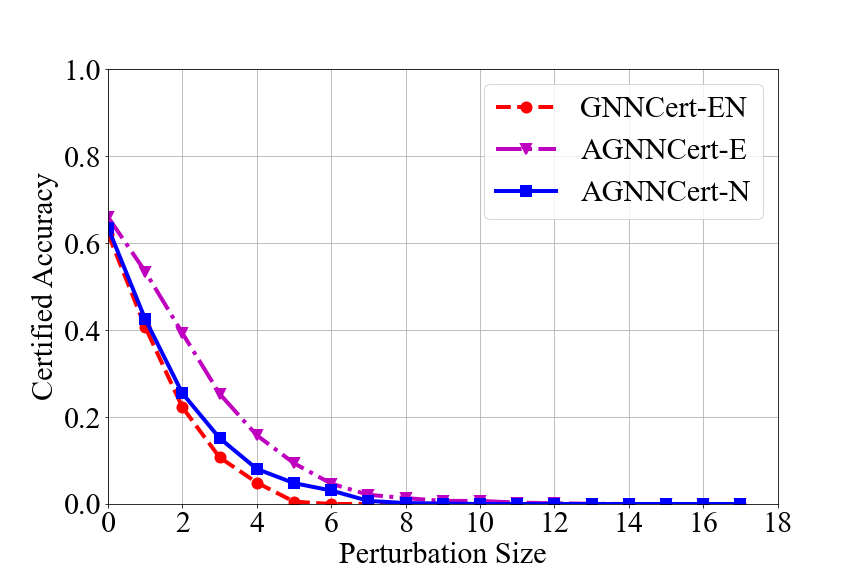}}\hfill
\subfloat[PROTEINS]{\includegraphics[width=0.25\textwidth]{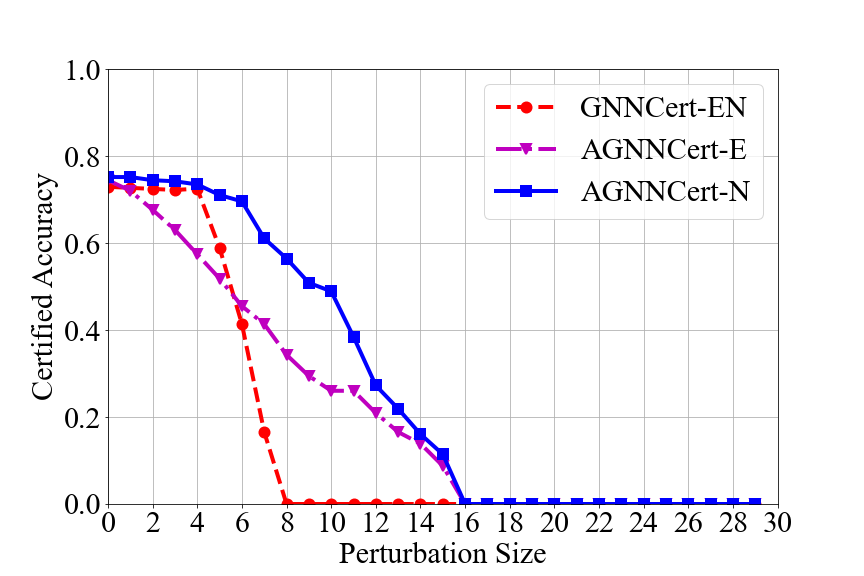}}\hfill
\subfloat[DD]{\includegraphics[width=0.25\textwidth]{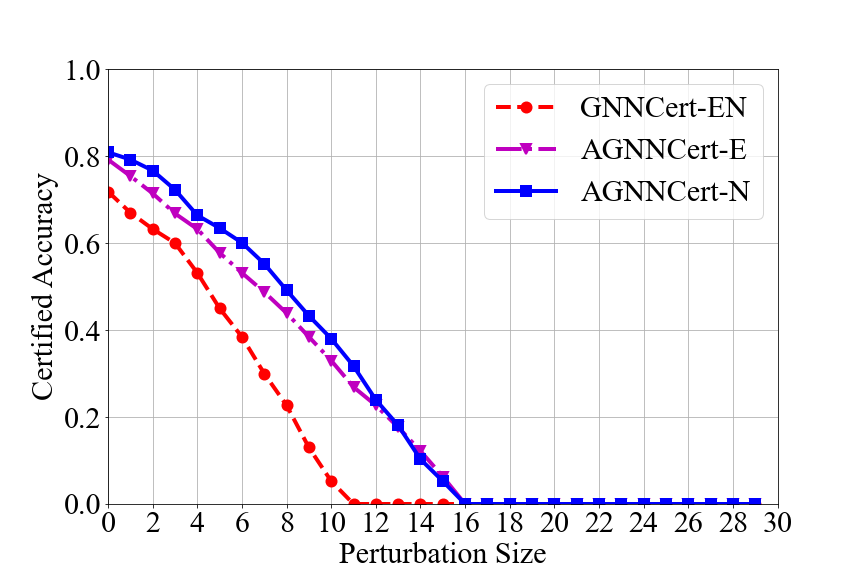}}\\
\caption{Certified graph accuracy of {\name} and GNNCert-EN against edge and node feature manipulation.}
\label{fig:ours-vs-gnncert-edgenode}
\vspace{-2mm}
\end{figure*}

\subsubsection{Comparing {\name} with Bi-RS and GNNCert}

\noindent {\bf Comparing {\name} with Bi-RS for node classification against node injection attacks:} 
We first add some details of Bi-RS.
Bi-RS assumes the number of injected nodes is $\rho$ and each node can connect at most $\tau$ edges, so the total perturbed edges is $\rho \cdot \tau$. It also involves two hyperparameters $p_e$ and $p_n$, which means the probability of deleting an edge and deleting a node (and all its connected edges), respectively. These parameters are used to derive the certified perturbation size (see its Eqn 5). 
In the experiment, we follow Bi-RS by setting $\tau=5$ and pick its best result from 9 combinations with $p_e=\{0.7,0.8,0.9\}$ and $p_n=\{0.7,0.8,0.9\}$. 
Figure~\ref{fig:ours-E-vs-biRS} shows the comparison results.

\begin{itemize}[leftmargin=*]
\vspace{-2mm}
\item {\bf {\nameE} vs Bi-RS:} 
We first mention the number of injected nodes in {\nameE} is calculated by dividing the bounded number of edges in Equation~\ref{eqn:cpz_edge} by $\tau$. We can see the two methods have comparable certified node accuracy w.r.t. the number of injected nodes, which indicates {\nameE} is already as effective as Bi-RS.  
Further, we highlight our {\nameE}'s theoretical result is \emph{deterministic} and \emph{far more general}---it bounds the total number of perturbed edges induced by the node inject attack, where the combination of the number of injected nodes and the number of incident edges for each injected node is arbitrary.  

\vspace{-2mm}
\item {\bf {\nameN} vs Bi-RS:} We can see {\nameN} has much better certified node accuracy than Bi-RS w.r.t. the number of injected nodes (under $\tau=5$). Furthermore, we highlight that each bounded node in {\nameN} can inject as many (even infinite) edges as possible. Hence, the total number of bounded edges in {\nameN} could be infinite, which is infinitely higher than Bi-RS's bound when using the total perturbed edges as the evaluation metric.  

\vspace{-2mm}
\end{itemize}

\vspace{+0.05in}
\noindent {\bf Comparing {\name} with GNNCert for graph classification against edge manipulation:} 
Recall that, when using the same hash function and same number of subgraphs in both defenses, {\nameE} and GNNCert-E produce the same subgraphs and same voting graph classifier. Hence, their certified graph accuracy/perturbation size are same.  
Here, we compare {\nameN} with GNNCert-E, and results are in Figure~\ref{fig:ours-vs-gnncert-edge}.    
{We observe both methods have close certified accuracy/perturbation size, implying they have comparable robustness guarantee  
against edge manipulation}.

\begin{table}[!t]
\centering
\caption{Big-O complexity comparison for defense training and certification. {We also include the base GNN for completeness.} We do not include other complexity factors in  training and certification, as they are similar in all defenses. In practice, $N_2$ can be as large as $100,000$; $N_1, T_e, T_n$ and $T$ have values $\le 100$. Hence $N_2 \gg N_1 \simeq T_e \simeq T_n \simeq T$.} 
\begin{tabular}{lcc}
\toprule
{\bf Defenses} & {\bf Training} & {\bf Certification} \\
\midrule
{\bf GNN} &  O($1$) & O($1$)\\
{\bf Bi-RS} &  O($N_1$) & O($N_2$)\\
{\bf GNNCert-E} &  $O(T_e)$ & $O(T_e)$\\
{\bf GNNCert-EN} &  $O(T_e \cdot T_n)$ & $O(T_e \cdot T_n)$\\
{\bf {\nameE}} & $O(T)$ & $O(T)$\\
{\bf {\nameN}} &  $O(T)$ & $O(T)$\\
\bottomrule
\label{computation-cost-training-testing}
\end{tabular}
\vspace{-8mm}
\end{table}

\vspace{+0.05in}
\noindent {\bf Comparing {\name} with GNNCert against edge AND node feature manipulation:} 
As analyzed in Section~\ref{sec:GNNCert}, the initial  guarantee of GNNCert is for edge manipulation \emph{or} node feature manipulation. To defend against both manipulations, it requires $T_e=T_n$. Figure~\ref{fig:ours-vs-gnncert-edgenode} shows the comparison results under $T_e=T_n=T$. 
We can see our  {\name} performs better than GNNCert-EN. For instance, on PROTEINS,  
{\nameE} can certify a total of 15 perturbed edges by both manipulations, and {\nameN} certifies a total of 15 edges and nodes whose features can be arbitrarily perturbed.
Instead, GNNCert-EN can only tolerate up to 7 edges and nodes.  
This may because, compared to {\name}, GNNCert-EN generates far more subgraphs ($T^2$) with each subgraph having less edges and many nodes in subgraphs do not have features (0 values), thus using much less information in the raw graph.

\begin{table*}[!t]\renewcommand\arraystretch{1}
\centering
\small
\addtolength{\tabcolsep}{-3.5pt}
\caption{Training and test time of provable defenses and undefended GNN on the evaluated datasets.}
\begin{tabular}{|c|c|c|c|c|c|c|c|c|c|c|}
\hline
\multicolumn{2}{|c|}{\bf Datasets} & {\bf Cora-ML}&{\bf Citeseer}&{\bf Pubmed} &{\bf Amazon-C}
&{\bf Datasets}& {\bf AIDS} &{\bf MUTAG}& {\bf PROT.} &{\bf DD} \\ \hline
& {\bf GCN}&0.03s &0.03s &0.12s &0.31s&GCN&6.66s&14.82s &3.87s &6.45s 
\\
{Training Time}& {\bf Bi-RS}& 16.73s&22.21s &117.57s &98.10s & {\bf GNNCert-E} &114.90s &388.01s &107.72s & 171.34s
\\
{(per epoch)}& {\nameE}& 17.46s&21.44s &110.58s &102.31s &{\nameE}&100.55s & 389.08s&95.70s &163.27s \\
& {\nameN}& 18.59s&22.47s &102.26s &96.55s &{\nameN}&101.94s & 400.97s&98.61s &151.18s 
\\\hline
\multirow{4}{*}{Test/Certification Time}& {\bf GCN} &0.01s &0.01s &0.02s & 0.08s&GCN&1.46s &2.66s &0.70s &1.02s 
\\
& {\bf Bi-RS} &1658s & 1943s&60589s &15792s &{\bf GNNCert-E} &22.15s &82.21s &26.38s &32.85s 
\\
& {\nameE}&7.35s &8.36s &44.91s & 35.29s&{\nameE}&24.34s &82.68s &23.05s &33.14s 
\\
& {\nameN}&7.45s &8.40s &42.69s & 36.41s&{\nameN}&22.57s &86.15s &25.45s &32.85s 
\\\hline
\end{tabular}
\label{tbl:exp_time}
\vspace{-2mm}
\end{table*}

\begin{figure*}[!t]
\centering
\subfloat[Amazon2M: {\nameN}]{\includegraphics[width=0.25\textwidth]{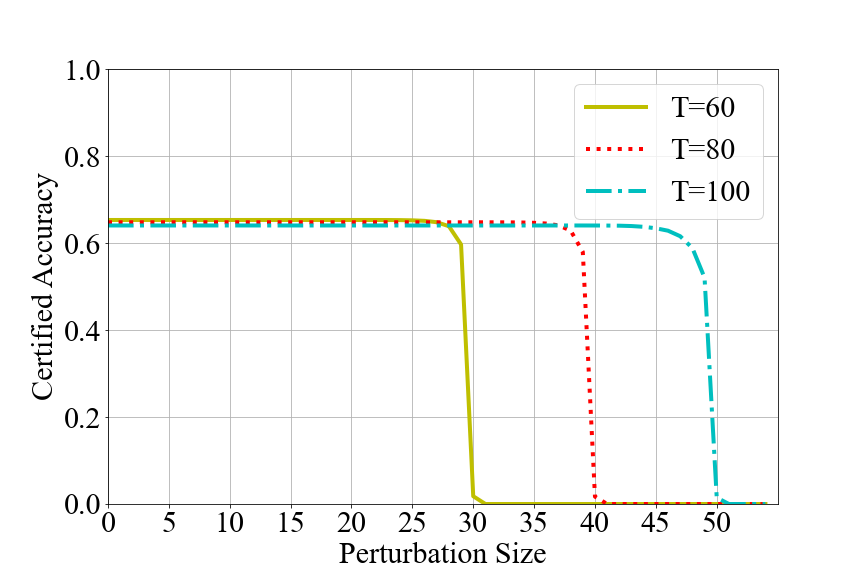}}
\subfloat[Amazon2M: {\nameE}]{\includegraphics[width=0.25\textwidth]{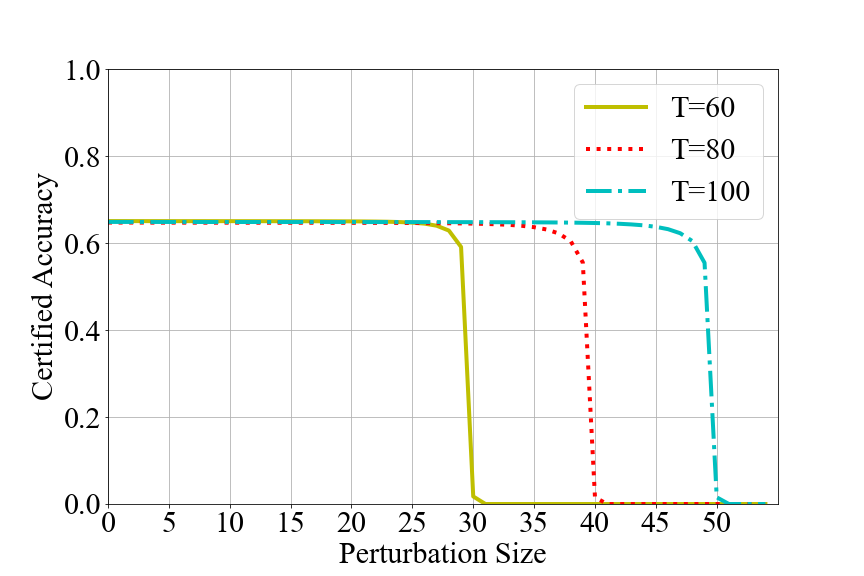}}
\subfloat[Big-Vul: {\nameN}]{\includegraphics[width=0.25\textwidth]{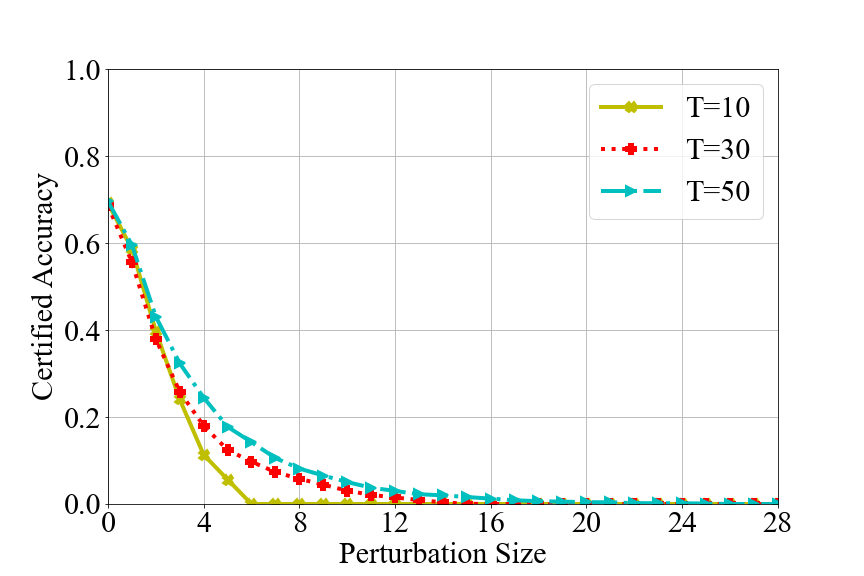}}\hfill
\subfloat[Big-Vul: {\nameE}]{\includegraphics[width=0.25\textwidth]{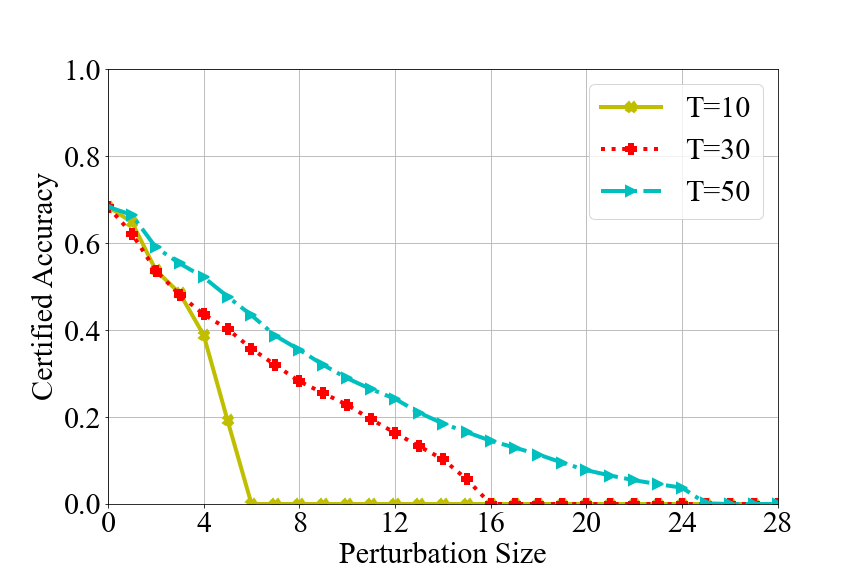}}\\
\caption{Certified node/graph accuracy of {\name} w.r.t. the number of subgraphs $T$ on Amazon2M and Big-Vul.
}
\label{fig:realworld}
\vspace{-2mm}
\end{figure*}

\noindent {\bf Comparing the computational complexity {\bf and runtime} of the defenses:} 
Table~\ref{computation-cost-training-testing} shows the Big-O complexity of the compared defenses and the base GNN for training and certification/testing. We only show the factor on the augmented graphs as other factors are similar in all methods.   
We observe that: 1) {As $N_1 \simeq T_e \simeq T_n \simeq T$,  all defenses have close Big-O complexity for training (except GNNCert-EN)}. 
2) GNNCert-E has a similar training and certification complexity as ours, but it can only defend against the edge manipulation. 3) Bi-RS is the least efficient for certification due to needing vast samples to ensure high confidence guarantees. 
{4) All defenses are $T$ slower than the base GNN in training and certification.}
{We also record the runtime in Table~\ref{tbl:exp_time} and these defenses' runtime matches the observations from the Big-O analysis.}

\section{Evaluations on Real-World Graph Datasets}
\label{sec:realworld}

In this section, we will evaluate {\name} on two real-world  graph datasets, i.e., Amazon2M co-purchasing dataset \cite{clusterGCN} for node classification and Big-Vul code vulnerability dataset \cite{big_vul} for graph classification.  

\subsection{Experimental Settings}

Amazon2M dataset is a network representation of products from Amazon, where nodes signify products, and edges indicate two products are frequently purchased together. 
This dataset consists of 2,449,029 nodes and 61,859,140 edges and is used for node classification --- each node has 100 features and is labeled as one of 47 products and the task is to classify products. 
We divide nodes into 30\% for training, 20\% for validation, and 50\% for testing. 

Big-Vul is widely-used code vulnerability dataset, which comprises extensive source code vulnerabilities extracted from 348 open-source C/C++ GitHub projects, spanning from 2002 to 2019. It contains 188,636 C/C++ functions, 
including 10,900 vulnerable ones (covering 91 vulnerability types), 
and 7,203 benign ones. Following the recent work \cite{chu2024graph}, 
 we built code graphs by taking code statements as nodes, control-flow or data-flow
dependencies as edges and utilizing GraphCodeBERT’s\cite{guo2020graphcodebert} token embedding layer to initialize node features.  Afterwards, we labeled these code graphs as vulnerable or benign,   
and formed vulnerability detection problem as a binary graph classification task \cite{chu2024graph}. We divide the graphs into 80\% for training, 10\% for validation, and 10\% for testing.

 We use GCN as the base GNN in {\name} (MD5 as the hash function) to train Big-Vul, and cluster-GCN  \cite{clusterGCN} (a more computation- and memory- efficient variant of GCN) as a base GNN in {\name} to train the large-scale Amazon2M.

\subsection{Experimental Results}

{\bf Runtime and accuracy:} Table \ref{tab:real-graph} shows the  training and test time, and test accuracy of {\name} and the base GNN on Amazon2M ($T=80$) and Big-Vul ($T=30$).   
We observe that: 1) Test accuracies of {\name} and base GNN are close, indicating {\name} maintains the utility in these real-world graphs; 2) {\name} is about $T$ times slower than the base GNN, again consistent with the Big-O analysis in Table \ref{computation-cost-training-testing}.

\noindent {\bf Certified accuracy:} Figure~\ref{fig:realworld} reports the certified accuracies of {\name} on the two datasets. The results validate that {\name} is also an effective defense for safeguarding real-world GNN applications against graph perturbations. 
For instance, {\nameN} can tolerate up to 50 edges and nodes on Amazon2M with arbitrary perturbations; and {\name}-E can defend against 24 arbitrarily perturbed edges on Big-Vul.

\begin{table}[!t]
\caption{Runtime and test accuracy of {\name} and the base undefended GNN on Amazon2M (T=80) and Big-Vul (T=30). As {\nameE} and  {\nameN} have close runtime and test accuracy, we simplicity use {\name} for brevity.}
    \centering
    \footnotesize
    \addtolength{\tabcolsep}{-3pt}
    \begin{tabular}{|c|c|c|c|c|}
    \hline
         {\bf Dataset}&{\bf Method}&{\bf Train time/epoch} &{\bf Test time} & {\bf Test acc.}  \\
         \hline
         \multirow{2}{*}{\bf Amazon2M}& {\bf Cluster-GCN}&3.2s&1.1s&0.72\\
         &{\name}&287s&107s&0.68\\ \hline
         \multirow{2}{*}{\bf Big-Vul}& {\bf GCN}&27.8s&2.3s&0.70\\
         &{\name}&827s&65s&0.69\\ \hline
    \end{tabular}
    \label{tab:real-graph}
\end{table}

\section{Discussions and Limitations}
\label{sec:discussion}

{\bf {\name}'s performance with larger/powerful GNNs:}
The certified robustness result is determined by the gap between the most votes (for the correct label) and second-most votes obtained by a GNN on subgraphs.
Hence, a GNN making more accurate predictions on subgraphs exhibits better certified robustness. A more powerful/larger GNN may achieve better robustness, as it is expected to provide more accurate predictions. For instance, we test a 6-layer ResGCN \cite{li2019deepgcns} on Pubmed, and its certified accuracy is 2\% higher than that of the used 3-layer GCN under the same perturbation size.

\vspace{+0.05in}
\noindent {\bf Node-centric vs. edge-centric {\name}:} When defending against 
node perturbations, {\nameN} outperforms {\nameE} because {\nameN}  guarantees an infinite number of perturbed edges, whereas {\nameE}'s guarantee is bounded. 
However, when defending against edge manipulation attacks, it is hard to say which method is better, as we cannot ascertain which $M$ value (in Equation \ref{eqn:cpz_edge} for {\nameE} and Equation \ref{eqn:cpz_node} for {\nameN}) is larger, considering the two methods use distinct graph division strategies.

\vspace{+0.05in}
\noindent {\bf {\name} may be ineffective against training-time attacks on GNNs:} 
The proposed {\name} is primarily designed to robustify a \emph{clean} GNN model against \emph{test-time} attacks. Its effectiveness relies on the gap between the most-votes and second-most-votes be sufficiently large. However, if the GNN model is poisoned \cite{wang2023turning} or backdoored \cite{zhang2021backdoor,xi2021graph,yang2024distributed} during training (e.g., a compromised model downloaded from the internet), and our defense is unaware of it, the derived bound may be weakened as the poisoned/backdoored model could reduce this gap. We will leave this as future work.

\vspace{+0.05in}
\noindent {\bf {\name} may be ineffective on graph similarity or matching tasks:} 
{\name} takes a \emph{single} graph as input. 
However, certain security applications involve a pair of graphs, e.g., GNN-based (binary or source) code similarity analysis \cite{he2022illuminati,li2019graph,liu2023learning,kim2022revisiting,gao2024sigmadiff,marcelli2022machine} takes as input a pair of  (e.g., control-flow) graphs generated from the code, and they can be formalized as a graph similarity/matching problem. In this context, an adversary is able to manipulate the source code such that the respective code graph could be largely changed (e.g., many node indexes and edges are changed), while maintaining the code functionality. This attack would make it hard to obtain the one-to-one correspondence between subgraphs generated from the two source graphs. 
Hence, it is difficult to directly apply {\name} for certification in this setting.

\vspace{+0.05in}
\noindent {\bf Inefficiency of {\name} to large-scale graphs:} As shown in Table \ref{computation-cost-training-testing}, our {\name} has a training and certification complexity that is $T$ times of the base GNN's. This overhead becomes significant when applying {\name} to large graphs (see Table \ref{tab:real-graph}). 
We acknowledge it is important future work to speed up {\name}, while holding its theoretical results.

\section{Related Work}
\label{sec:related}

\noindent {\bf Adversarial attacks on GNNs:} Various  works~\cite{zugner2018adversarial,dai2018adversarial,wu2019adversarial,wang2019attacking,xu2019topology,sun2020adversarial,zhang2021backdoor,wan2021adversarial,zhang2022projective,ma2020towards,mu2021a,wang2022bandits,wang2023turning,chenunderstanding,ju2023let,wang2024efficient} show GNN classifiers are vulnerable to adversarial perturbations. Given a GNN (node/graph) classifier and a graph, an attacker could inject a few nodes~\cite{sun2020adversarial,ju2023let}, slightly modify the graph structure~\cite{zugner2018adversarial,dai2018adversarial,xu2019topology}, and/or perturb node features~\cite{zugner2018adversarial} such that the classifier makes wrong predictions for the perturbed graph (in graph classification) or target nodes (in node classification).    
For instance, \cite{sun2020adversarial} utilizes reinforcement learning techniques to design node injection attacks, while \cite{dai2018adversarial} designs graph perturbation attacks to both graph and node classification.  
Most attacks require the attacker fully/partially knows the GNN model (e.g., parameters, architecture), while \cite{mu2021a,wang2022bandits} relaxing this to  only have black-box access, i.e., only query the GNN model API. For example,  \cite{wang2022bandits} formulates this black-box attack to GNNs as an online optimization with bandit feedback. The original problem is NP-hard and they then propose an online attack based on (relaxed) bandit convex optimization which is proven to be {sublinear} to the query number. 

\vspace{+0.05in}
\noindent {\bf Defenses against attacks on GNNs:}
Many empirical defenses~\cite{wu2019adversarial,xu2019topology,zhu2019robust,entezari2020all,tao2021adversarial,zhao2021expressive} were proposed against the adversarial attacks on GNNs. However, these defenses do not have guaranteed performance under the worst-case setting, and were soon broken by adaptive/stronger  attacks~\cite{mujkanovic2022defenses}. Hence, we focus on certified defense in this work. 

Certified defenses~\cite{jin2020certified,jia2020certified,bojchevski2020efficient,wang2021certified,xia2024gnncert,lai2023nodeawarebismoothingcertifiedrobustness} design robust GNNs that guarantee the same predicted label on clean and perturbed graphs, when the perturbation size (e.g., number of perturbed edges, node features, or injected nodes) on the graph is bounded. 
\cite{bojchevski2020efficient} and~\cite{wang2021certified} generalized randomized smoothing (RS)~\cite{lecuyer2019certified,cohen2019certified,hong2022unicr}, the state-of-the-art certified defense against adversarial perturbations on the image domain, to the graph domain and certify any GNN against the edge perturbation. \cite{lai2023nodeawarebismoothingcertifiedrobustness} designs a node-aware Bi-RS certified defense against the node injection attack and achieve the state-of-the-art. 
Further, \cite{xia2024gnncert} extended randomized ablation~\cite{levine2020robustness}, a voting-based defense for image models, to build provably robust graph classifier against the node feature perturbation, edge perturbation, and combined edge and feature perturbations.   

However, all existing certified defenses face several  limitations. First, except \cite{xia2024gnncert} against edge and node feature perturbation, all can only 
certify 
one type of perturbation, e.g., edge perturbation. 
Second, they are only applied 
for a particular task such as node classification or graph classification, but not both. Adapting these defenses for both tasks would yield unsatisfactory guarantees as shown in our results in Section~\ref{sec:eval}. 
Third, their robustness guarantees are not 100\% (except \cite{xia2024gnncert}), implying  the guarantee could be inaccurate with certain probability. Our {\name} addresses all these limitations.

\vspace{+0.05in}
\noindent {\bf Voting-based certified defenses:} 
Voting is a versatile ensemble method in machine learning (ML)~\cite{dietterich2000ensemble} primarily for classification, and each method defines the voter for its own purpose. 
Recently, voting has been also used to robustify ML models against adversarial attacks, including adversarial image perturbation~\cite{levine2020deep},  graph perturbation~\cite{xia2024gnncert,yang2024distributed,li2025provably}, image patch perturbation~\cite{levine2020randomized,xiang2021patchguard}, text perturbation~\cite{pei2023textguard,zhang2024text}, and data poisoning attacks~\cite{jia2021intrinsic,jia2022certified}.
The key steps of this type of defense are: divide an input data (e.g., an image, a graph, or a sentence) into a set of sub-data, build a voting classifier to predict all sub-data (each prediction is a vote), and derive the robustness guarantee for the voting classifier. 
The essential requirement is to ensure only a bounded number of predictions are changed with a bounded adversarial perturbation. 
The key difference among these defenses is they create problem-dependent sub-data and voters for the majority voting.

\section{Conclusion}
\label{sec:conclusion}

We study the robustness of GNNs against adversarial attacks. Particularly, we develop {\name}, the first certified defense for GNNs against arbitrary perturbations (on nodes, edges, and node features) with deterministic guarantees. 
{\name} designs novel graph division strategies and  leverages the message-passing mechanism in GNNs for deriving the robustness guarantee. 
The universality of  {\name} makes it encompass existing certified defenses as special cases.  
Evaluation results validate  {\name}'s effectiveness and efficiency against arbitrary perturbations on GNNs and superiority over the state-of-the-art certified defenses. 

\vspace{+0.05in}
\noindent {\bf Acknowledgement:} We sincerely thank all the anonymous reviewers and our shepherd for their valuable feedback and constructive comments. We also extend our gratitude to Kexin Pei and Yuede Ji for providing the real-world code vulnerability dataset and conducting the evaluations.
This work is partially supported by the National Science Foundation (NSF) under grant Nos. ECCS-2216926, CNS-2241713, CNS-2331302, CNS-2339686, and the Cisco Research Award. 
Any opinions, findings and conclusions or recommendations expressed in this material are those of the author(s) and do not necessarily reflect the views of the funding agencies.

\section{Ethics Considerations}

This research strictly adheres to ethical guidelines and responsibilities, ensuring compliance with established standards.

\vspace{+0.1in}
\noindent {\bf 1) Identification of Stakeholders}
\vspace{+0.05in}

\noindent {\bf Researchers:} Those advancing the field by building upon this work, focusing on both defending GNNs against adversarial attacks and exploring trustworthy GNNs (e.g., against training-time poisoning attacks and both training- and test-time backdoor attacks) in general.

\noindent {\bf Developers and Practitioners of AI Systems:} Individuals and organizations implementing or applying provably robust GNNs in real-world graph-related applications such as fraud detection in social networks, web, online auction networks, intrusion detection, and software vulnerability detection.

\noindent {\bf End-users:} People interacting with GNN-powered systems, including users of social networks, recommender systems, or financial platforms.

\noindent {\bf Society at Large:} Individuals impacted by ethical considerations and risks associated with deploying AI technologies, especially in domains leveraging GNNs (e.g., social networks, healthcare, finance).

\vspace{+0.1in}
\noindent {\bf 2) Potential Risks for Stakeholders and Mitigations}
\vspace{+0.05in}

\noindent {\bf For Researchers.}
\emph{Potential Risk:} Adversaries may develop novel attacks that surpass the guaranteed bounds of the considered threat model (e.g., perturbations beyond the certified perturbation size). \emph{Mitigation:} With larger perturbations on graph data, those perturbed graphs might have significant differences with normal graphs. Therefore, researchers can leverage detection methods, such as structural-similarity based methods, to identify the perturbed graphs. Researchers can also collaborate with ethics experts to ensure that the research aligns with best practices for responsible AI development.

\noindent {\bf For Developers and Practitioners.}
\emph{Potential Risk:}  The proposed defense method may not generalize well to other graph learning applications that are different from the considered applications.
\emph{Mitigation:} Comprehensive empirical validation across diverse graph datasets and real-world scenarios ensures robustness. Clear communication of limitations will help developers manage risks effectively.

\noindent {\bf For End-users.}
\emph{Potential Risk:}  Robust GNN mechanisms might inadvertently compromise data privacy or produce biased outcomes.
\emph{Mitigation:} Incorporating privacy-preserving (such as differential privacy and cryptographic methods) and fair training techniques enhances data security and fairness.

\noindent {\bf For Society.}
\emph{Potential Risk:} Misuse of robust GNNs in critical domains (e.g., healthcare, finance) could exacerbate social inequities, privacy breaches, or manipulation of vulnerable populations.
\emph{Mitigation:} Balancing AI security advancements with societal considerations (including fairness, transparency, and accountability) mitigates potential harm. Ethical implications for vulnerable populations will be addressed, prioritizing societal well-being.

\vspace{+0.1in}
\noindent {\bf 3) Considerations Motivating Ethical-Related Decisions}
\vspace{+0.05in}

\noindent {\bf Research Goal:} The primary objective is to enhance the robustness of GNNs against adversarial attacks while minimizing potential harm to stakeholders. Defense strategies are designed to be both practical and ethical.

\noindent {\bf Benefits and Harms:} \emph{Benefits:} Improved robustness 
of GNN systems reduces risks of adversarial manipulation and protecting users. \emph{Harms:} Potential empowerment of malicious actors and overestimating the effectiveness of defense methods.

\noindent {\bf Rights:} We are particularly concerned with privacy rights, as adversarial attacks can sometimes expose sensitive data or violate individuals' privacy. Our defense strategies aim to mitigate such risks, promoting the ethical use of GNNs while safeguarding individuals’ rights.

\vspace{+0.1in}
\noindent {\bf 4) Awareness of Ethical Perspectives}
\vspace{+0.05in}

\noindent We are aware that different members of the research community may hold differing views on the ethical implications of trustworthy AI. Some may prioritize transparency in revealing attack strategies to help build better defenses, while others may argue that such knowledge could be misused. In line with the principles of responsible AI research, we have opted to emphasize defense over offense, focusing on methods that mitigate risk without creating new avenues for harm.

\section{Open Science}

In compliance with the Open Science Policy, we have made our code, pretrained models, and data openly accessible at \url{https://github.com/JetRichardLee/AGNNCert}. Additionally, all artifacts have been published on the Zenodo platform \url{https://zenodo.org/records/14737141} to facilitate the reproduction of the research described in the paper. 

Through these efforts, we aim to contribute to the broader scientific community while upholding the highest standards of ethical conduct.

\bibliographystyle{plain}
\bibliography{refs}

\appendix
\section{Proof of Theorem~\ref{thm:suffcond}}
\label{app:suffcond}

We prove for node classification and it is identical for graph classification. 

Recall $y_a$ and $y_b$ are respectively the class with the most vote $c_{y_a}$ and with the second-most vote $c_{y_b}$ on predicting the target node $v$ in the subgraphs $\{G_i\}'s$. Hence, 
\begin{align}
&    c_{y_{a}}-\mathbb{I}(y_{a}>y_{b})\geq c_{y_{b}} \label{eqn:16} \\
 &   c_{y_{b}}-\mathbb{I}(y_{b}>y_{c})\geq c_{y_{c}}, \forall y_c \in \mathcal{Y}\setminus\{y_{a}\} \label{eqn:17}
\end{align}
where $\mathbb{I}$ is the indicator function, and we pick the class with a smaller index when there exist ties. 

Further, on the perturbed graph $G'$ after the attack, the vote $c_{y_{a}}'$ of the class $y_a$ and vote $c_{y_{c}}'$ of any other class $y_{c}\in \mathcal{Y}\setminus \{y_{a}\}$ satisfy the below relationship: 
\begin{equation}
\label{eqn:18}
c_{y_{a}}'\geq c_{y_{a}} - \sum_{i=1}^{T}\mathbb{I}(f(G_{i})_v\neq f(G'_{i})_v) 
\end{equation}
\begin{equation}
\label{eqn:19}
c_{y_{c}}'\leq c_{y_{c}} + \sum_{i=1}^{T}\mathbb{I}(f(G_{i})_v\neq f(G'_{i})_v)
\end{equation}

To ensure the returned label by the voting node classifier $\bar{f}$ does not change, i.e., $\bar{f}(G)_v = \bar{f}(G')_v, \forall G'$, we must have:
\begin{equation}
\label{eqn:20}
c_{y_{a}}'\geq c_{y_c}'+\mathbb{I}(y_{a}>y_{c}),\forall y_{c}\in \mathcal{Y}\setminus \{y_{a}\}
\end{equation}

Combining with Eqns \ref{eqn:18} and \ref{eqn:19}, the sufficient condition for Eqn~\ref{eqn:20} to satisfy is to ensure: 
\begin{equation}
c_{y_{a}} - \sum_{i=1}^{T}\mathbb{I}(f(G_{i})_v\neq f(G'_{i})_v)  \geq 
c_{y_{c}} + \sum_{i=1}^{T}\mathbb{I}(f(G_{i})_v\neq f(G'_{i})_v)
\end{equation}
Or, 
\begin{equation}
c_{y_{a}}\geq c_{y_{c}} + 2\sum_{i=1}^{T}\mathbb{I}(f(G_{i})_v\neq f(G'_{i})_v)+\mathbb{I}(y_{a}>y_{c}).
\end{equation}

Plugging Eqn~\ref{eqn:17}, we further have this condition:
\begin{equation}
\label{eqn:23}
c_{y_{a}} \geq c_{y_{b}}-\mathbb{I}(y_{b}>y_{c})+ 2\sum_{i=1}^{T}\mathbb{I}(f(G_{i})_v\neq f(G'_{i})_v)+\mathbb{I}(y_{a}>y_{c})
\end{equation}
We observe that:
\begin{equation}
\label{eqn:24}
\mathbb{I}(y_{a}>y_{b})\geq \mathbb{I}(y_{a}>y_{c})-\mathbb{I}(y_{b}>y_{c})
,\forall y_{c}\in \mathcal{Y}\setminus \{y_{a}\}\end{equation}
Combining Eqn~\ref{eqn:24} with Eqn~\ref{eqn:23}, we have:
\begin{equation}
c_{y_{a}} \geq c_{y_{b}}+2\sum_{i=1}^{T}\mathbb{I}(f(G_{i})_v\neq f(G'_{i})_v)+\mathbb{I}(y_{a}>y_{b})
\end{equation}
Let $M = {\lfloor c_{y_a}-c_{y_b}-\mathbb{I}(y_{a}>y_{b})\rfloor} / {2}$, hence 
$\sum\nolimits_{i=1}^{T}\mathbb{I}(f(G_{i})_v\neq f(G'_{i})_v) \leq M$.

\begin{figure*}[t]
    \centering
    \captionsetup[subfloat]{labelsep=none, format=plain, labelformat=empty}

    \subfloat[{\small (a) Edge-Centric Graph Division for Node Classification against edge deletion, node deletion and node feature manipulation}]{
    \includegraphics[width=\linewidth]{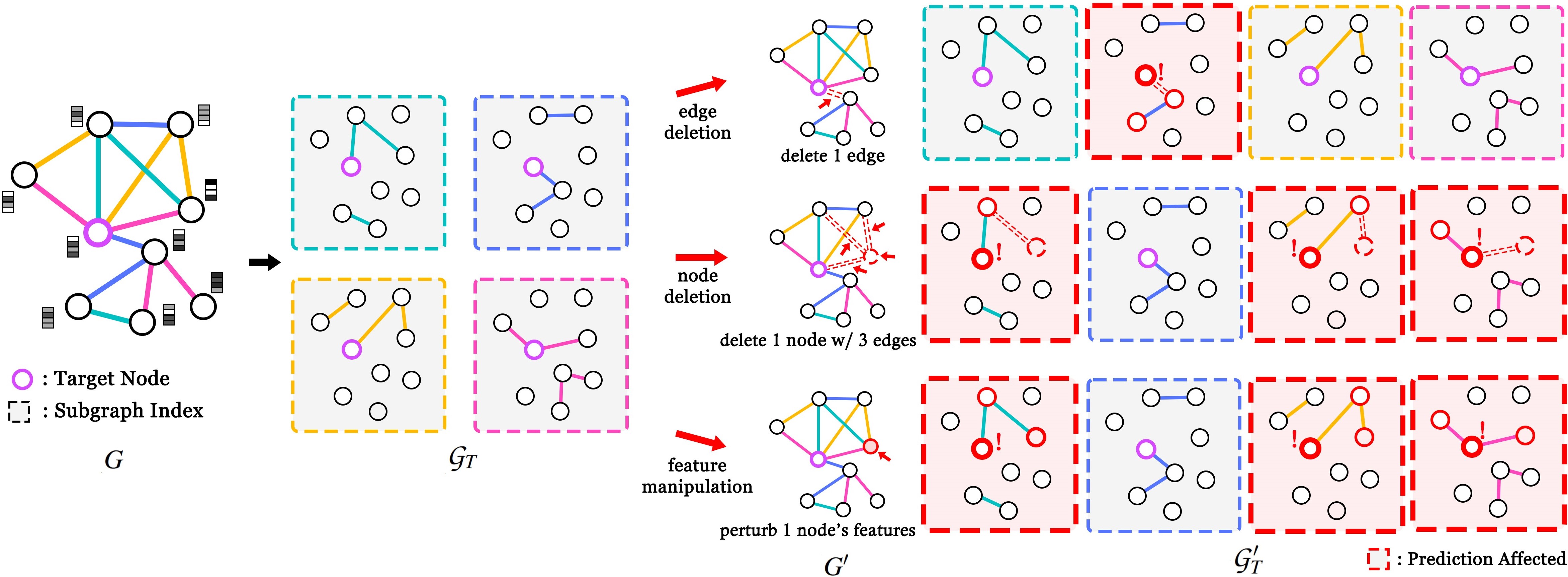}}
    \hspace{+10mm}
   
    \subfloat[{\small (b) Node-Centric Graph Division for Node Classification against edge deletion, node deletion and node feature manipulation}]{
    \includegraphics[width=\linewidth]{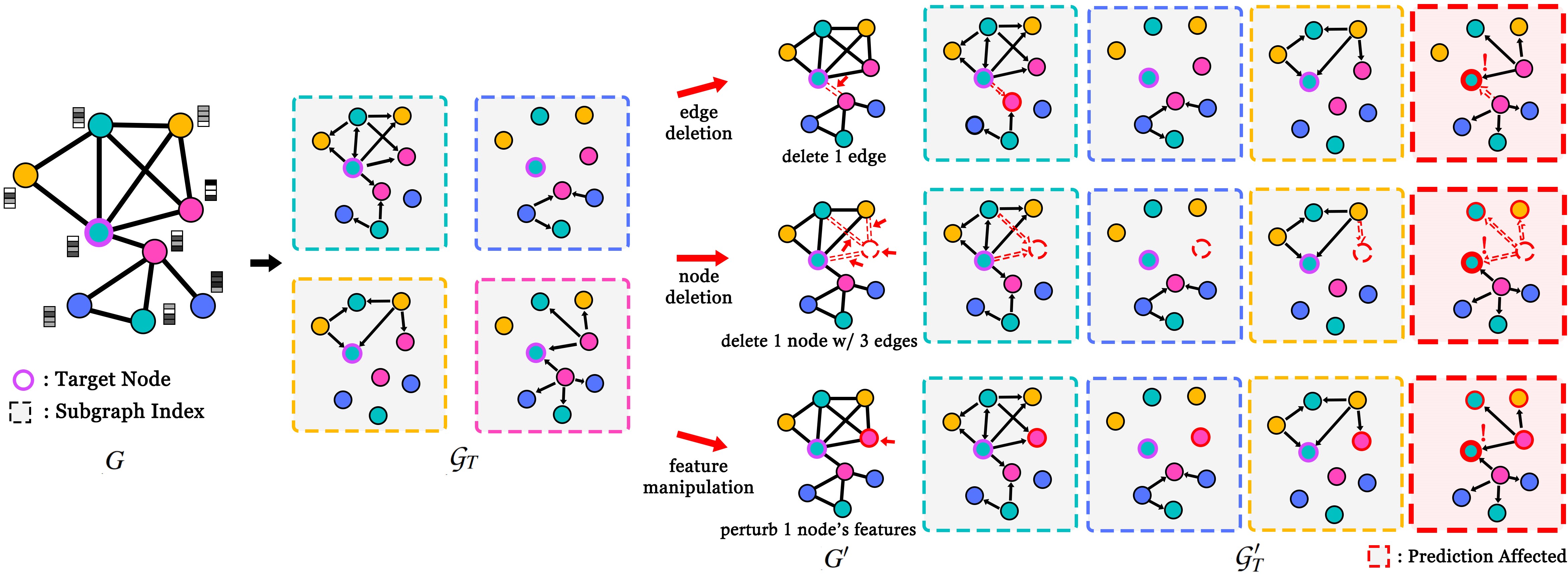}}
    \hspace{+10mm}
    \caption{Illustration of our edge-centric and node-centric graph division strategies for node classification against edge deletion, node deletion, and node feature manipulation. 
    {\bf To summarize:} 1 deleted edge  affects at most 1 subgraph prediction in both graph division strategies. In contrast, 1 deleted node with, e.g., $3$ incident edges can affect at most 3 subgraph predictions with edge-centric graph division, but at most 1 subgraph prediction with node-centric graph division.
    }
    \label{fig:subgraphs_NC_more}
   \vspace{-2mm}
\end{figure*}

\begin{figure*}[t]
    \centering
    \captionsetup[subfloat]{labelsep=none, format=plain, labelformat=empty}

    \subfloat[{\small (a) Edge-Centric Graph Division for Graph Classification against edge manipulation, node manipulation and feature manipulation}]{
    \includegraphics[width=0.9\linewidth]{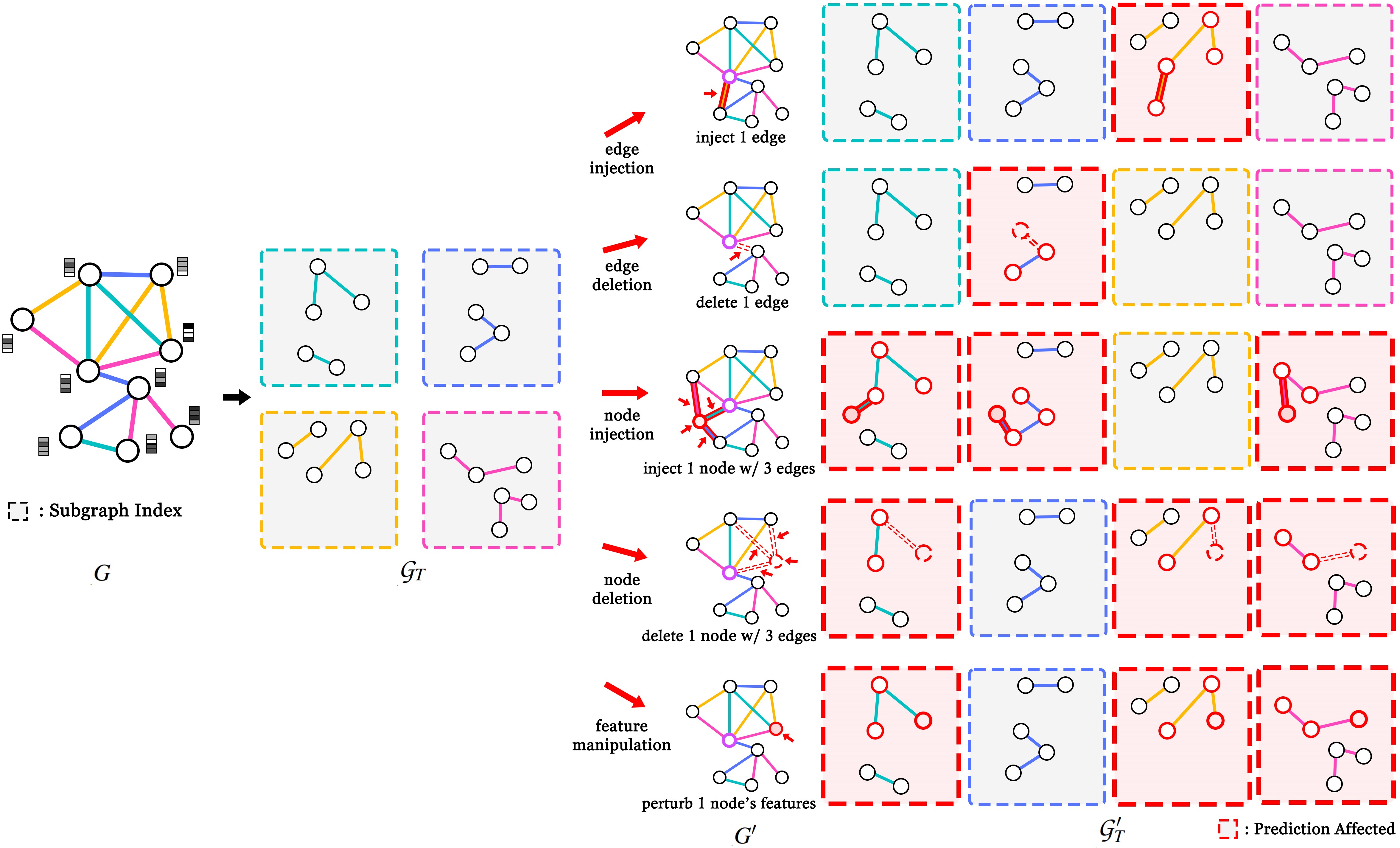}}
    \hspace{+20mm}
    
    \subfloat[{\small (b) Node-Centric Graph Division for Graph Classification against edge manipulation, node manipulation and feature manipulation}]{
    \includegraphics[width=0.9\linewidth]{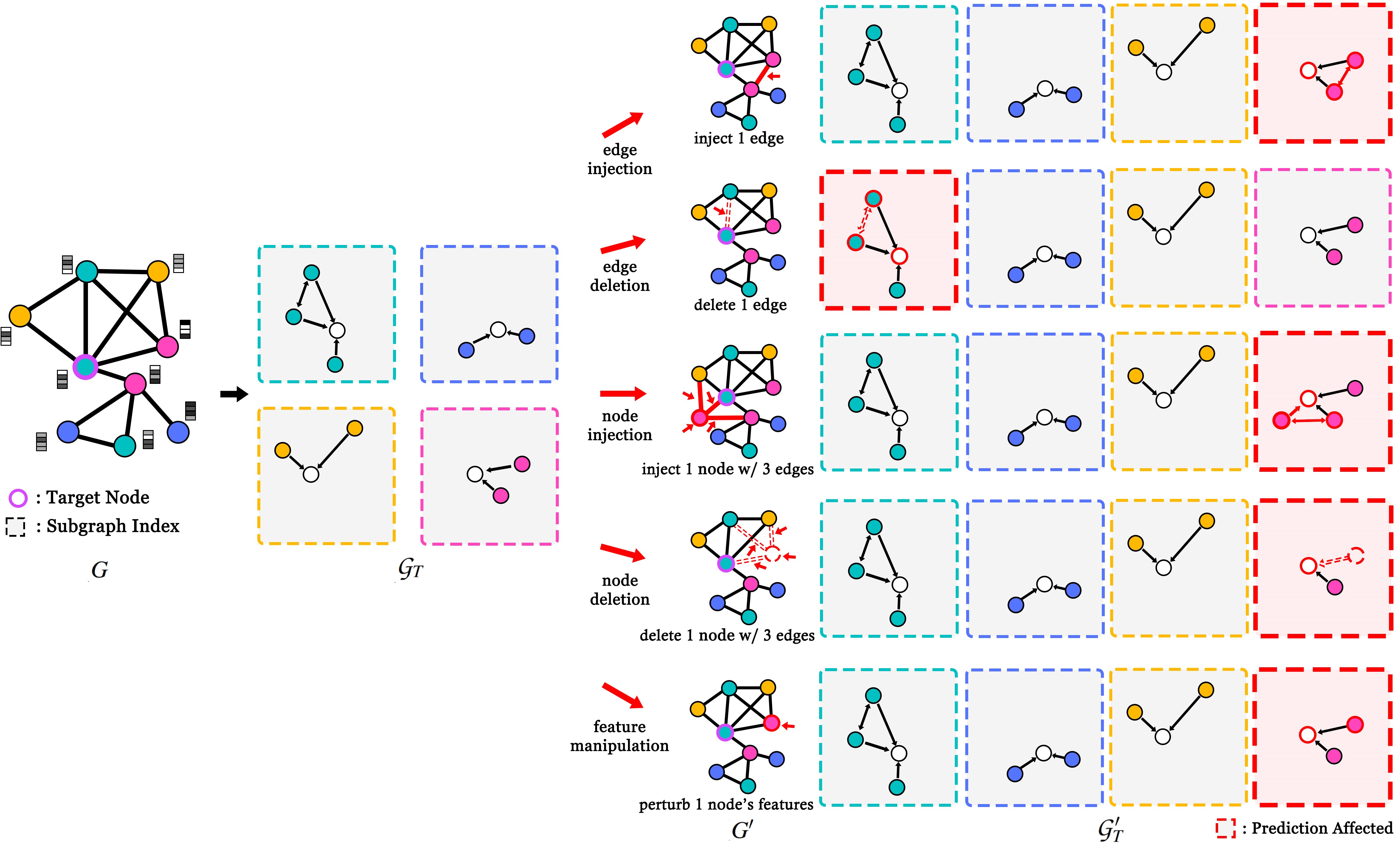}}
    \vspace{-2mm}
    \caption{Illustration of our edge-centric and node-centric graph division strategies for graph classification. The conclusion are similar to those for node classification.}
    \label{fig:subgraphs_GC}
    \vspace{-4mm}
\end{figure*}

\begin{figure*}[!t]
\centering
\subfloat[Cora-ML]{\includegraphics[width=0.25\textwidth]{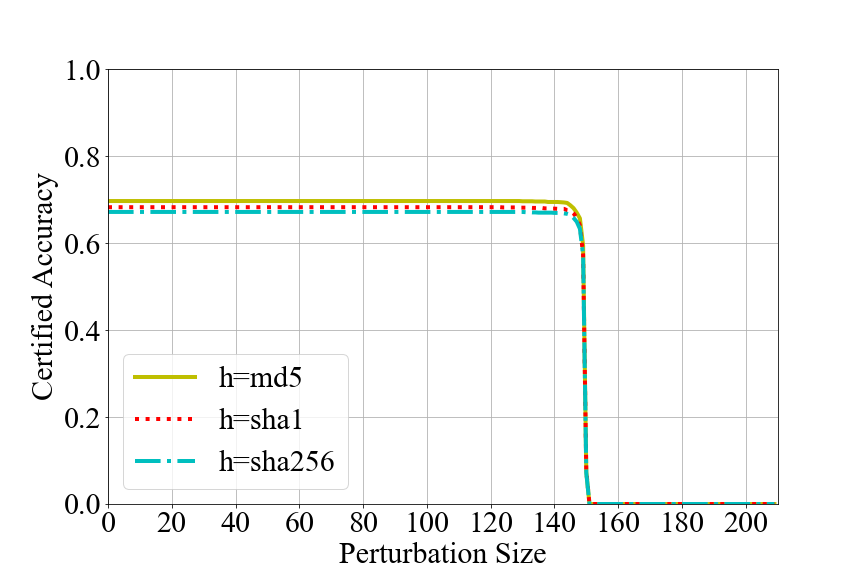}}\hfill
\subfloat[Citeseer]{\includegraphics[width=0.25\textwidth]{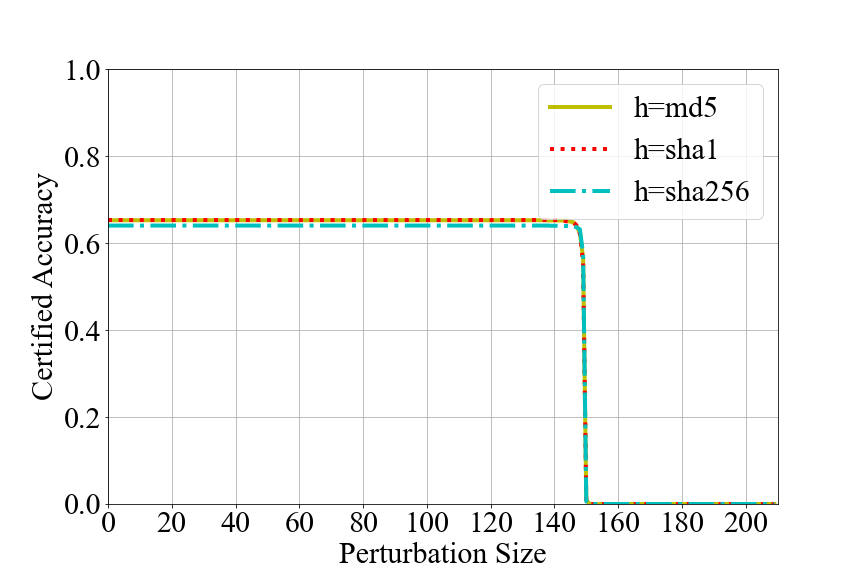}}\hfill
\subfloat[Pubmed]{\includegraphics[width=0.25\textwidth]{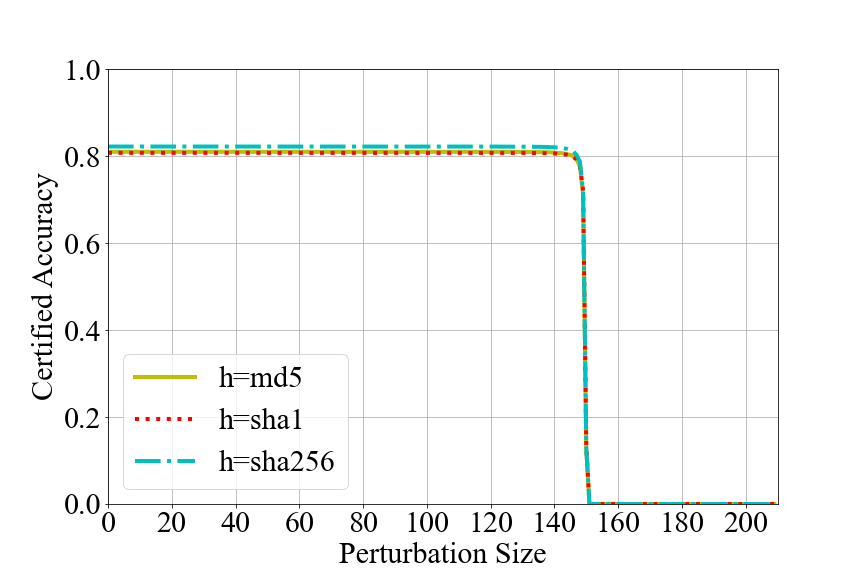}}\hfill
\subfloat[Amazon-C]{\includegraphics[width=0.25\textwidth]{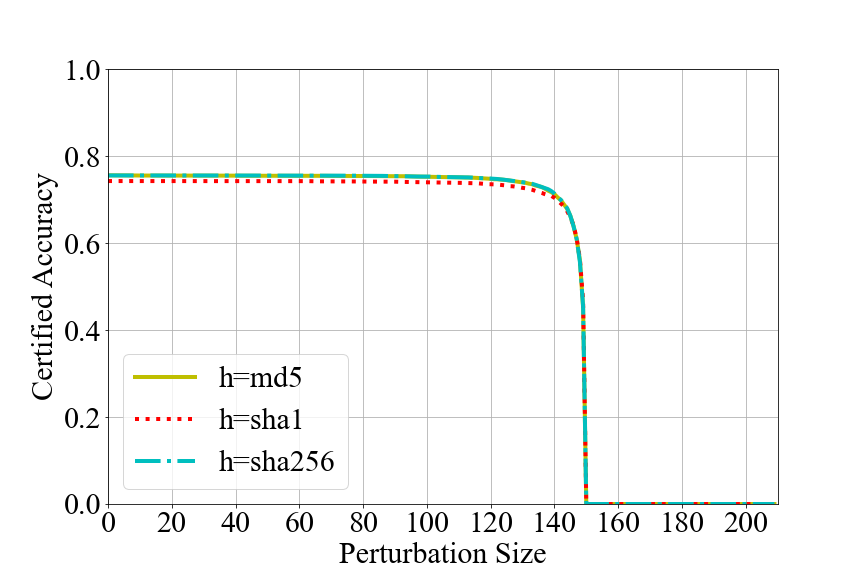}}\\
\vspace{-2mm}
\caption{Certified node accuracy of our {\nameE} w.r.t. the hash function $h$.}
\label{fig:node-EC-hash}
\vspace{-6mm}
\end{figure*}

\begin{figure*}[!t]
\centering
\subfloat[Cora-ML]{\includegraphics[width=0.25\textwidth]{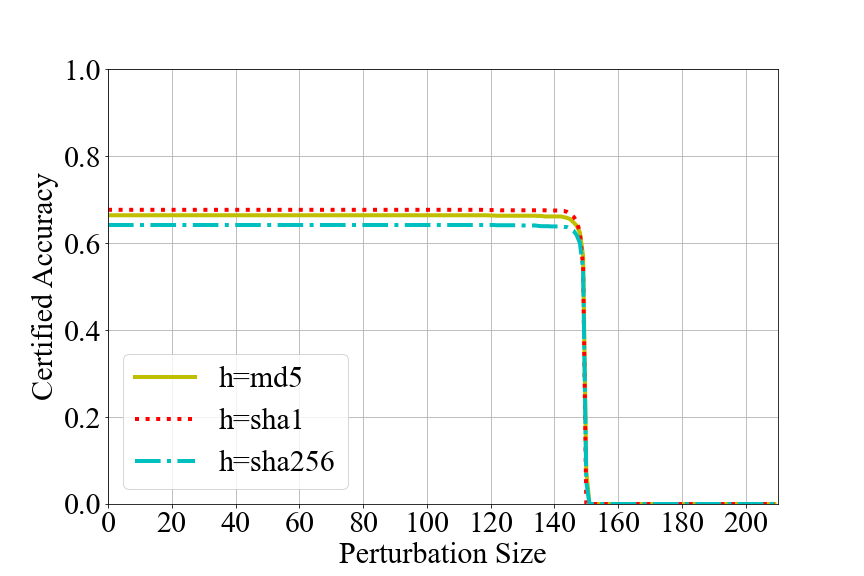}}\hfill
\subfloat[Citeseer]{\includegraphics[width=0.25\textwidth]{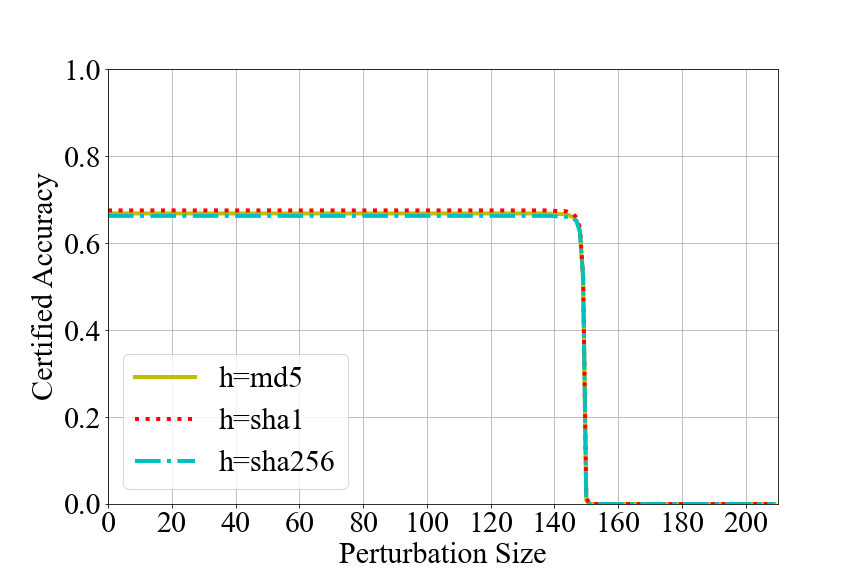}}\hfill
\subfloat[Pubmed]{\includegraphics[width=0.25\textwidth]{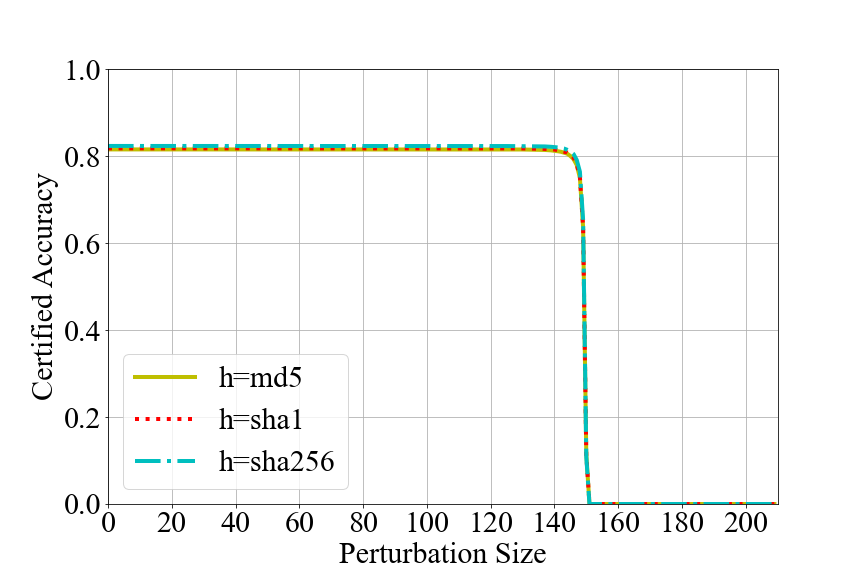}}\hfill
\subfloat[Amazon-C]{\includegraphics[width=0.25\textwidth]{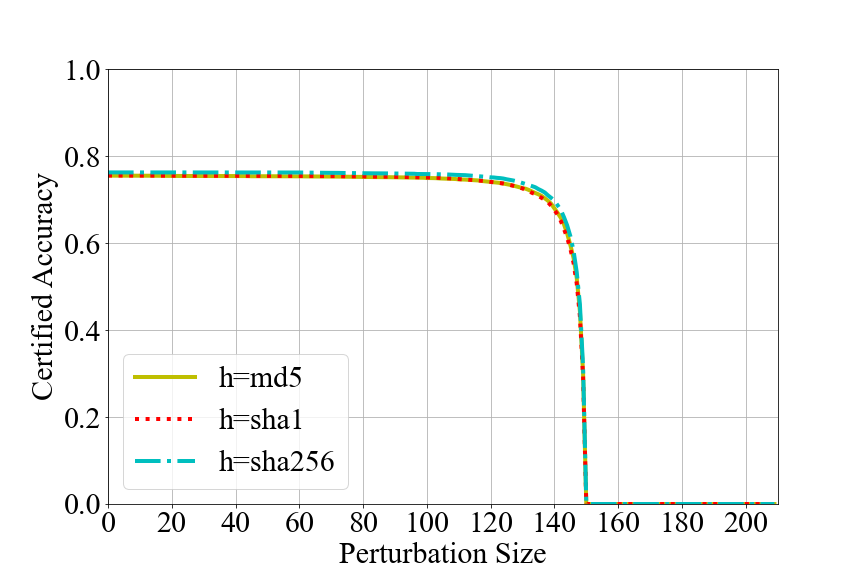}}\\
\vspace{-2mm}
\caption{Certified node accuracy of our {\nameN} w.r.t. the hash function $h$.}
\label{fig:node-NC-hash}
\vspace{-6mm}
\end{figure*}

\begin{figure*}[!t]
\centering
\subfloat[AIDS]{\includegraphics[width=0.25\textwidth]{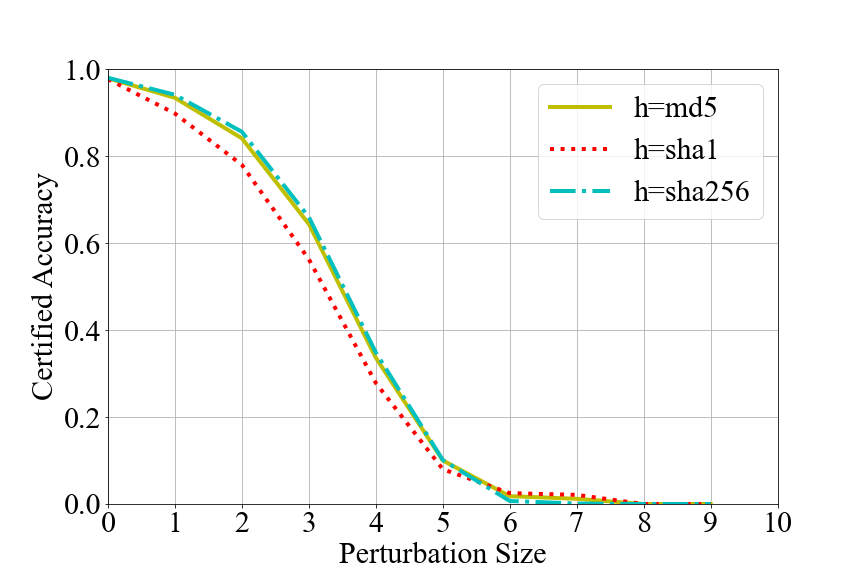}}\hfill
\subfloat[MUTAG]{\includegraphics[width=0.25\textwidth]{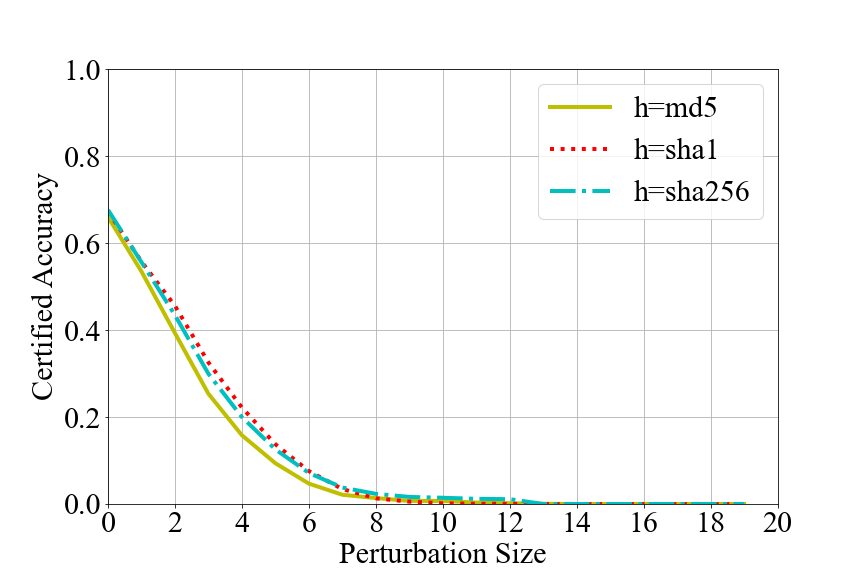}}\hfill
\subfloat[PROTEINS]{\includegraphics[width=0.25\textwidth]{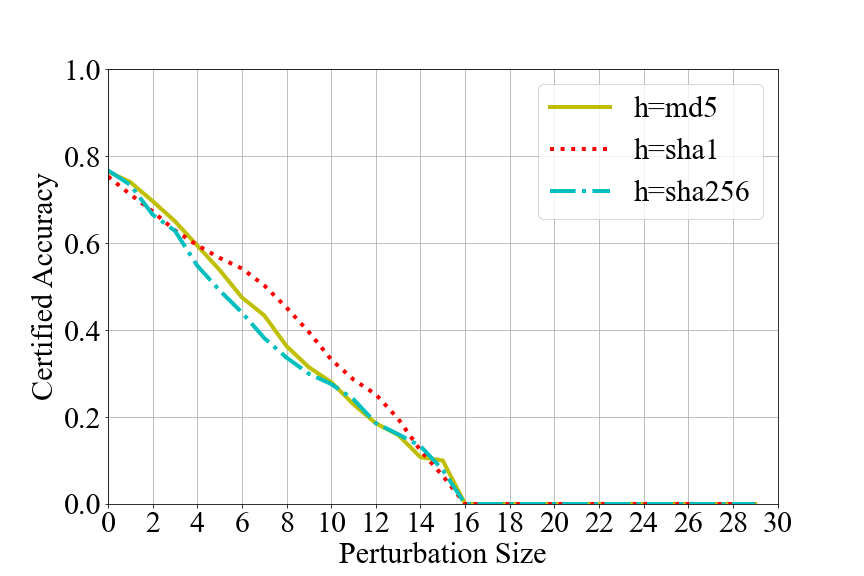}}\hfill
\subfloat[DD]{\includegraphics[width=0.25\textwidth]{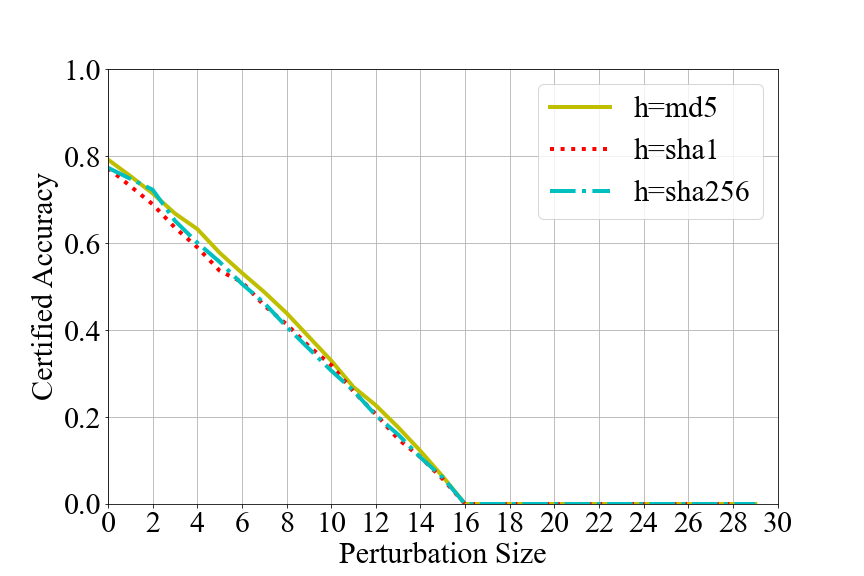}}\\
\vspace{-2mm}
\caption{Certified graph accuracy of our {\nameE} w.r.t. the hash function $h$.}
\label{fig:graph-EC-hash}
\vspace{-6mm}
\end{figure*}

\begin{figure*}[!t]
\centering
\subfloat[AIDS]{\includegraphics[width=0.25\textwidth]{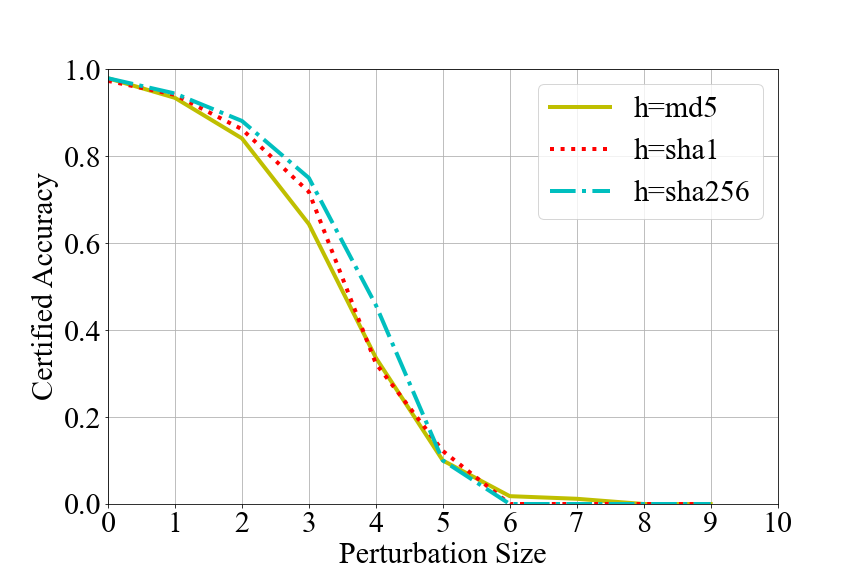}}\hfill
\subfloat[MUTAG]{\includegraphics[width=0.25\textwidth]{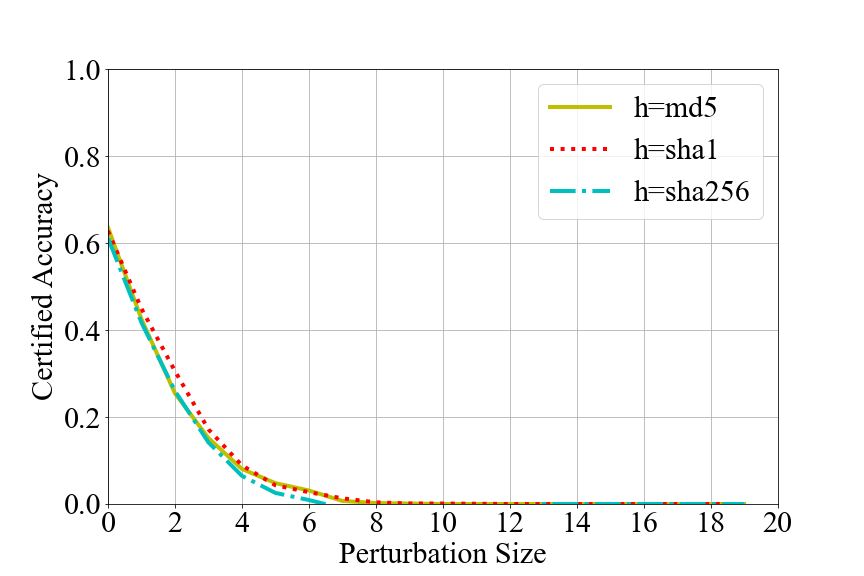}}\hfill
\subfloat[PROTEINS]{\includegraphics[width=0.25\textwidth]{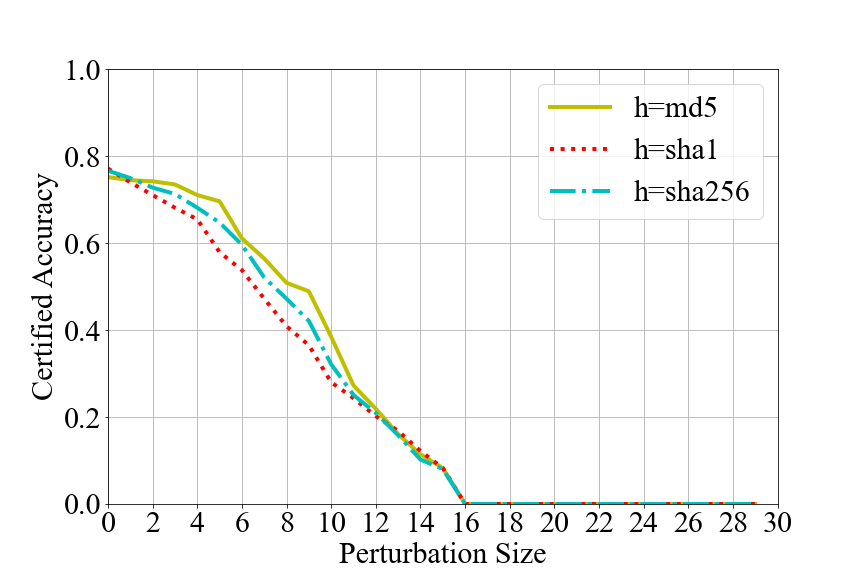}}\hfill
\subfloat[DD]{\includegraphics[width=0.25\textwidth]{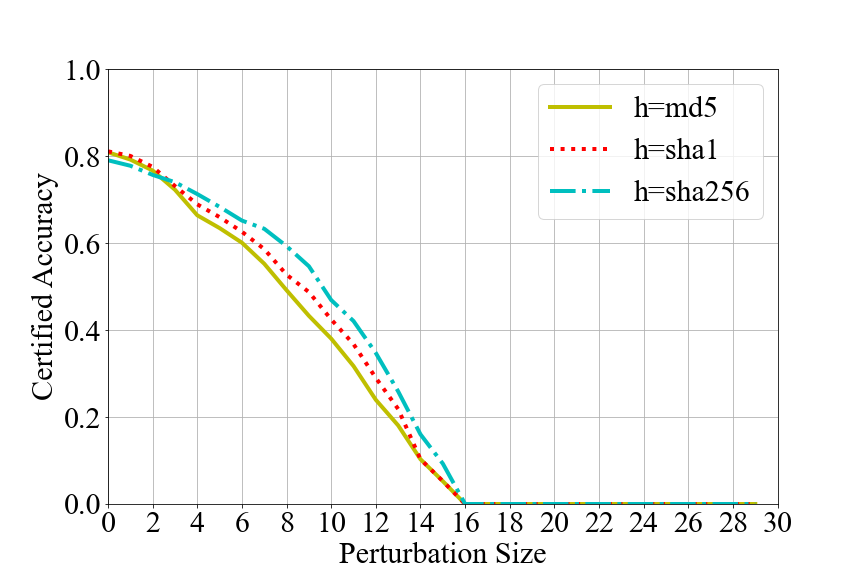}}\\
\vspace{-2mm}
\caption{Certified graph accuracy of our {\nameN} w.r.t. the hash function $h$.}
\label{fig:graph-NC-hash}
\vspace{-4mm}
\end{figure*}

\section{More Experimental Results}

\vspace{+0.05in}
\noindent 
Figure~\ref{fig:node-EC-hash}-Figure~\ref{fig:graph-NC-hash} show the certified node/edge accuracy of {\nameE} and  {\nameN} with different hash functions. 
We observe that our certified accuracy and certified perturbation size are almost the same in all cases. This reveals  {\name} is insensitive to hash functions, and \cite{xia2024gnncert} draws a similar conclusion. 

\vspace{+0.05in}
\noindent  Figures~\ref{fig:node-EC-T-GSAGE}-\ref{fig:graph-NC-T-GSAGE} show the results where we use GSAGE~\cite{hamilton2017inductive} as the base GNN classifier\footnote{During certification, we use all neighbors of a node,  instead of using randomly sampled nodes in the raw GSAGE, to maintain the divided subgraphs be deterministic.}, and Figures~\ref{fig:node-EC-T-GAT}-\ref{fig:graph-NC-T-GAT} the results where we use GAT~\cite{velivckovic2018graph} as the base GNN classifier. We can see they have similar certified node/graph classification at perturbation size as the model trained using GCN as the base classifier. 

\vspace{+0.05in}
\noindent Figure~\ref{fig:node-EC-w-wo}-Figure~\ref{fig:graph-EC-w-wo} show the certified node/graph classification with or without subgraphs for training the GNN classifier. We observe the certified accuracy can be much higher when the subgraphs are used for training. This is because certification also uses subgraphs.

\begin{figure*}[!t]
\centering
\subfloat[Cora-ML]{\includegraphics[width=0.25\textwidth]{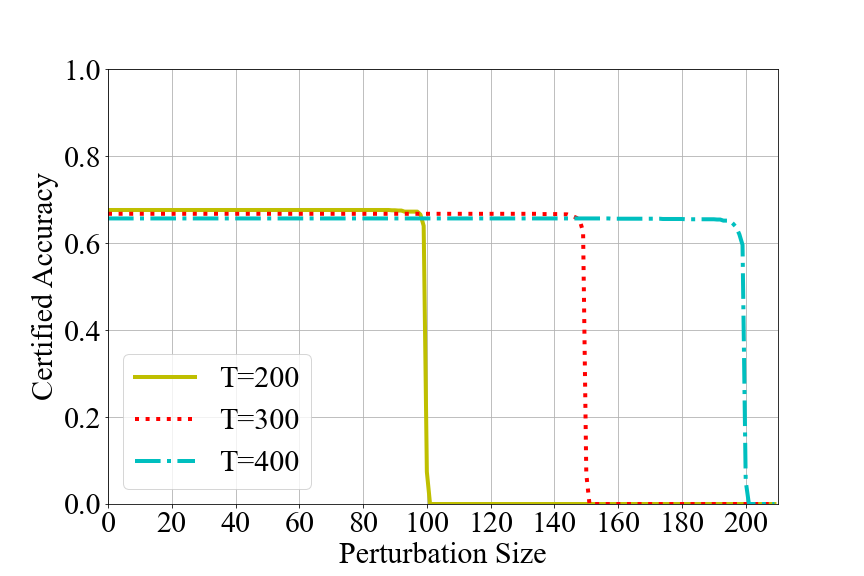}}\hfill
\subfloat[Citeseer]{\includegraphics[width=0.25\textwidth]{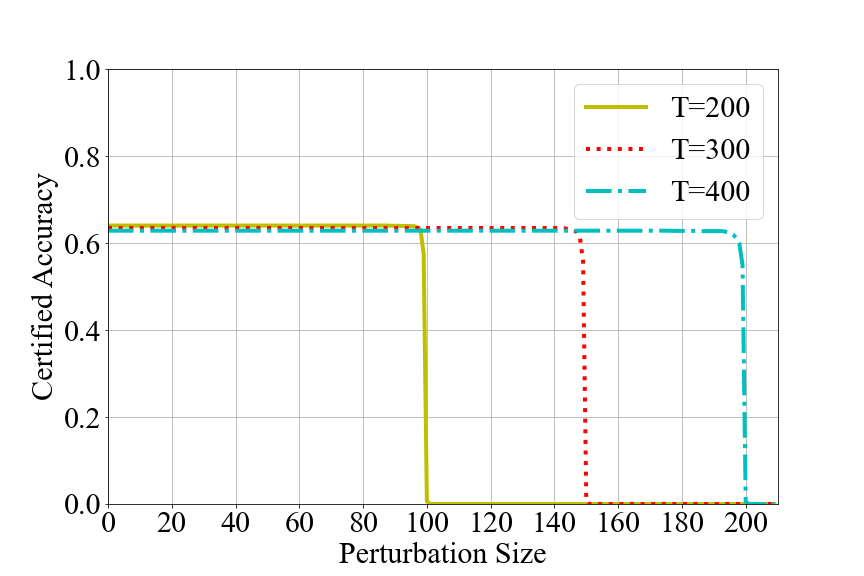}}\hfill
\subfloat[Pubmed]{\includegraphics[width=0.25\textwidth]{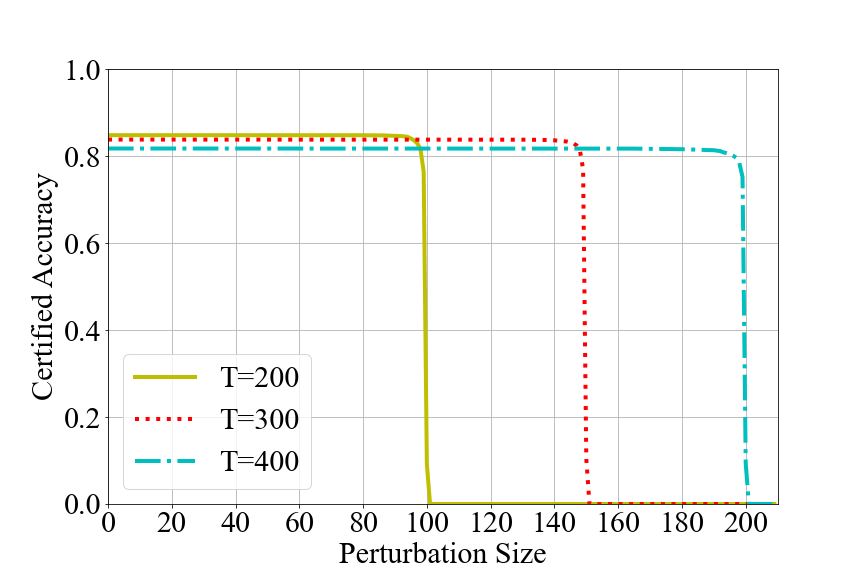}}\hfill
\subfloat[Amazon-C]{\includegraphics[width=0.25\textwidth]{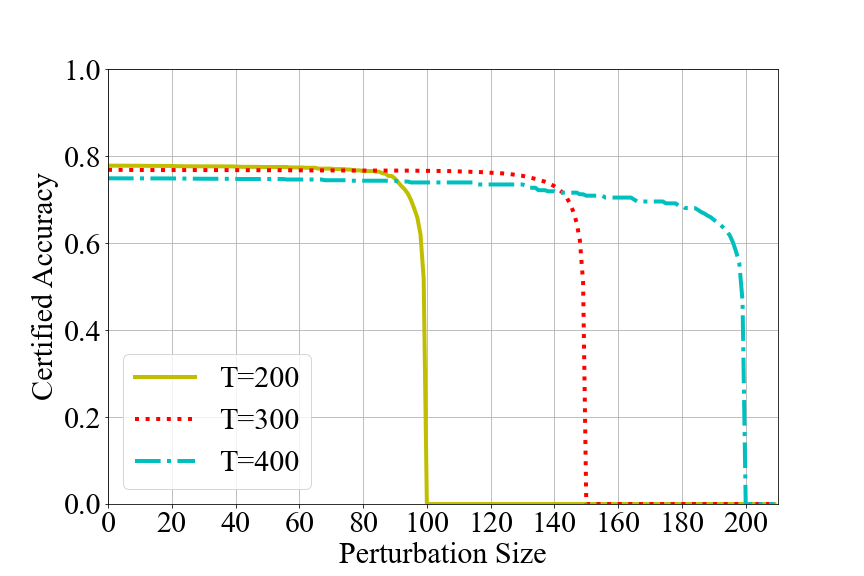}}\\
\caption{Certified node accuracy of our {\nameE} with GSAGE w.r.t. the number of subgraphs $T$.}
\label{fig:node-EC-T-GSAGE}
\vspace{-6mm}
\end{figure*}

\begin{figure*}[!t]
\centering
\subfloat[Cora-ML]{\includegraphics[width=0.25\textwidth]{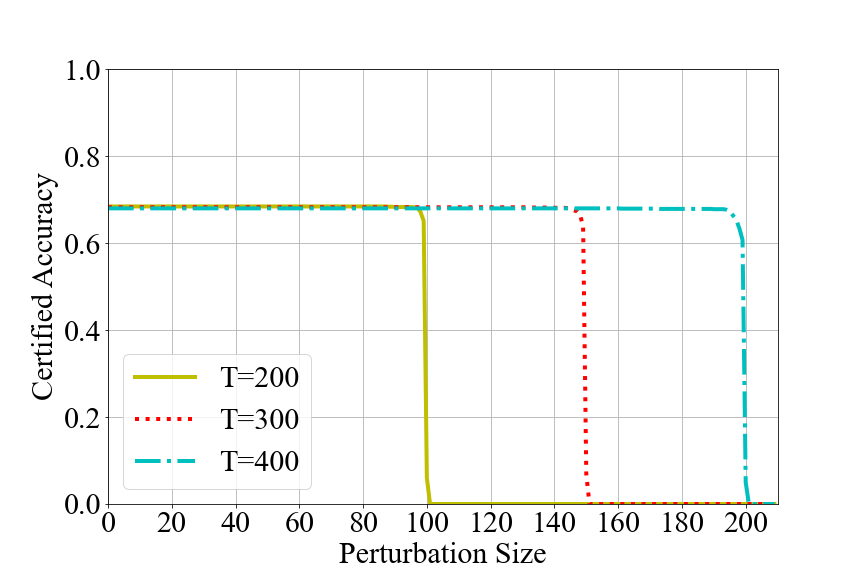}}\hfill
\subfloat[Citeseer]{\includegraphics[width=0.25\textwidth]{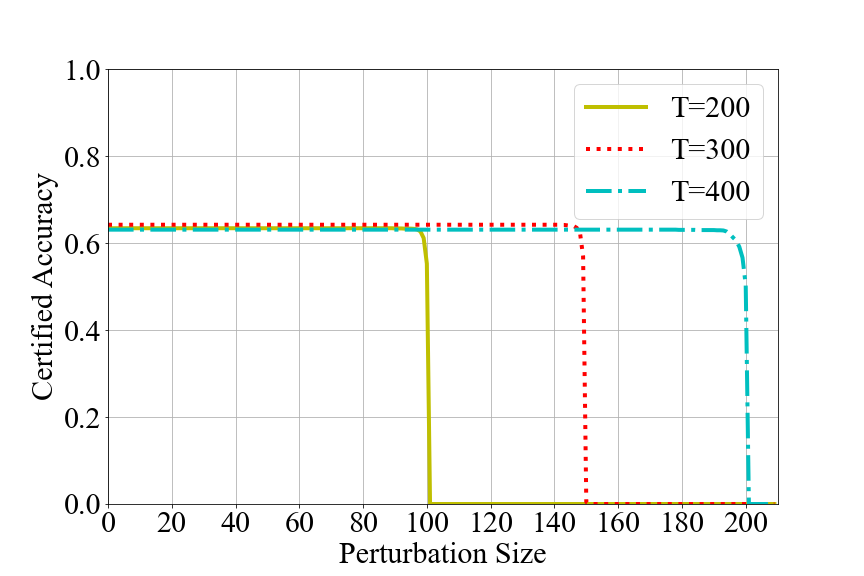}}\hfill
\subfloat[Pubmed]{\includegraphics[width=0.25\textwidth]{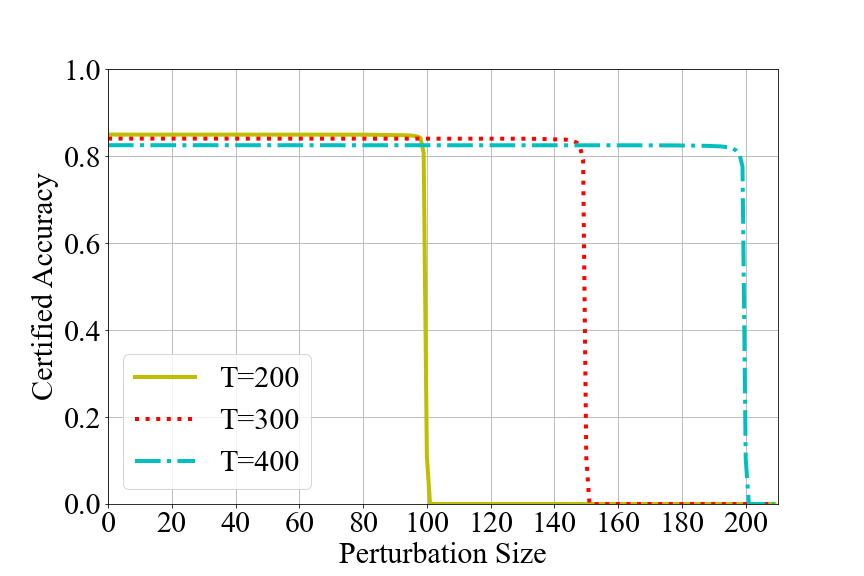}}\hfill
\subfloat[Amazon-C]{\includegraphics[width=0.25\textwidth]{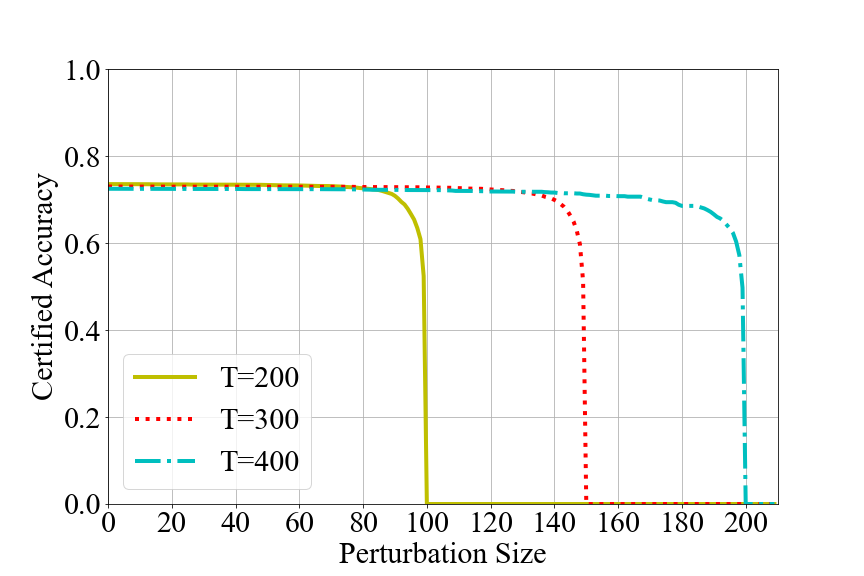}}\\
\caption{Certified node accuracy of our {\nameN} with GSAGE w.r.t. the number of subgraphs $T$.}
\label{fig:node-NC-T-GSAGE}
\vspace{-4mm}
\end{figure*}

\begin{figure*}[!t]
\centering
\subfloat[AIDS]{\includegraphics[width=0.25\textwidth]{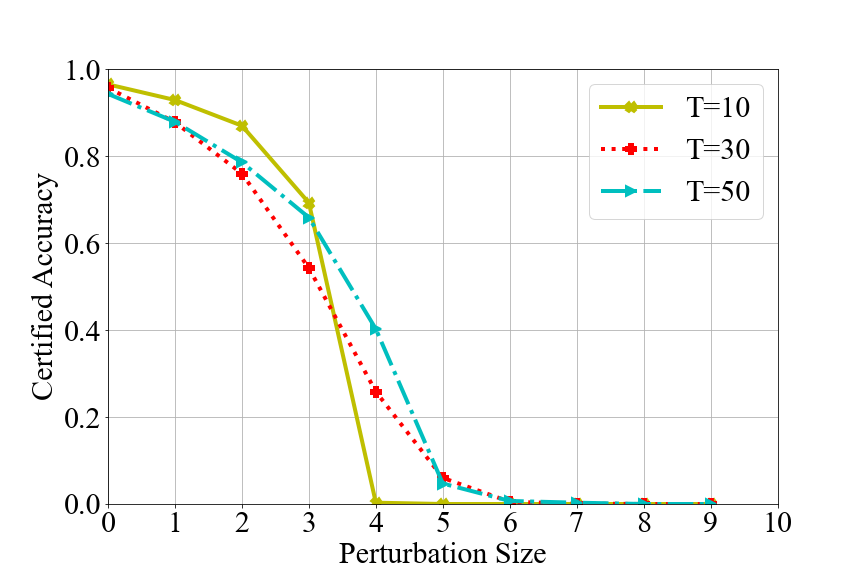}}\hfill
\subfloat[MUTAG]{\includegraphics[width=0.25\textwidth]{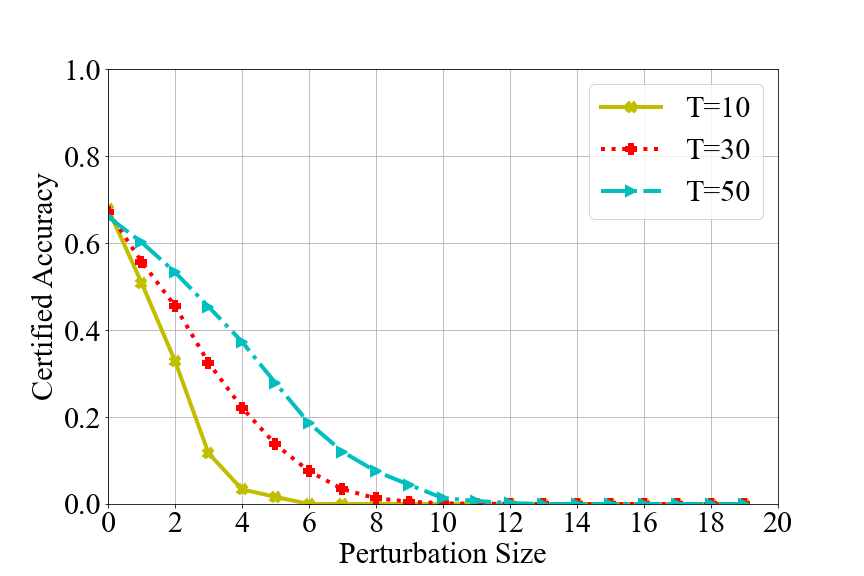}}\hfill
\subfloat[PROTEINS]{\includegraphics[width=0.25\textwidth]{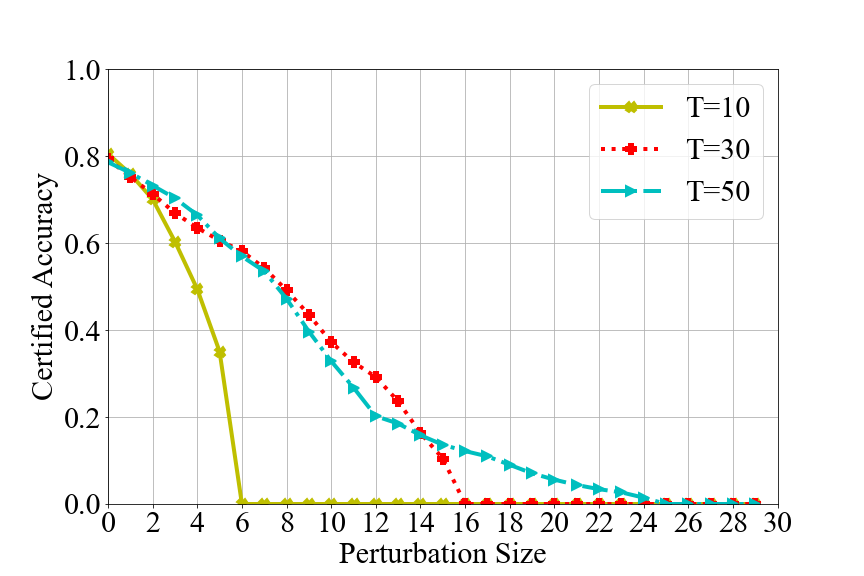}}\hfill
\subfloat[DD]{\includegraphics[width=0.25\textwidth]{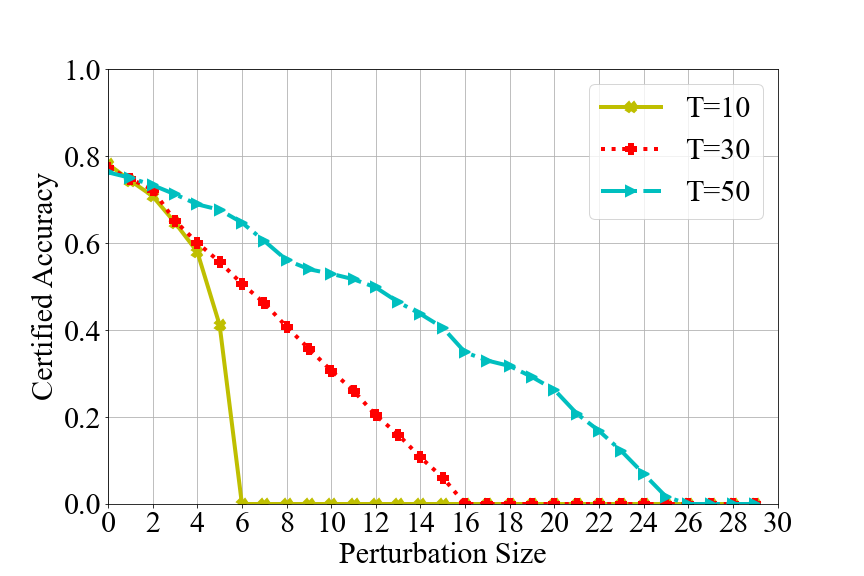}}\\
\caption{Certified graph accuracy of our {\nameE} with GSAGE w.r.t. the number of subgraphs $T$.}
\label{fig:graph-EC-T-GSAGE}
\vspace{-6mm}
\end{figure*}

\begin{figure*}[!t]
\centering
\subfloat[AIDS]{\includegraphics[width=0.25\textwidth]{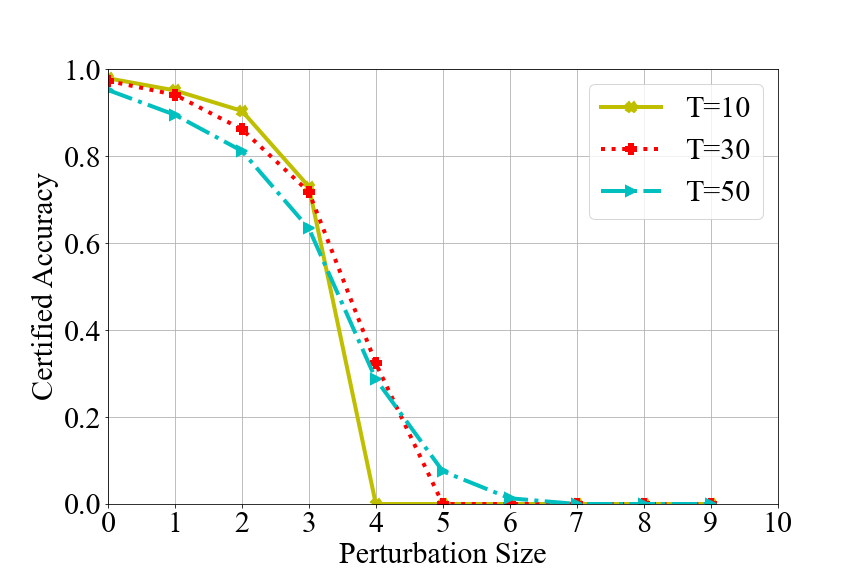}}\hfill
\subfloat[MUTAG]{\includegraphics[width=0.25\textwidth]{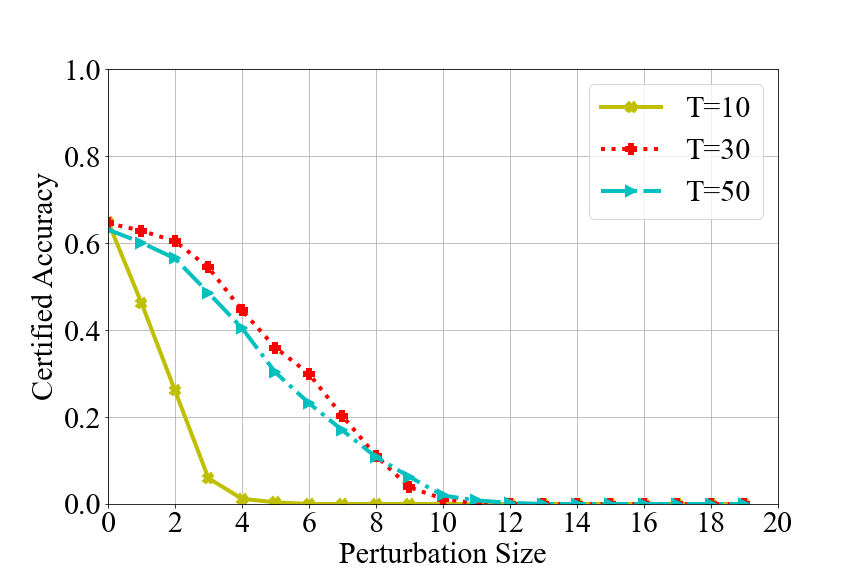}}\hfill
\subfloat[PROTEINS]{\includegraphics[width=0.25\textwidth]{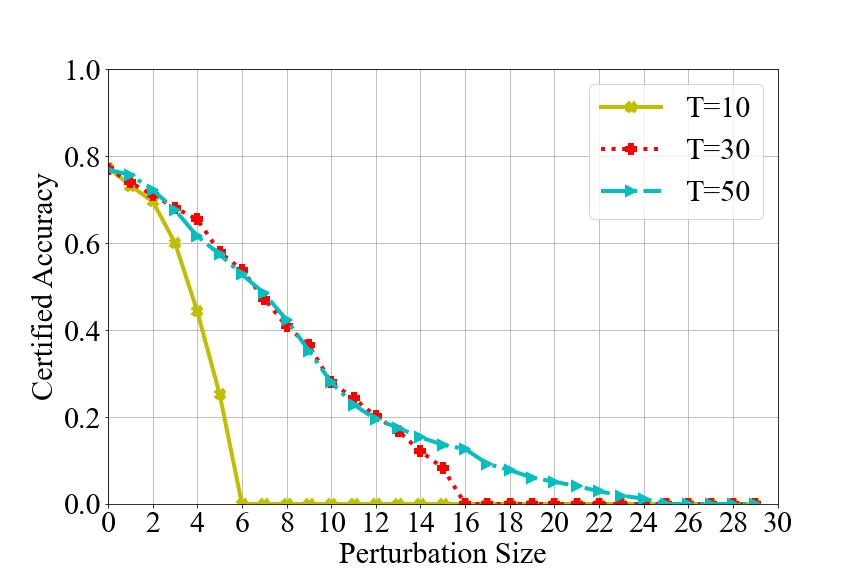}}\hfill
\subfloat[DD]{\includegraphics[width=0.25\textwidth]{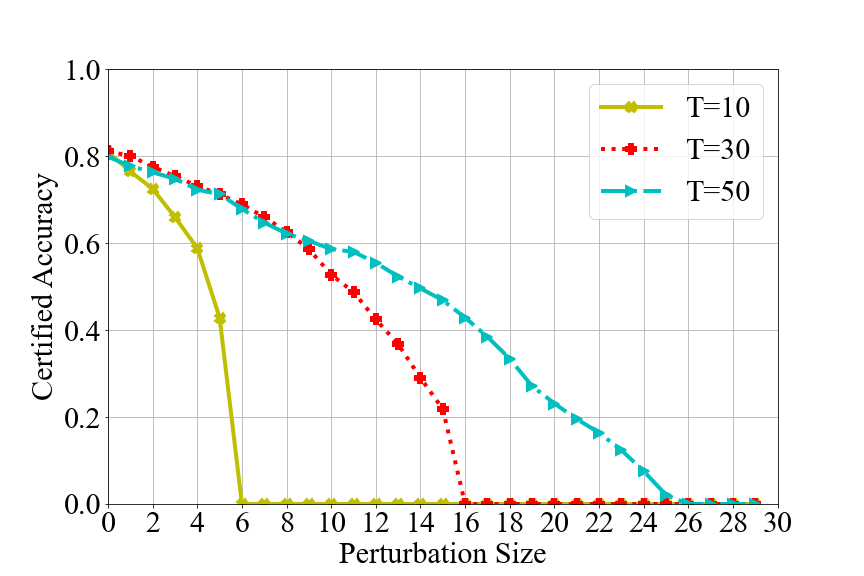}}\\
\caption{Certified graph accuracy of our {\nameN} with GSAGE w.r.t. the number of subgraphs $T$.}
\label{fig:graph-NC-T-GSAGE}
\vspace{-4mm}
\end{figure*}

\begin{figure*}[!t]
\centering
\subfloat[Cora-ML]{\includegraphics[width=0.25\textwidth]{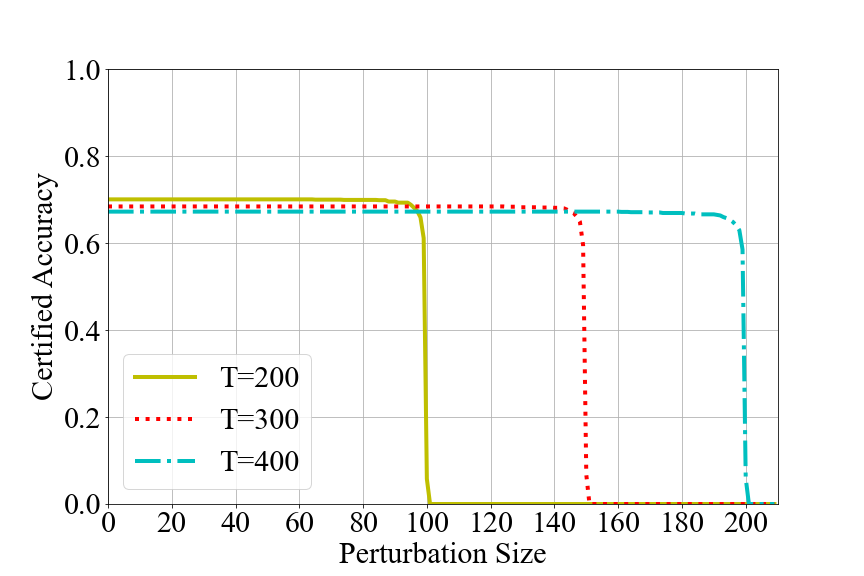}}\hfill
\subfloat[Citeseer]{\includegraphics[width=0.25\textwidth]{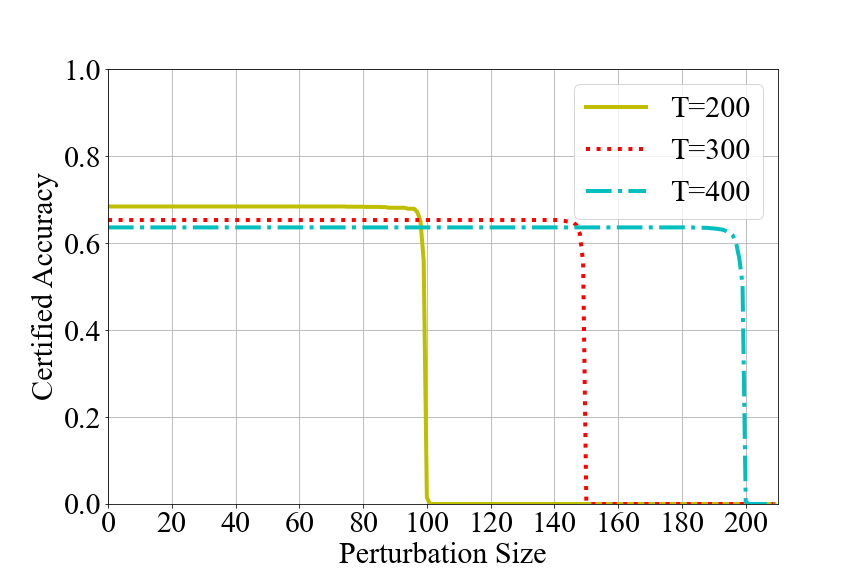}}\hfill
\subfloat[Pubmed]{\includegraphics[width=0.25\textwidth]{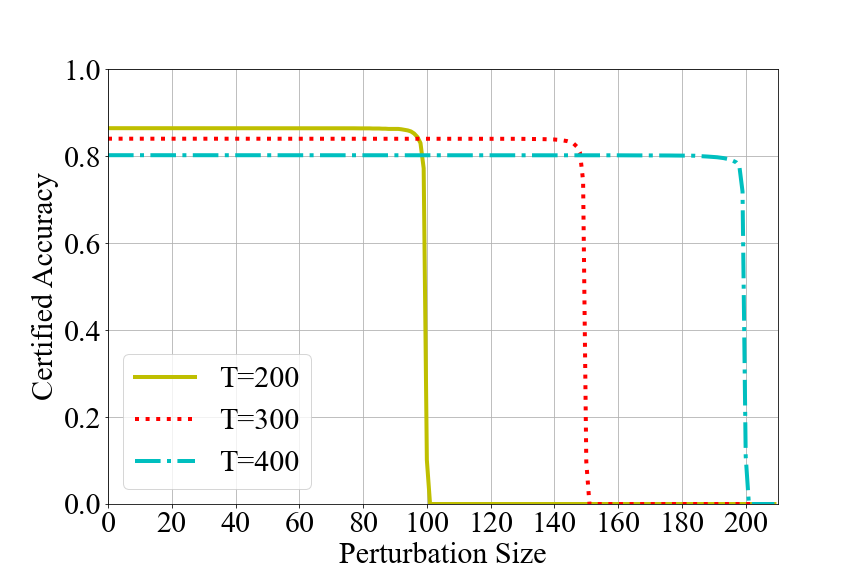}}\hfill
\subfloat[Amazon-C]{\includegraphics[width=0.25\textwidth]{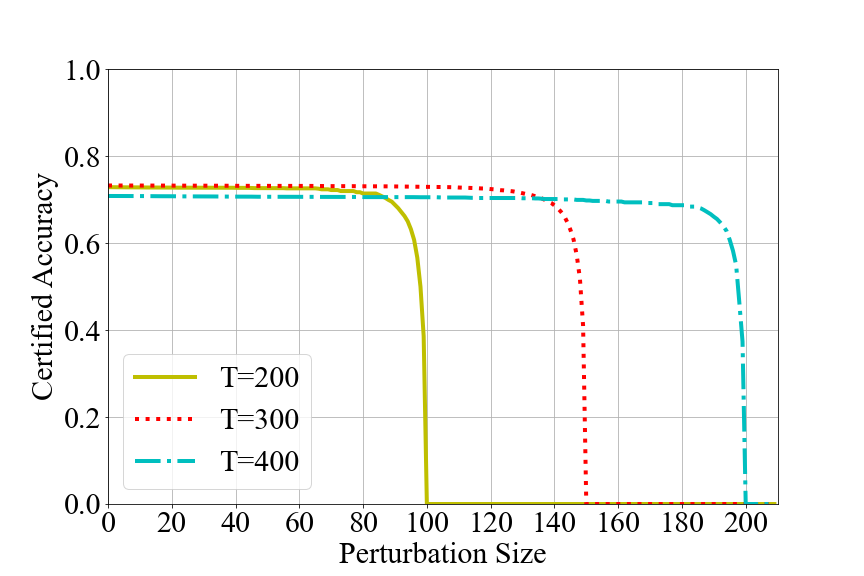}}\\
\caption{Certified node accuracy of our {\nameE} with GAT w.r.t. the number of subgraphs $T$.}
\label{fig:node-EC-T-GAT}
\vspace{-6mm}
\end{figure*}

\begin{figure*}[!t]
\centering
\subfloat[Cora-ML]{\includegraphics[width=0.25\textwidth]{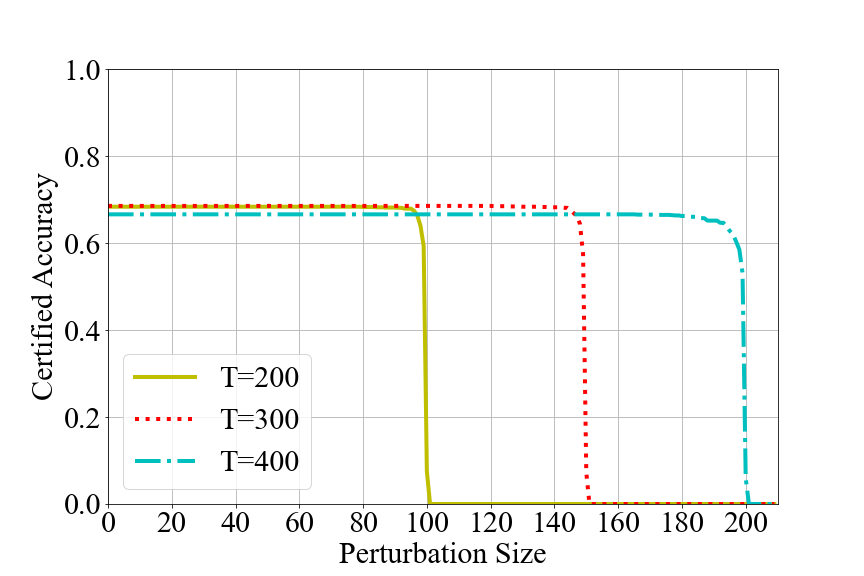}}\hfill
\subfloat[Citeseer]{\includegraphics[width=0.25\textwidth]{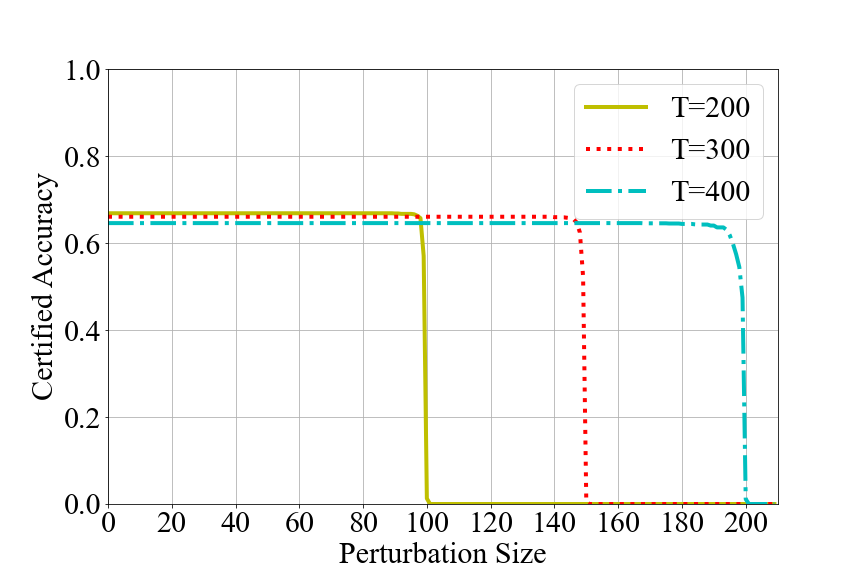}}\hfill
\subfloat[Pubmed]{\includegraphics[width=0.25\textwidth]{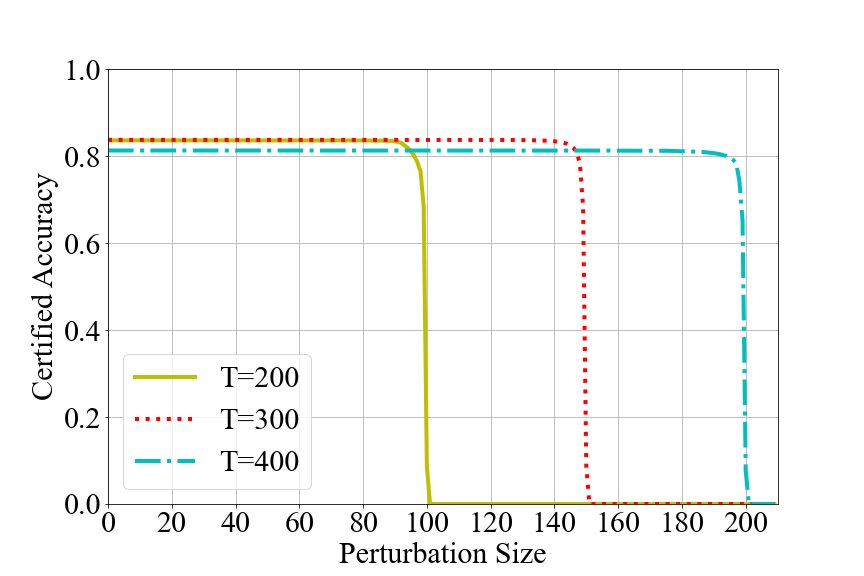}}\hfill
\subfloat[Amazon-C]{\includegraphics[width=0.25\textwidth]{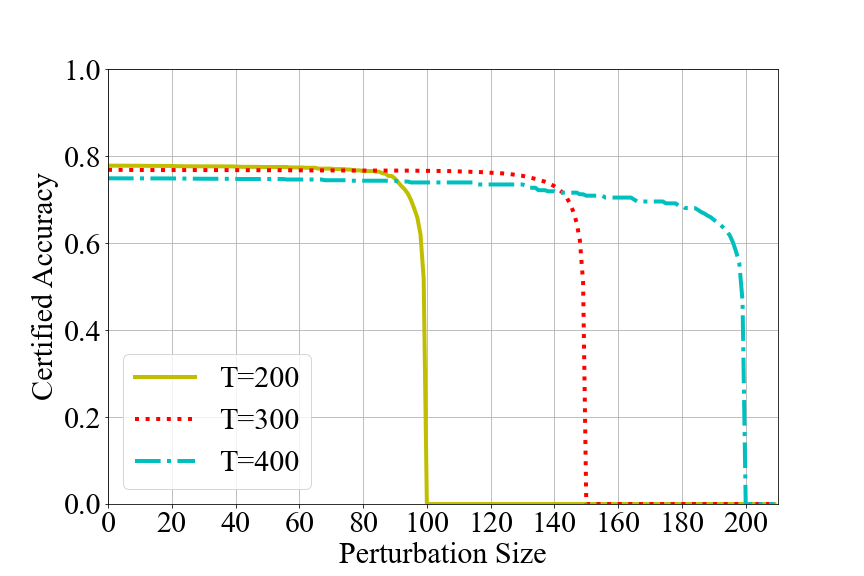}}\\
\caption{Certified node accuracy of our {\nameN} with GAT w.r.t. the number of subgraphs $T$.}
\label{fig:node-NC-T-GAT}
\vspace{-4mm}
\end{figure*}

\begin{figure*}[!t]
\centering
\subfloat[AIDS]{\includegraphics[width=0.25\textwidth]{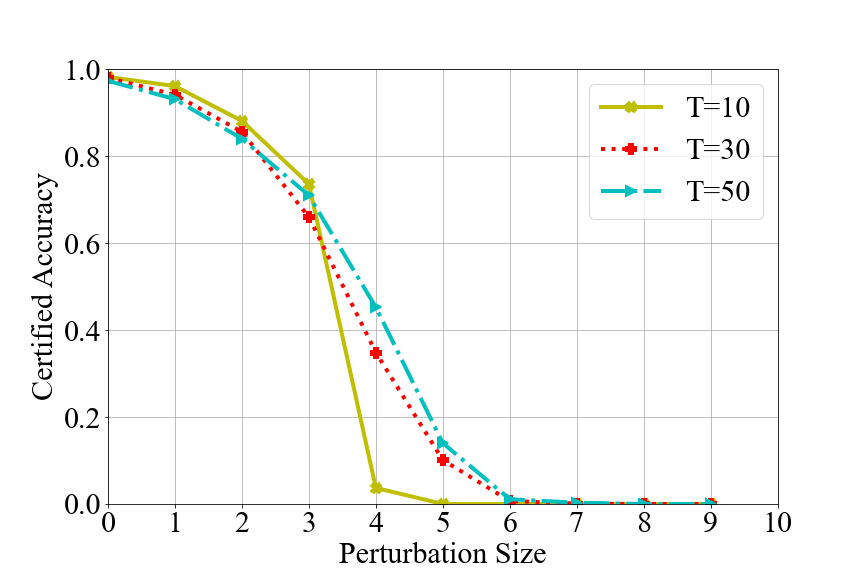}}\hfill
\subfloat[MUTAG]{\includegraphics[width=0.25\textwidth]{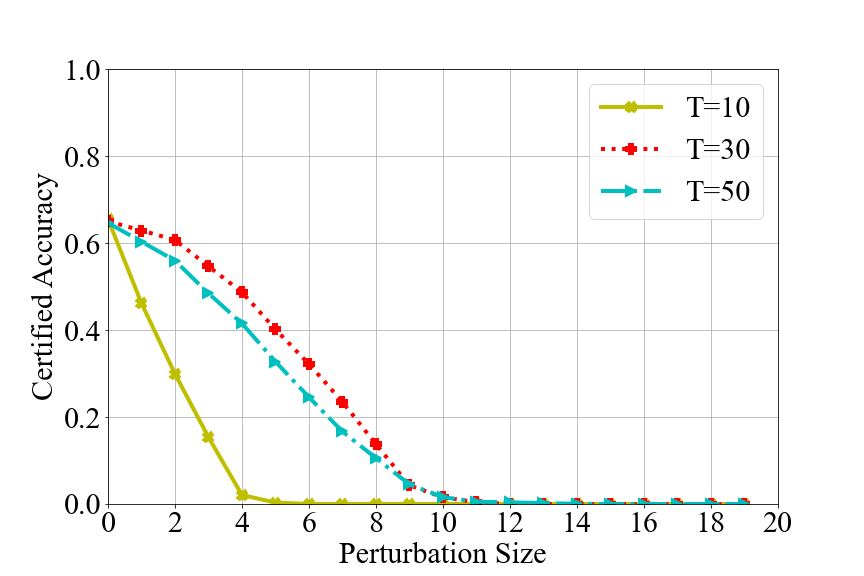}}\hfill
\subfloat[PROTEINS]{\includegraphics[width=0.25\textwidth]{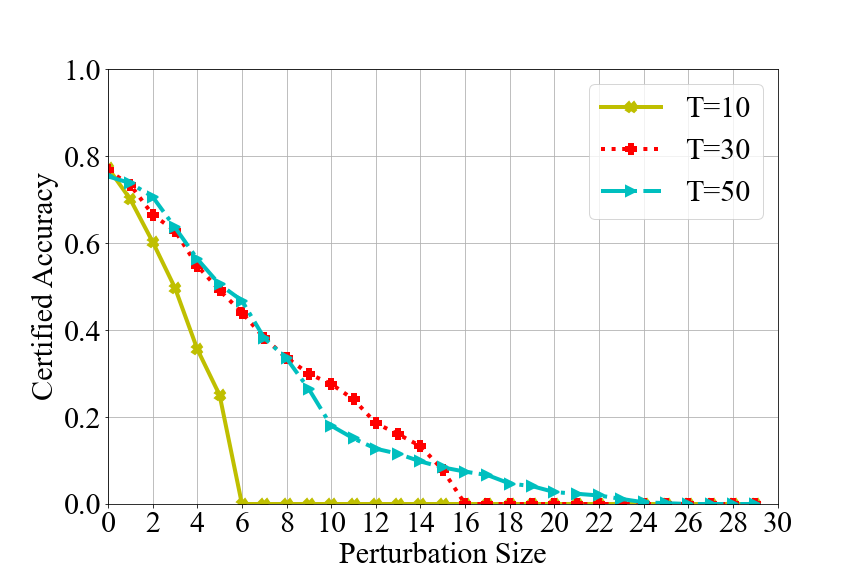}}\hfill
\subfloat[DD]{\includegraphics[width=0.25\textwidth]{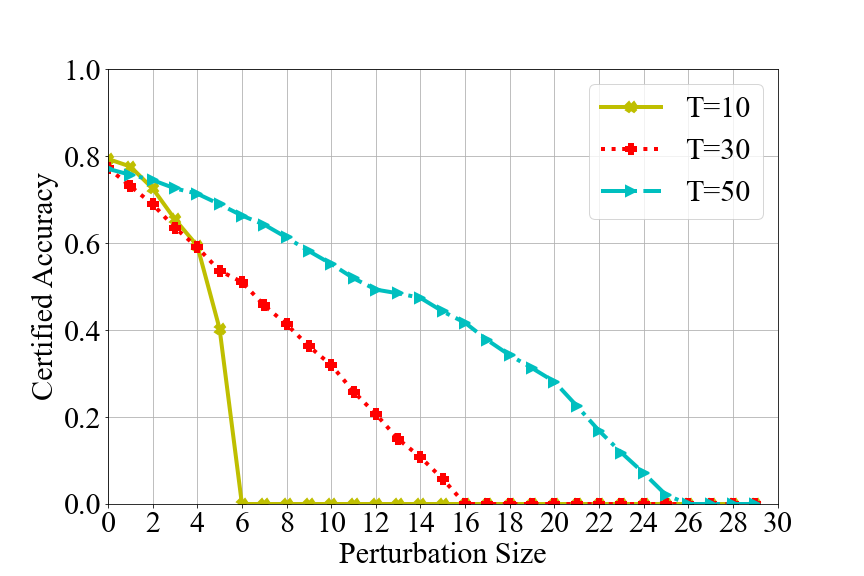}}
\\
\caption{Certified graph accuracy of our {\nameE} with GAT w.r.t. the number of subgraphs $T$.}
\label{fig:graph-EC-T-GAT}
\vspace{-4mm}
\end{figure*}

\begin{figure*}[!t]
\centering
\subfloat[AIDS]{\includegraphics[width=0.25\textwidth]{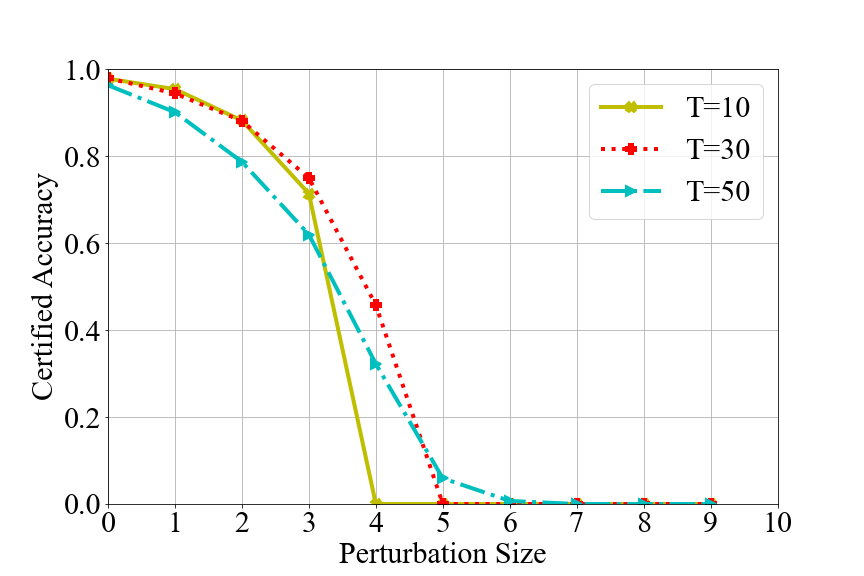}}\hfill
\subfloat[MUTAG]{\includegraphics[width=0.25\textwidth]{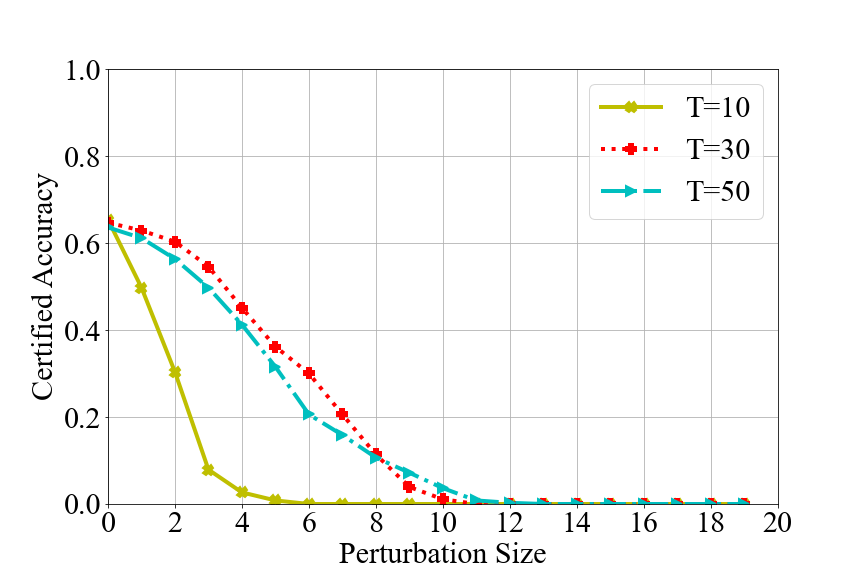}}\hfill
\subfloat[PROTEINS]{\includegraphics[width=0.25\textwidth]{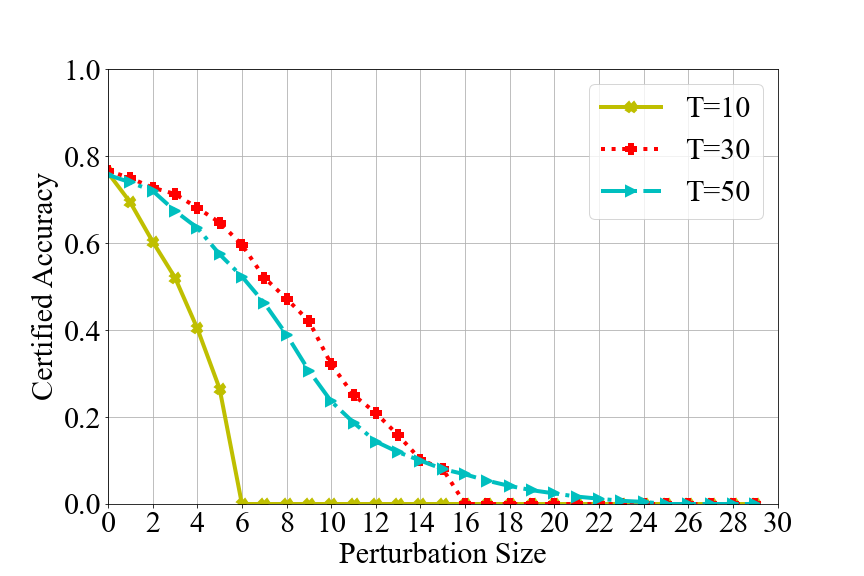}}\hfill
\subfloat[DD]{\includegraphics[width=0.25\textwidth]{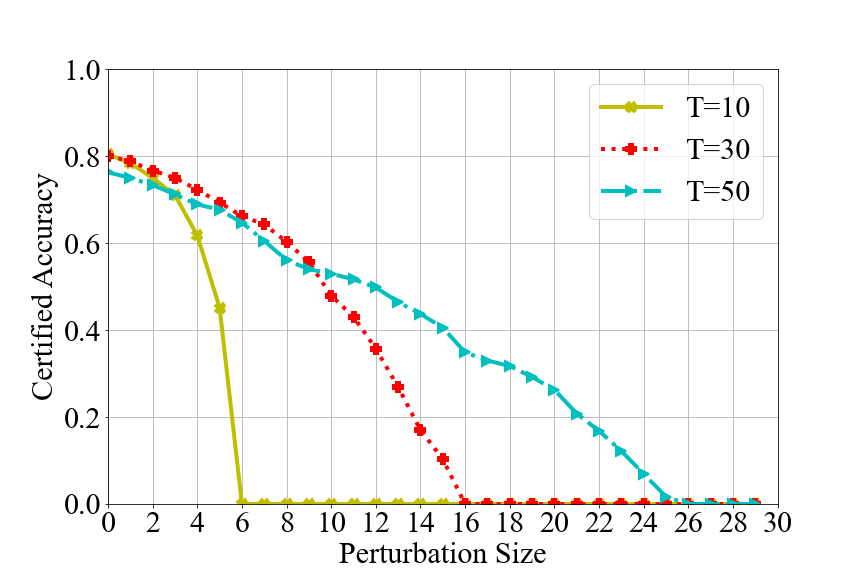}}\\
\caption{Certified graph accuracy of our {\nameN} with GAT w.r.t. the number of subgraphs $T$.}
\label{fig:graph-NC-T-GAT}
\vspace{-4mm}
\end{figure*}

\begin{figure*}[!t]
\centering
\subfloat[Cora-ML]{\includegraphics[width=0.25\textwidth]{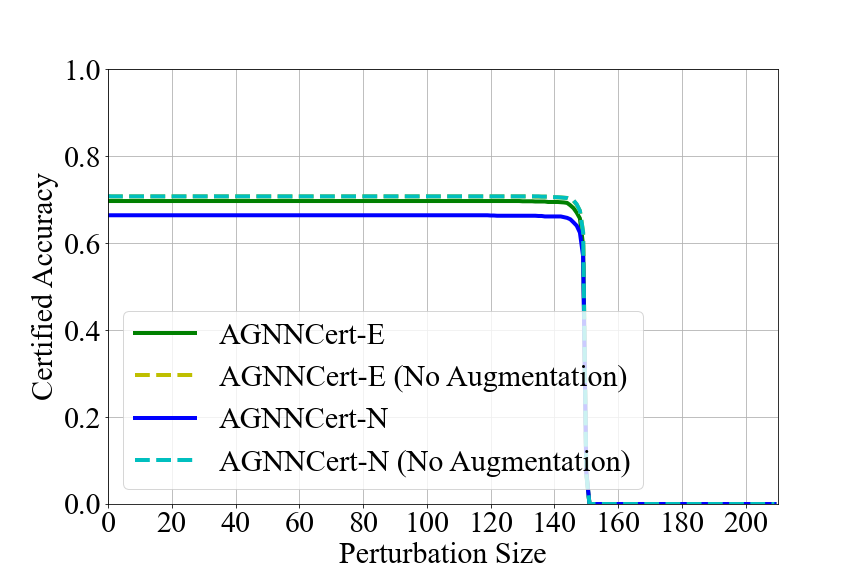}}\hfill
\subfloat[Citeseer]{\includegraphics[width=0.25\textwidth]{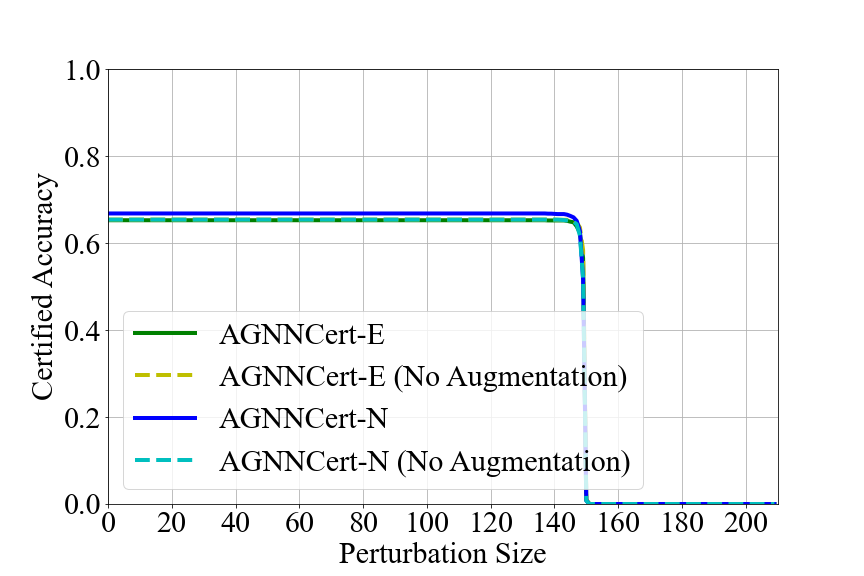}}\hfill
\subfloat[Pubmed]{\includegraphics[width=0.25\textwidth]{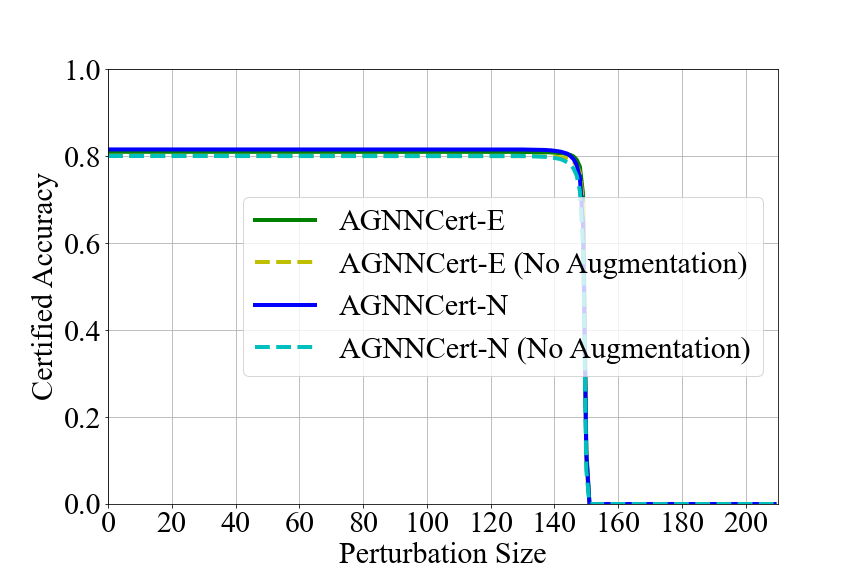}}\hfill
\subfloat[Amazon-C]{\includegraphics[width=0.25\textwidth]{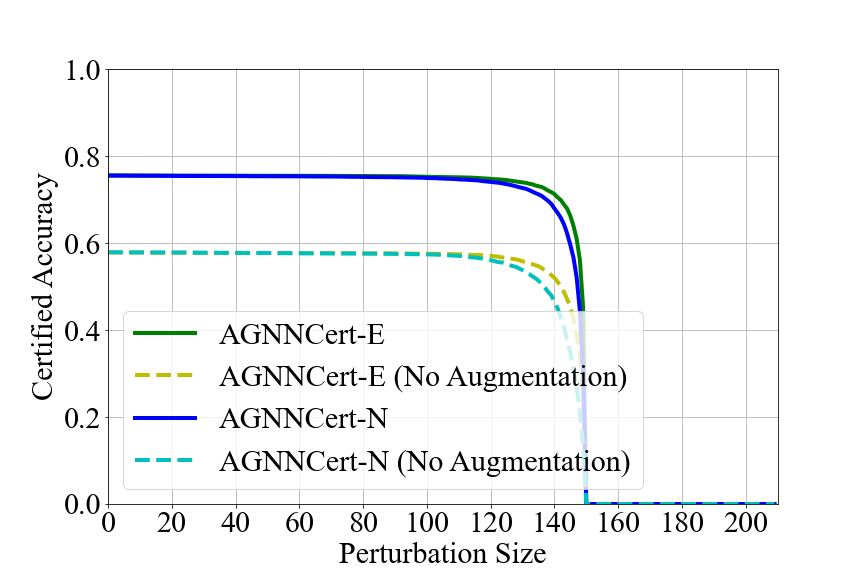}}\\
\caption{Certified node accuracy of our {\name} with and without subgraphs for training under the default setting.}
\label{fig:node-EC-w-wo}
\vspace{-4mm}
\end{figure*}

\begin{figure*}[!t]
\centering
\subfloat[AIDS]{\includegraphics[width=0.25\textwidth]{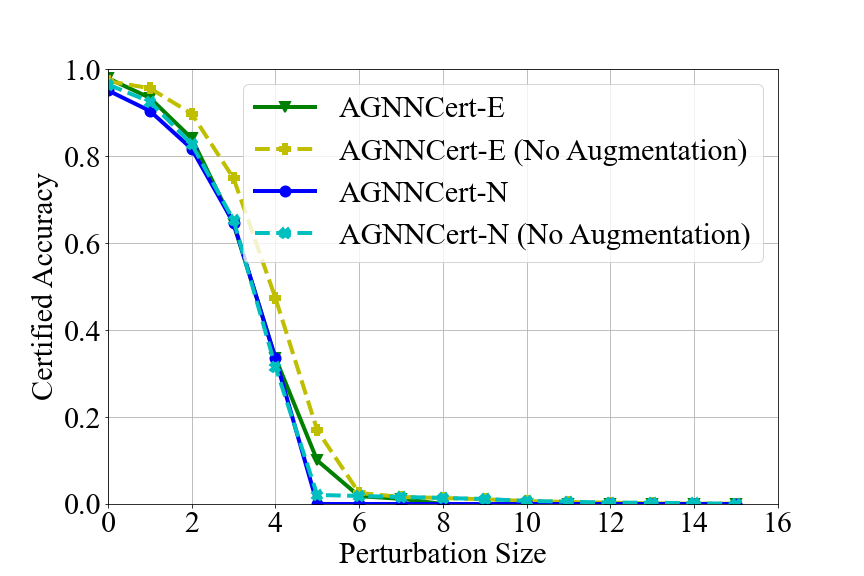}}\hfill
\subfloat[MUTAG]{\includegraphics[width=0.25\textwidth]{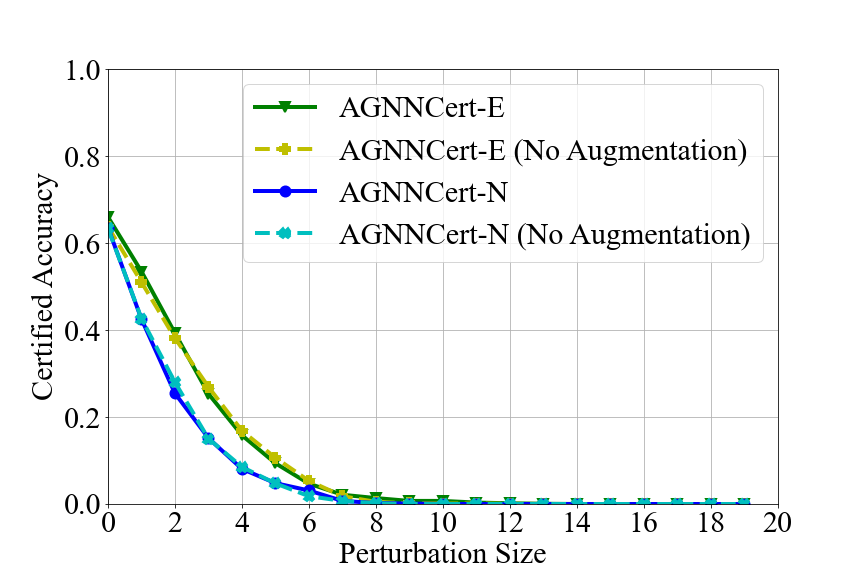}}\hfill
\subfloat[PROTEINS]{\includegraphics[width=0.25\textwidth]{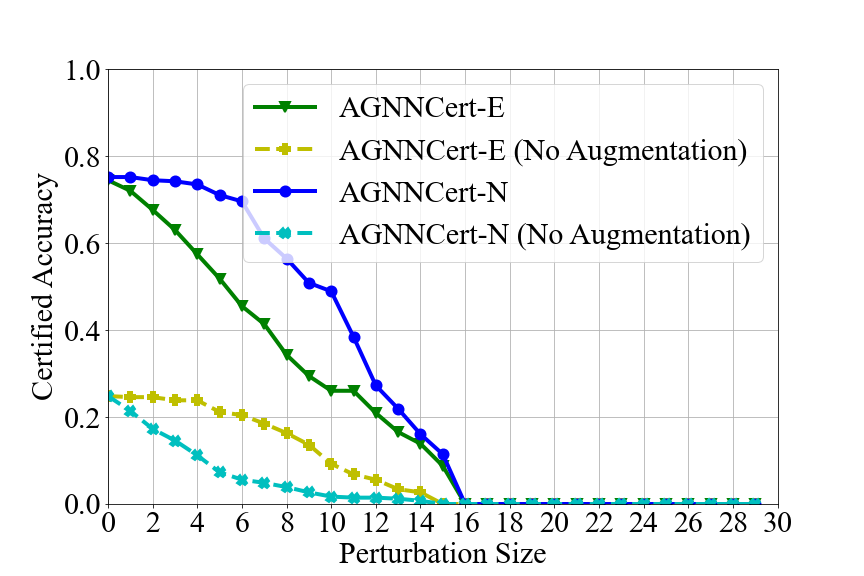}}\hfill
\subfloat[DD]{\includegraphics[width=0.25\textwidth]{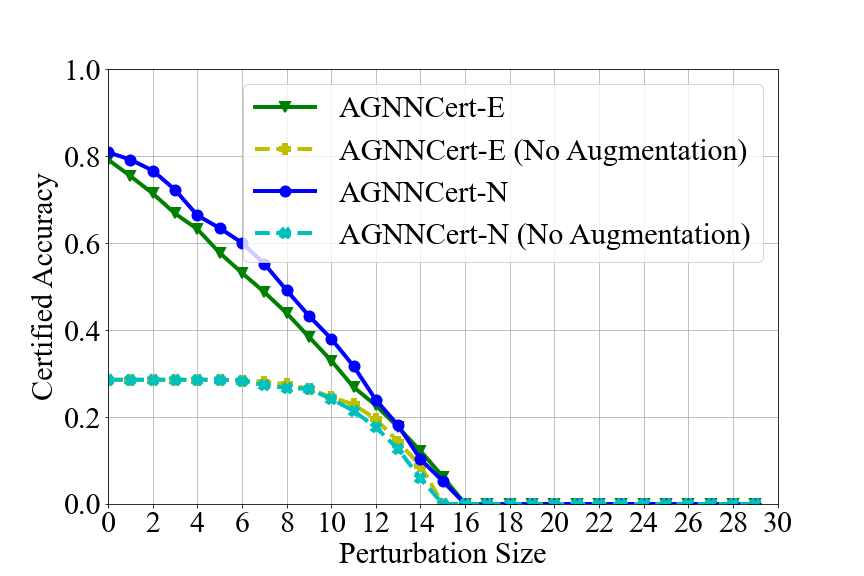}}\\
\caption{Certified graph accuracy of our {\name} with and without subgraphs  for training under the default setting.}
\label{fig:graph-EC-w-wo}
\end{figure*}

\end{document}